\documentclass[aps,10pt,twocolumn,showpacs,pra,citeautoscript,amsmath,amssymb,floatfix,nofootinbib,superscriptaddress,longbibliography]{revtex4-1}

\usepackage{amsmath}
\usepackage{amssymb}
\usepackage{bbm}
\usepackage{amsthm, thm-restate}

\usepackage{graphicx}
\newcommand{\caphead}[1]{{\bf #1}}
\usepackage{xcolor}

\usepackage{natbib}
\usepackage{hyperref}
\usepackage[capitalise]{cleveref}
\crefname{figure}{Figure}{figures}

\usepackage{algpseudocode}
\usepackage{enumerate}

\newcounter{lemmaN}

\newtheorem{lemma}[lemmaN]{Lemma}

\newcounter{lemmaA}

\newtheorem{applemma}[lemmaA]{Lemma}

\newcounter{colN}

\newtheorem{corollary}[colN]{Corollary}
\newtheorem{definition}{Definition}
\newtheorem{theorem}{Theorem}
\newtheorem{example}{Example}

\newtheorem*{assumption}{Assumption}

\newcommand{\ket}[1]{|#1\rangle}

\newcommand{\ketbra}[2]{|#1\rangle\langle #2|}
\DeclareMathOperator{\tr}{tr}
\newcommand{\id}{\mathbbm{1}}

\newcommand{\app}[2]{\left( #1, #2 \right)}
\newcommand{\inn}[2]{\langle #1, #2 \rangle}

\newcommand{\reals}{\mathbb{R}}
\newcommand{\comp}{\mathbb{C}}
\newcommand{\quat}{\mathbb{H}}

\newcommand{\trans}[0]{^{\rm T}}
\newcommand{\ct}[0]{^{\dag}}
\newcommand{\conj}[0]{^{*}}

\newcommand{\Rsap}[1]{\mathbf{H}_{#1}^{+}(\mathbb{R})}
\newcommand{\Csa}[1]{\mathbf{H}_{#1}(\mathbb{C})}
\newcommand{\Csap}[1]{\mathbf{H}_{#1}^{+}(\mathbb{C})}
\newcommand{\Hsa}[1]{\mathbf{H}_{#1}(\mathbb{H})}
\newcommand{\Hsap}[1]{\mathbf{H}_{#1}^{+}(\mathbb{H})}

\newcommand{\jord}[0]{\bullet}

\newcommand{\sjord}[0]{\bullet}

\newcommand{\rehead}[0]{\mathrm{Re}}
\newcommand{\imhead}[0]{\mathrm{Im}}
\newcommand{\jmhead}[0]{\mathrm{Jm}}
\newcommand{\kmhead}[0]{\mathrm{Km}}
\newcommand{\re}[1]{\rehead\!\left(#1\right)}
\newcommand{\im}[1]{\imhead\!\left(#1\right)}
\newcommand{\jm}[1]{\jmhead\!\left(#1\right)}
\newcommand{\km}[1]{\kmhead\!\left(#1\right)}

\newcommand{\inlineheading}[1]{\textbf{{#1}}}

\makeatletter
\newcommand*{\balancecolsandclearpage}{%
  \close@column@grid
  \clearpage
  \twocolumngrid
}
\makeatother

\makeatletter
\def\tocdepth@fullmunge{%
\let\l@section@saved\l@section
\let\l@section\@gobble@tw@
\let\l@subsection@saved\l@subsection
\let\l@subsection\@gobble@tw@
}%
\def\tocdepth@fullrestore{%
\let\l@section\l@section@saved
\let\l@subsection\l@subsection@saved
}%
\newcommand{\hidetoc}[0]{\addtocontents{toc}{\string\tocdepth@fullmunge}}
\newcommand{\restoretoc}[0]{\addtocontents{toc}{\string\tocdepth@fullrestore}}
\makeatother

\newcommand{\IQOQI}{Institute for Quantum Optics and Quantum Information,\\ Austrian Academy of Sciences, Boltzmanngasse 3, A-1090 Vienna, Austria}
\newcommand{\Peri}{Perimeter Institute for Theoretical Physics, 31 Caroline Street North, Waterloo, ON N2L 2Y5, Canada}
\newcommand{\VCQ}{Vienna Center for Quantum Science and Technology (VCQ), Faculty of Physics,\\ University of Vienna, Boltzmanngasse 5, A-1090 Vienna, Austria}

\begin{document}

\title{Testing quantum theory by generalizing noncontextuality}

\author{Markus P.\ M\"uller}
\affiliation{\IQOQI{}}
\affiliation{\VCQ{}}
\affiliation{\Peri{}}
\author{Andrew J.\ P.\ Garner}
\affiliation{\IQOQI{}}

\date{October 2, 2023}

\begin{abstract}
It is a fundamental prediction of quantum theory that states of physical systems are described by complex vectors or density operators on a Hilbert space. 
However, many experiments admit effective descriptions in terms of other state spaces, such as classical probability distributions or quantum systems with superselection rules. 
Which kind of effective statistics would allow us to experimentally falsify quantum theory as a fundamental description of nature? 
Here, we address this question by introducing a methodological principle that generalizes Spekkens' notion of noncontextuality: processes that are statistically indistinguishable in an effective theory should not require explanation by multiple distinguishable processes in a more fundamental theory. 
We formulate this principle in terms of linear embeddings and simulations of one probabilistic theory by another, show how this concept subsumes standard notions of contextuality, and prove a multitude of fundamental results on the exact and approximate embedding of theories (in particular into quantum theory). 
We prove that only Jordan-algebraic state spaces are exactly embeddable into quantum theory, and show how results on Bell inequalities can be used for the certification of non-approximate embeddability.
From this, we propose an experimental test of quantum theory by probing single physical systems without assuming access to a tomographically complete set of procedures or calibration of the devices, arguably avoiding a significant loophole of earlier approaches. 
\end{abstract}

\maketitle

\tableofcontents

\section{Introduction}
Quantum theory is the basis of almost all of modern physics. 
Its prescription to represent states of physical systems by complex vectors and measurements by projectors (or, more generally, density operators and elements of a POVM, i.e.\ of a positive operator-valued measure) is the template that underlies our very notion of the construction of a physical theory, including quantum mechanics, the standard model of particle physics, and quantum theories of gravity.

However, in many regimes, these fundamental quantum structures cannot be seen directly. 
Many physical phenomena have instead \emph{effective descriptions} with representations of states and measurements that differ \emph{a priori} from the quantum postulates. 
For example, decoherence processes~\cite{Zurek03} lead to quantum systems with noise or superselection rules, including systems exactly described by classical probability theory. 
Similarly, representations of symmetries generate systems that can be described by quantum theory over other fields than the complex numbers, such as the real numbers or the quaternions~\cite{Baez12,BarnumGW20}.

Indeed, we know that quantum theory is just one special case among a plethora of \emph{general probabilistic theories}~\cite{Hardy01,Barrett07,Mueller21,Plavala21}, with differing physical and information-processing properties. 
This raises an important foundational line of questioning: assuming the fundamental validity of standard complex quantum theory, which effective probabilistic theories can we expect to find in nature?
For which theories of this kind is it \emph{possible} to arise from quantum theory, and for which ones is it \emph{plausible}?

These questions are not only interesting in their own right, but have crucial implications for \emph{experimental tests} of quantum theory. 
As is done for other pillars of modern physics such as the equivalence principle~\cite{Will14}, we should arguably also submit quantum theory to precise scrutiny.
This turns out to be difficult.
For example, potential beyond-quantum phenomena like 
 higher-order interference~\cite{Sorkin94,SinhaCJLW10,UdudecBE10,BarnumMU14,LeeS17}
  or quaternionic amplitudes~\cite{Peres79,Adler95,Graydon11,Procopio17} can always be faked by introducing additional degrees of freedom~\cite{SawantSSSS14,BarnumGW20,Namdar21} -- in some sense, it is quantum theory's enormous flexibility that makes it so hard to falsify.

In this article, we give a partial classification of (and many general results on) the probabilistic theories that can plausibly arise -- even approximately -- as effective descriptions from quantum theory, and use this to propose an experimental test of quantum theory that arguably avoids the aforementioned problems to a large extent. 
We do so by generalizing and amending concepts that have recently been studied in a different context and with different goals in quantum information theory.

First, to falsify quantum theory as a fundamental description of nature, it makes sense to draw inspiration from the complementary goal of falsifying a fundamental \emph{classical} description of nature -- that is, of proving the \emph{nonclassicality} of quantum phenomena~\cite{Bell64,ClauserHSH69,BrunnerCPSW14,KochenS1967,CabelloEG96,Cabello08,AraujoQBCC13,CabelloSW14,AcinFLS15,CataniLSS21}.
One well-motivated notion of nonclassicality of exceptional conceptual clarity is \emph{generalized contextuality} as proposed by \citet{Spekkens05}. 
In contrast to earlier proposals like Kochen-Specker's, this notion applies to a wide range of physical phenomena~\cite{SpekkensBKTP08,Pusey14,PuseyL15,SchmidS18,SelbySWSKS21} and can be subjected to direct experimental test~\cite{Kunjwal16,MazurekPKRS16,KrishnaSW17}.  
It can be interpreted as a version of Leibniz' principle of the ``identity of the indiscernibles''~\cite{Spekkens19,MazurekPKRS16}: procedures that are indistinguishable at the operational level should be identical at the ontological level.

We generalize the notion of generalized contextuality even further, and reinterpret it as a methodological principle that relates different levels in a hierarchy of physical description: \emph{statistically indistinguishable processes in an effective theory should not require multiple distinct processes that explain it in a more fundamental theory}. 
We formulate this condition mathematically in terms of \emph{approximate embeddings}, show that it reduces to Spekkens' notion in the case of embeddings into classical theory, and prove a multitude of related results, including a complete characterization of the ``unrestricted'' probabilistic theories (defined below) that have an exact embedding into quantum theory.

Subsequently, we use the resulting insights to propose an experimental test of quantum theory that builds on the recently proposed scheme of \emph{theory--agnostic tomography}~\cite{GraboweckyPCSR21,MazurekPRS21}.
In a nutshell, the idea is to probe a given physical system with as many preparations and measurements as possible, and to fit a probabilistic theory to the data. 
Our proposal is that the results of such experiments should typically admit an approximate linear embedding into quantum theory, even if the experiment is not tomographically complete~\cite{GraboweckyPCSR21,MazurekPKRS16,MazurekPRS21} in the usual sense. 
Demonstrating the opposite would therefore challenge quantum theory, and we give results that allow for the robust certification of this generalization of contextuality.

\textbf{Our article is organized as follows.}
We begin in \cref{sec:Maths} with a brief review of the operational setting of prepare--and--measure experiments,
 and how these can be expressed in the framework of generalized probabilistic theories.
In \cref{sec:Simulation}, we formalize the relationship between effective and fundamental theories as a \emph{simulation},
 introduce a generalization of noncontextuality as a property such simulations may have (univalence),
 and explore the mathematical structure that such univalent simulations must exhibit.
Then, in \cref{sec:StandardContextuality}, we show that our notion of univalence reduces to Spekkens' notion of generalized noncontextuality in the special case where the fundamental theory is classical.
Next, in \cref{sec:QuantumExact}, we categorize the unrestricted probabilistic theories that have an exact univalent quantum simulation.
Then, in \cref{sec:QuantumApprox}, by adapting results from nonlocality studies, we formulate a bound on the accuracy of univalently simulating effective theories even approximately by quantum theory.
Finally, in \cref{sec:Experiment}, we outline how this bound can be applied to experimentally test quantum theory.

\section{Framework}
\label{sec:Maths}
\subsection{Prepare-and-measure experiments}
\label{sec:PrepAndMeasure}

\vskip -1.5em
\begin{figure}[tbh]
\begin{centering}
\includegraphics[width=0.37\textwidth]{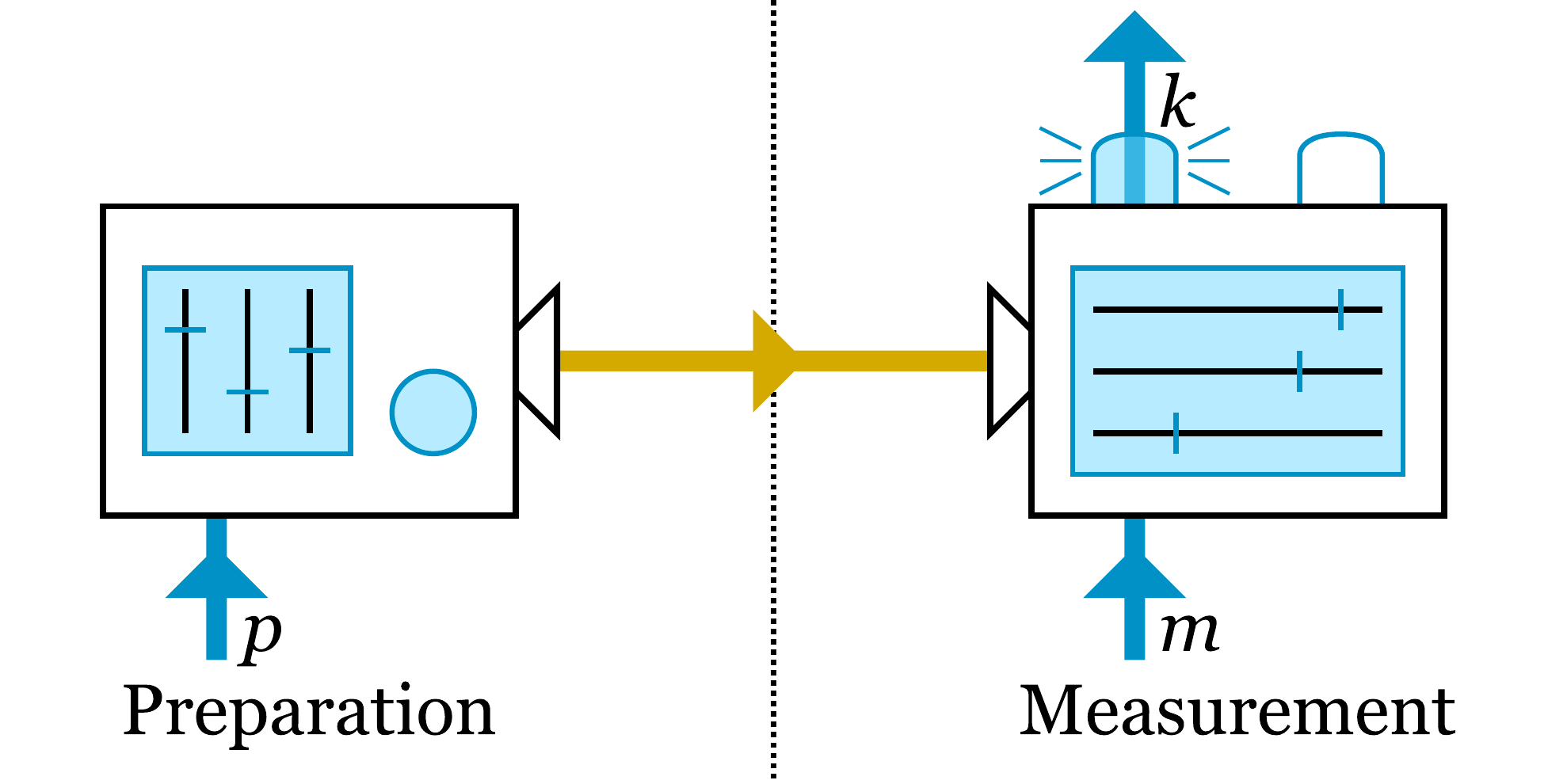}
\caption{%
\label{fig:prep_measure}
\caphead{Prepare-and-measure scheme.}
In this article, we consider a class of experimental scenarios consisting of a preparation procedure ($p$) followed by a measurement procedure ($m$) resulting, perhaps probabilistically, in an outcome ($k$).
A statistical description of such experiments is given by the conditional probability distributions $P(k|p,m)$.
}
\end{centering}
\end{figure}

In this article, we consider a class of simple two-step experimental scenarios that consist of a preparation procedure followed by a measurement procedure.
An \emph{operational theory}~\cite{Spekkens05} for such a scenario assigns some probability $P(k|p,m)$ to each conceivable outcome $k$ of measurement $m$ taken on the output of preparation procedure $p$ (see \cref{fig:prep_measure}).

What could constitute an experimental procedure can be quite broad. 
For example, in a cold atom experiment, an experimentalist has the choice of how long to apply a microwave pulse for,
 as well as a choice about the colour of socks to wear while performing the experiment,
 but the former is much more likely to be of interest than the latter.
As such, if we wish to describe the behaviour of a particular experimental scenario as opposed to that of the entire universe,
  we need to select a set of preparation and measurement procedures that we are interested in.
This in turn defines the domain of an operational theory that is supposed to describe the scenario. 
  
In general, a scenario (and its operational theory) may comprise a large collection of experimental procedures, together with ways to transform and compose them in succession or in parallel; and possibly only certain types of measurements may be applied to certain types of preparations. 
Here, however, we restrict ourselves to prepare-and-measure scenarios where all measurements $m$ can follow all preparations $p$. 
In this case, the selection of the set of procedures of interest can be interpreted as defining the \emph{physical system}: intuitively, this is the thing (e.g.\ particle) that is produced by preparations and probed by measurements.

Suppose that by defining the boundaries of our experimental scenario, we chose to ignore the experimentalist's color of socks, but included the paint color of the preparation boxes in our description. This attribute may still turn out to be irrelevant for the experimental results. 
A systematic approach to eliminate irrelevant details of this kind is to group together all procedures that result in the same statistical behaviour.
In this two-step setting, this amounts to identifying an equivalence class of preparation procedures that, for each measurement procedure, assign the same probabilities to each outcome.
That is, $p\sim p'$ if $P(k|p,m) = P(k|p',m)$ for all $m$ and $k$.
These equivalence classes are called the operational {\em states} of the experimental scenario.
Likewise, we can also divide pairs of measurement procedures and outcomes into equivalence classes that give the same probability for every preparation procedure: $(k,m) \sim (k', m')$ if $P(a|p,m) = P(k'|p,m')$ for all $p$.
The equivalence classes of these measurement outcomes are known as the {\em effects}.

Not every detail of preparation ``ignored'' by grouping into states is trivial -- especially when we consider experiments on nonclassical systems.
For instance, an experimentalist could prepare with equal probability a horizontally or vertically polarized photon (procedure $p_1$),
 and dispatch this photon into some measurement apparatus (whose ultimate behaviour, let us assume, depends {\em only} on the photon it receives).
Alternatively, the experimentalist could prepare with equal probability a right-circularly polarized photon, or a left-circularly polarized one (procedure $p_2$)
 and likewise dispatch this photon into the measurement device.
Although this measurement device might have different behaviour for all four polarized inputs, according to quantum theory, its average statistical behaviour after procedure $p_1$ will be identical to that of procedure $p_2$,
 and hence $p_1$ and $p_2$ describe the preparation of the same state. 
All properties of a preparation (or measurement) procedure that are not captured by its equivalence class (e.g.,\ for $p_1$ and $p_2$, the polarizations that are mixed) define the \emph{context} of the procedure~\cite{Spekkens05}.

\subsection{Generalized probabilistic theories}
\label{SubsecGPT}
Once we identify such states and effects of a physical system,
 we can employ the mathematical framework of {\em generalized probabilistic theories} (GPTs) (see e.g.~\cite{Hardy01,Barrett07} or the introductions~\cite{Mueller21,Plavala21}) to describe the system.
Here, the effects (equivalence classes of measurement outcomes)
 are represented by elements of a real vector space $A$  (see \cref{fig:GPTCones}). 
Except for infinite-dimensional quantum systems that we treat in detail in \cref{sec:Infinite}, we will only consider the case where $A$ is finite-dimensional.
We denote by $E_A$ the set of all effects of the system.

Since we permit probabilistic post-processing (i.e.\ statistical mixing), $E_A$ is a convex set. 
Furthermore, we may always assume that the linear span of $E_A$ is $A$, because otherwise, we could restrict ourselves to the subspace ${\rm span}(E_A)$. 
The set $E_A$ generates the \emph{cone~\cite{AliprantisT07} of effects} $A_+:=\{\lambda e\,\,|\,\, \lambda\geq 0, e\in E_A\}$.
Within $E_A$, there is a unique {\em unit effect} $u_A$, corresponding to the affirmative answer to the question ``is the system there?'',
 with the special property that for all $x\in E_A$, also $u_A - x \in E_A$. 
We assume that $E_A$ is topologically closed and bounded, i.e.\ compact; for the physical motivation of these assumption, see Ref.~\cite{Plavala21} (indeed, the GPTs determined via theory-agnostic tomography~\cite{MazurekPRS21,GraboweckyPCSR21}, our main cases of interest, will have these properties by construction.)
An ($n$-outcome) measurement is represented by a collection of effects $\{a_i\}_{i=1\ldots n}$ satisfying $\sum_{i=1}^n a_i = u_A$.

A {\em state}, as represented in a GPT, is a functional $\omega$ that describes the possible statistics arising from an equivalence class of preparation procedures. 
That is, a state assigns probabilities to every effect $a\in E_A$ by $(\omega,a):=\omega(a)$.
As probabilistic mixtures of effects must yield the correctly--weighted outcome probability~\cite{Barrett07}, $\omega$ is a linear functional, i.e.,\ an element of the dual vector space $A\conj$. 
The normalized states $\Omega_A$ are then the states in the GPT satisfying $\app{\omega}{u_A} = 1$. 
For similar reasons as for $E_A$, we will assume that $\Omega_A$ is convex and compact.
This set generates the \emph{cone of states} $A\conj_+:=\{\lambda\omega\,\,|\,\,\lambda\geq 0,\omega\in\Omega_A\}$. 

A GPT is said to be \emph{unrestricted}, or to satisfy the \emph{no-restriction hypothesis}~\cite{ChiribellaAP10,JanottaLal13}, if every functional that gives nonnegative probabilities for all effects is a valid state, and vice versa (that is, $\Omega_A=\{\omega\in A^*\,\,|\,\, 0\leq (\omega,e)\leq 1\forall e\in E_A\}$. 
In this case, $A_+$ and $A\conj_+$ are \emph{dual cones} to each other~\cite{AliprantisT07}.
In terms of preparations and measurements, this means that there exists at least one preparation procedure for every mathematically conceivable state acting on the system's effects, and vice versa.
Generally, however, a GPT may be restricted, and in this case $A_+$ and $A\conj_+$ are only subsets of each others' duals.

We thus specify a GPT by the triplet $(A,\Omega_A,E_A)$ of its vector space $A$, its set of normalized states $\Omega_A$, and its set of effects $E_A$. 
The other objects of interest ($u_A$, $A_+$, $A\conj_+$, etc.)\ are uniquely implied by these.

\begin{figure}[tbh]
\begin{centering}
\includegraphics[width=0.425\textwidth]{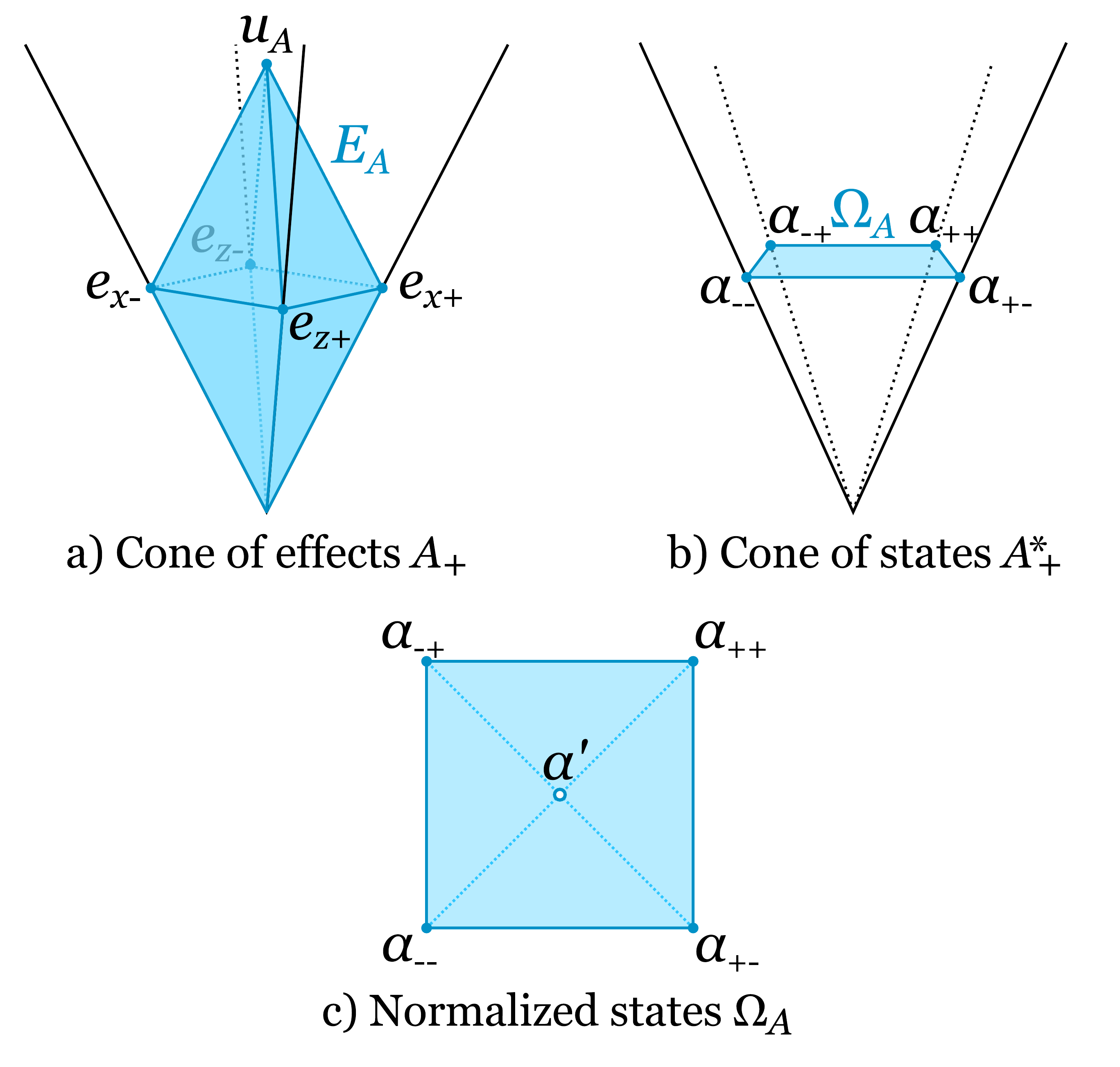}
\caption{%
\label{fig:GPTCones}
\caphead{Effects and states in generalized probabilistic theories (GPTs).}
A GPT $(A, \Omega_A, E_A)$ is specified by a vector space $A$, a set of normalized states $\Omega_A$ and a set of effects $E_A$.
Such a GPT defines two cones: the first (a), the cone of effects $A_+$ is generated by $E_A$;
 the second (b), the cone of states $A\conj_+$ is generated by $\Omega_A$ (c).
 The example drawn is a {\em gbit}, explained in the text.
}
\end{centering}
\end{figure}

\inlineheading{Example: classical probability theory.}
An $n$-outcome classical GPT is defined as $\mathcal{C}_n := \left(C, \Omega_C, E_C\right)$, 
 where $C :=  \reals^n$, 
 $\Omega_C$ is the set of $n$-outcome normalized probability distributions,
 and $E_C$ is the set of nonnegative vectors where no single element is greater than $1$.
Here, the effect cone $C_+$ is the set of all vectors with nonnegative elements, and the unit effect $u_{C} = \left(1, \ldots, 1\right)\trans$. 
Via the usual `dot' inner product, the classical state vector space $C\conj$ is also $\reals^n$, such that $\Omega_C$ may be written as the set of $n$-dimensional probability vectors with nonnegative elements that sum to $1$.
As every state that yields a valid probability is included in the classical state space, classical probability theory is unrestricted.

\inlineheading{Example: quantum theory.}
As a GPT, an $n$-level quantum system is given by $\mathcal{Q}_n := \left(B, \Omega_B, E_B\right)$
 where $B:=\Csa{n}$ is the vector space of complex Hermitian $n\times n$ matrices, 
 $\Omega_B$ the set of unit trace $n\times n$ density matrices, 
 and $E_B$ the set of all possible $n\times n$ POVM elements~\cite{NielsenC00}.
Quantum theory is also an unrestricted GPT.
Here, $B_+ = \Csap{n}$ is the cone of positive--semidefinite Hermitian matrices, and $u_B=\id_n$ is the $n\times n$ identity matrix. 
We identify $B$ and $B\conj$ via the Hilbert--Schmidt inner product $\inn{x}{y} := {\rm tr}(x y)$, such that $B\conj_+ = B_+ =\Csap{n}$. 
That is, not only are the two cones duals of each other (as in every unrestricted GPT), but there is an inner product such that they become exactly equal. This property of $\mathcal{Q}_n$ is known as \emph{strong self-duality}~\cite{Wilce11a,MuellerU12}, and is shared by the classical GPTs $\mathcal{C}_n$ (among some others).

Restricting the sets of states and effects to $n_i\times n_i$ block matrices in some basis gives us \textbf{quantum theory with superselection rules}~\cite{BartlettRS07}, defined on the linear space $\bigoplus_i \Csa{n_i}$.
This encompasses classical probability theory, since $\reals^n\simeq \bigoplus_{i=1}^n \Csa{1}$,  i.e.\ probability vectors can be identified with diagonal density matrices.

Although quantum theory is unrestricted, it has some interesting restricted subsets: for example the {\em stabilizer subset of quantum theory} (see, e.g., \cite{Gottesman97}).
For the stabilizer subset of $2$-level quantum theory, the allowed states correspond to the eigenvectors of the Pauli matrices and mixtures thereof, such that $\Omega_S$ is an octahedron.
The full dual cone of effects yielding positive probabilities includes measurements beyond those possible within standard quantum theory (corresponding to a cube) -- but stablizer quantum theory is defined to not admit all such effects, instead {\em restricting} the effect space to correspond also only to outcomes of Pauli matrix-measurements and mixtures thereof.

\inlineheading{Example: gbits.}
A foil theory to quantum mechanics are the {\em gbits}~\cite{Barrett07}.
These arise as the marginals of maximally nonlocal Popescu--Rohrlich boxes~\cite{PopescuR94}.
Let us specifically consider ``$2$ measurements $2$ outcomes'' gbits,
 and write this GPT as $\left(\reals^3, E_A, \Omega_A\right)$ where we define $E_A$ and $\Omega_A$ in the following paragraphs.

The effect space is generated by two choices of measurement ($X$ and $Z$), each that results in one of two outcomes ($+1$ or $-1$).
The extremal effects associated with each outcome can be represented by the vectors
\begin{align}
\label{eq:gbit_effects}
e_{x+}\!=\begin{small}\!\left(\begin{array}{c} 1 \\ 1  \\ 0 \end{array}\right)\end{small},
e_{x-}\!=\begin{small}\!\left(\begin{array}{c} 1 \\ -1  \\ 0 \end{array}\right)\end{small},
e_{z+}\!=\begin{small}\!\left(\begin{array}{c} 1 \\ 0  \\ 1 \end{array}\right)\end{small},
e_{z-}\!=\begin{small}\!\left(\begin{array}{c} 1 \\ 0  \\ -1 \end{array}\right)\end{small},
\end{align}
such that the unit effect is $u_A = e_{x+} + e_{x-} = e_{z+} + e_{z-} = \left( 2, 0, 0\right)\trans$.
The set of effects $E_A$ is the convex hull of all of the above, $u_A$, and $\left(0,0,0\right)\trans$
 (see \cref{fig:GPTCones}a).

We then take the unrestricted state space dual to the above effects.
This corresponds to admitting arbitrary probability distributions for both measurements.
Let us write each state as a functional (on effects) in the form $e\mapsto \frac{1}{2} \alpha\cdot e$ for $\alpha\in\reals^3$,
where $\cdot$ is the Euclidean dot product. 
Hence, we identify each state with such a vector $\alpha$,
 so that the extremal normalized gbit states are:
\begin{align}
& \alpha_{++} = \begin{small}\left(\begin{array}{c} 1 \\ 1  \\ 1 \end{array}\right)\end{small}, \;
\alpha_{+-} = \begin{small}\left(\begin{array}{c} 1 \\ 1  \\ -1 \end{array}\right)\end{small}, \; \nonumber \\
& \qquad \alpha_{-+} = \begin{small}\left(\begin{array}{c} 1 \\ -1  \\ 1 \end{array}\right)\end{small}, \;
\alpha_{--} = \begin{small}\left(\begin{array}{c} 1 \\ -1  \\ -1 \end{array}\right)\end{small}.
\label{eq:gbit_states}
\end{align}
The state space $\Omega_A$ is the convex hull of these states, and is a square (see \cref{fig:GPTCones}b and c).

\subsection{Jordan algebras}
A certain class of GPTs will turn out to be of fundamental importance in the following: those corresponding to special Euclidean Jordan algebras~\cite{JordanvNW34,AlfsenS03,FarautK94,McCrimmon04}.
Here we only give a minimal introduction; a more thorough treatment in the quantum foundations context can be found, for example, in~\citet{BarnumGW20}.
 
A \emph{Jordan algebra} $\mathcal{J}$ is a real vector space with a commutative bilinear product $\jord$ that satisfies the Jordan identity $(x^2\jord y)\jord x = x^2\jord  (y\jord x)$, where $x^2:=x\jord x$. 
In this paper, we will only consider finite-dimensional Jordan algebras. 
A Jordan algebra is \emph{Euclidean} if it is equipped with an inner product $\langle \cdot,\cdot\rangle$ such that $\langle x\jord y,z\rangle=\langle y,x\jord z\rangle$ for all $x,y,z\in\mathcal{J}$. 
The simplest example of a Euclidean Jordan algebra is the algebra $\mathbf{H}_n(\mathbb{C})$ of complex Hermitian $n\times n$ matrices with the Jordan product
\begin{align}
   x\jord y:=\frac{1}{2}\left(xy+yx\right),
\end{align}
equipped with the Hilbert-Schmidt inner product $\langle x,y\rangle=\tr(xy)$.
A Jordan algebra is \emph{special} if it is embeddable into $\mathbf{H}_n(\mathbb{C})$ for some $n$.

To every Euclidean Jordan algebra $\mathcal{J}$, we can associate a GPT (which we also call $\mathcal{J}$) in the following way. 
Define the effect cone as the cone of squares, $\mathcal{J}_+=\{x^2\,\,|\,\,x\in\mathcal{J}\}$, and the unit effect $u_{\mathcal{J}}$ as the unique unit element of the algebra.
The cone $\mathcal{J}_+$ is self-dual under the inner product that makes $\mathcal{J}$ Euclidean, i.e.\ $\{x\in\mathcal{J}\,\,|\,\, \langle x,y\rangle\geq 0\;\forall\, y\in\mathcal{J}_+\}=\mathcal{J}_+$. 
If we thus use this inner product to identify $\mathcal{J}$ and $\mathcal{J}^*$, we can define the state cone to be equal to $\mathcal{J}_+$, and this gives us a strongly self-dual unrestricted GPT $(\mathcal{J},\Omega_{\mathcal{J}},E_{\mathcal{J}})$, where $\Omega_{\mathcal{J}}$ is the set of normalized elements of $\mathcal{J}_+$, and $E_{\mathcal{J}}$ is the set of elements $e$ of $\mathcal{J}_+$ such that $\langle \omega,e\rangle\in [0,1]$ for all $\omega\in\Omega_{\mathcal{J}}$.

The Euclidean Jordan algebras have appeared in various places in the foundations of quantum mechanics, including in early work of Jordan, von Neumann and Wigner~\cite{JordanvNW34} who also gave a complete classification: 
 these are the algebras of Hermitian matrices over the real numbers $\mathbb{R}$, the complex numbers $\mathbb{C}$, or the quaternions $\mathbb{H}$; so-called spin factors; the exceptional Jordan algebra $\mathbf{H}_3(\mathbb{O})$ of self-adjoint $3\times 3$ octonionic matrices; and direct sums of those. 
The exceptional Jordan algebra $\mathbf{H}_3(\mathbb{O})$ is not special, and it cannot appear as a summand in any special Jordan algebra.

The corresponding GPTs appear in many different axiomatizations of quantum theory (e.g.~\cite{Wilce11,Niestegge12,BarnumMU14,Wilce19})
 since they share many important properties with standard complex quantum theory. 
For example, if a GPT admits a strong symmetry property and a spectral decomposition of states, then it is either classical or corresponds to a simple Euclidean Jordan algebra~\cite{BarnumMU14,BarnumH19}.
We will give a more detailed description of the GPTs of special Euclidean Jordan algebras in \cref{sec:embed_examples}, where they will surprisingly show up as the {\em only} unrestricted GPTs that admit univalent simulation by standard complex quantum theory.

\section{Simulations, embeddings, and contextuality}
\label{sec:Simulation}
We are often confronted with effective theories $\mathcal{A}$ that arise from more fundamental theories $\mathcal{B}$ via some approximation or coarse-graining. 
For example, thermodynamics ($\mathcal{A}$) is expected to emerge from classical mechanics ($\mathcal{B}$) under some sort of coarse-graining; 
 or low-energy effective quantum field theory ($\mathcal{A}$) is expected to provide an approximation to some hitherto unknown high-energy theory $\mathcal{B}$ that perhaps includes a version of quantum gravity.

In this article, we are concerned with the special case that both $\mathcal{A}$ and $\mathcal{B}$ are probabilistic theories. 
The question of whether a GPT $\mathcal{A}$ can supervene on some other, potentially more fundamental, GPT $\mathcal{B}$ arises in several physically interesting contexts:
\begin{itemize}
\item The case of effective classical probability theory (CPT) $\mathcal{A}$ and fundamental quantum theory (QT) $\mathcal{B}$ addresses the question of \emph{emergence of classicality} via, for example, decoherence processes.
\item Reversing the order, i.e.\ asking for QT $\mathcal{A}$ to effectively arise from the more fundamental CPT $\mathcal{B}$, corresponds to the study of \emph{hidden-variable models} (ontological models) of QT and their properties (such as nonlocality or contextuality).
\item If $\mathcal{A}$ is an arbitrary GPT and $\mathcal{B}$ is QT, then we arrive at the main question of this paper: 
 which GPTs can we expect to find as effective theories (for example experimentally) if we assume that QT is fundamentally correct? 
 Can we find criteria that \emph{falsify} QT as a fundamental theory, given only experimental data on the effective theory?
\end{itemize}
In all these cases, there is a \emph{fundamental} level of physical description, corresponding to a GPT $\mathcal{B}$, and an \emph{effective} level described by another GPT $\mathcal{A}$.
Whenever we think that we are preparing a state $\omega_A$ (resp.\ measuring an effect $e_A$) in the effective theory, we will in fact prepare some state $\omega_B$ (resp.\ measure some effect $e_B$) in the fundamental theory, such that the probabilities $(\omega_A,e_A)$ are reproduced via $(\omega_B,e_B)$.

In more detail, consider the (``effective'') operational theory (recall Subsection~\ref{sec:PrepAndMeasure}) that describes all preparation and measurement procedures of an effective physical system available to us in some experimental scenario.
As explained in Subsection~\ref{SubsecGPT}, this gives rise to a description in terms of an (effective) GPT $\mathcal{A}$. 
However, here we imagine a scenario where there is also a more fundamental description in terms of some GPT $\mathcal{B}$.

Suppose we have two different preparation procedures, $p$ and $p'$, that prepare the same effective state $\omega_A$. 
In the effective operational theory, these procedures are thus equivalent, $p\sim p'$. 
However, $p$ and $p'$ may lead to the preparation of \emph{different} states of the fundamental GPT, $\omega_B\neq \omega'_B$. 
More generally, consider all preparation procedures $p$ that prepare the effective state $\omega_A$ (mathematically, all elements of the equivalence class $\omega_A$), and collect all the actually fundamentally prepared states $\omega_B$ into a set which we call ${\rm Sim}_{\mathcal{B}}(\omega_A)$. 
These are all the fundamental states that ``simulate'' the effective state $\omega_A$.

Similarly, to every effective effect $e_A$, we obtain a set of simulating effects ${\rm Sim}_{\mathcal{B}}(e_A)$ in the fundamental theory.
Together, these will reproduce the effective probabilities:
\begin{align}
   (\omega_A,e_A)=(\omega_B,e_B)\enspace\forall\omega_B\in{\rm Sim}_{\mathcal{B}}(\omega_A),e_B\in {\rm Sim}_{\mathcal{B}}(e_A).
\end{align}
While the ``true'' effective GPT $\mathcal{A}$ governing our experimental scenario must satisfy this equation exactly by construction, we will typically encounter situations where experimentally determined approximations $\mathcal{A}''\approx \mathcal{A}$ (see Subsection~\ref{SubsecApproxExp}) reproduce the probabilities only approximately. 
This motivates the more general formulation below (\cref{eq:match_stats} in \cref{def:epssim}).

What can we say about the set ${\rm Sim}_{\mathcal{B}}(\omega_A)$ of states that simulate a given state $\omega_A$? 
We would like to have a definition of simulation that is as broad as possible, and make no further assumptions except for those that follow from the very definition of an operational theory.
One crucial property of an operational theory is that it is in principle possible to have \emph{fluctuating preparations}, i.e.\ to use external randomness for the selection of a preparation procedure. 
Concretely, consider two preparation procedures $p$ and $p'$, and denote the effective resp.\ fundamental states that they prepare by $\omega_A$ and $\omega'_A$ resp.\ $\omega_B$ and $\omega'_B$. 
If $0<\lambda<1$ is any real number, then we can consider the preparation procedure that triggers $p$ with probability $\lambda$ and $p'$ with probability $(1-\lambda)$. 
Effectively, this prepares the state $\lambda \omega_A+(1-\lambda)\omega'_A$; fundamentally, it prepares the state $\lambda \omega_B+(1-\lambda)\omega'_B$. 
Thus, the latter is a fundamental state that simulates the former effective state:
\begin{equation}
   \lambda\omega_B+(1-\lambda)\omega'_B\in {\rm Sim}_{\mathcal{B}}(\lambda\omega_A+(1-\lambda)\omega'_A).
   \label{eqFluctuations}
\end{equation}
(We stress, in the absence of further restrictions, there remains the possibility that ${\rm Sim}_{\mathcal{B}}(\lambda\omega_A+(1-\lambda)\omega'_A)$ contains other states that also simulate $\lambda\omega_A+(1-\lambda)\omega'_A$.)

We can repeat this argument for the effects and obtain \cref{eq:Mixing1} in the following definition:
\begin{definition}[Simulation]
\label{def:epssim}
Consider $\mathcal{A}=(A,\Omega_A,E_A)$ (the ``effective GPT'') and $\mathcal{B}=(B,\Omega_B,E_B)$ (the ``fundamental GPT''), and let $\varepsilon\geq 0$. 
An \textbf{$\varepsilon$-simulation} of $\mathcal{A}$ by $\mathcal{B}$ assigns to each $\omega_A\in\Omega_A$ a nonempty set of states ${\rm Sim}_{\mathcal{B}}(\omega_A)\subset\Omega_B$ (``the states of $\mathcal{B}$ that simulate $\omega_A$''), and to every normalized effect $e_A\in E_A$ a nonempty set of effects ${\rm Sim}_{\mathcal{B}}(e_A)\subset E_B$ (``the effects of $\mathcal{B}$ that simulate $e_A$''), such that the following conditions hold:
\begin{itemize}
\item all outcome probabilities are reproduced up to $\varepsilon$: for all $\omega_A\in\Omega_A,e_A\in E_A$, we have
\begin{equation}
   |(\omega_A,e_A)-(\omega_B,e_B)|\leq \varepsilon
      \label{eq:match_stats}
\end{equation}
for all $\omega_B\in {\rm Sim}_{\mathcal{B}}(\omega_A),e_B\in {\rm Sim}_{\mathcal{B}}(e_A)$;
\item mixtures of simulating states (effects) are valid simulations of mixtures of states (effects): for all probabilities $0<\lambda<1$, we have
\begin{align}
\omega_B\in{\rm Sim}_{\mathcal{B}}(\omega_A)\mbox{ and }\omega'_B\in{\rm Sim}_{\mathcal{B}}(\omega'_A)\nonumber\\
\Rightarrow\lambda\omega_B+(1-\lambda)\omega'_B\in{\rm Sim}_{\mathcal{B}}\left(\strut \lambda\omega_A+(1-\lambda)\omega'_A\right),
   \label{eq:Mixing1}
\end{align}
and the analogous condition for effects;
\item the fundamentally impossible effect is a valid simulation of the effectively impossible effect:
\begin{align}
0 \in {\rm Sim}_{\mathcal{B}}(0).
\label{eq:ZeroToZero}
\end{align}
\end{itemize}
An ($\varepsilon\!=\!0$)-simulation is called an \textbf{exact simulation}. 
The simulation is called \textbf{preparation--univalent} if $|{\rm Sim}_{\mathcal{B}}(\omega_A)|=1$ for all $\omega_A\in\Omega_A$, \textbf{measurement--univalent} if $|{\rm Sim}_{\mathcal{B}}(e_A)|=1$ for all $e_A\in E_A$, and \textbf{univalent} if it is both preparation-- and measurement--univalent. 
As the negation of ``univalent'', we use the term ``multivalent''.
\end{definition}

\begin{figure}[tbh]
\includegraphics[width=0.3\textwidth]{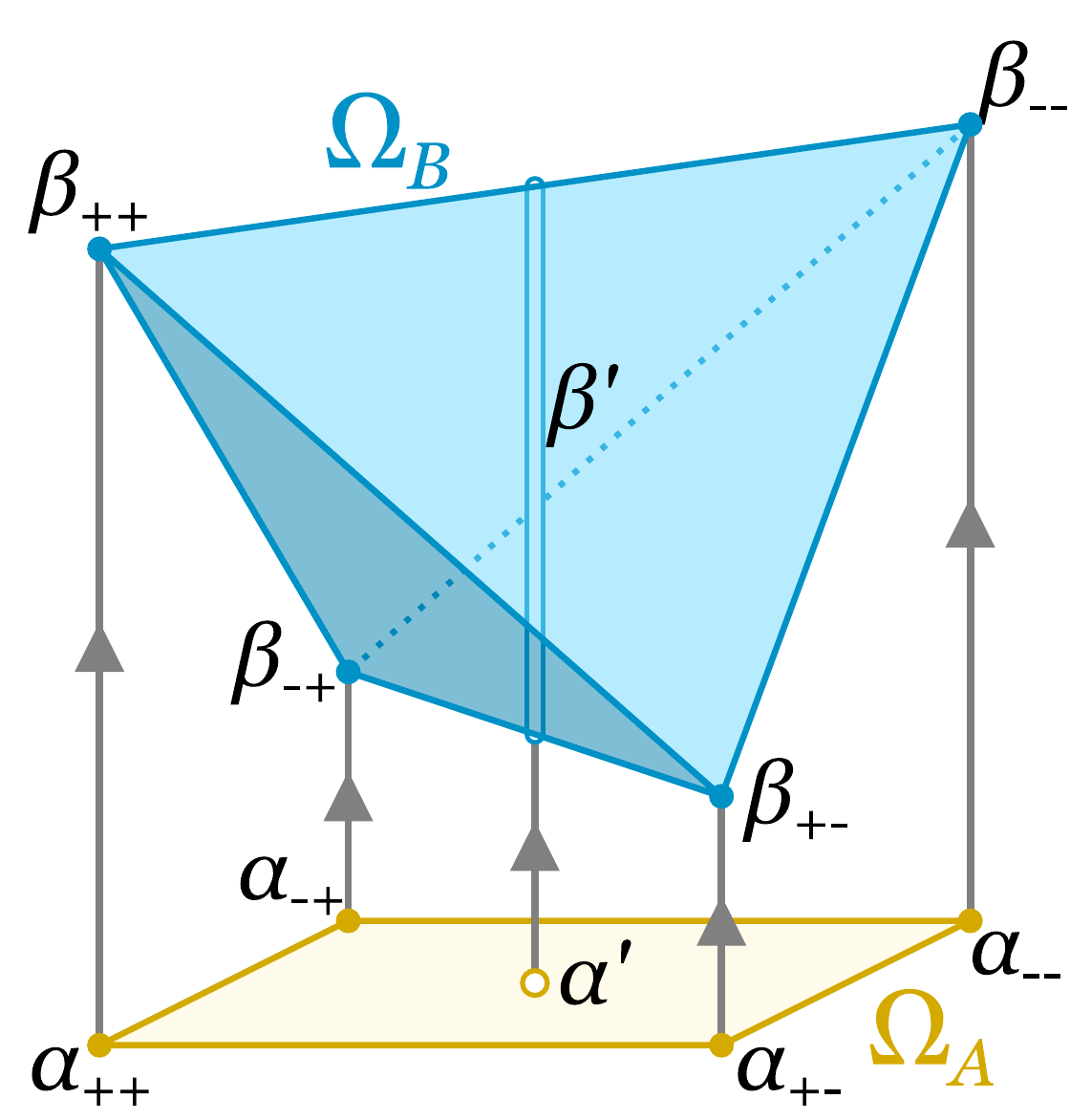}
\caption{
\label{fig:Holevo}
\caphead{Simulating gbit states classically.} 
The ``Holevo projection''~\cite{Holevo82} simulates a gbit (yellow square $\Omega_A$) by a four-level classical system (blue tetrahedron $\Omega_B$).
While the extremal gbit states $\{\alpha_{++}, \alpha_{+-}, \alpha_{-+}, \alpha_{--}\}$ are mapped to unique classical states $\{\beta_i\} := {\rm Sim}_{\mathcal{B}}(\alpha_i)$,
 for general states the mapping is not single-valued.
E.g.\ the state $\alpha' = \frac{1}{2}(\alpha_{++} + \alpha_{--}) = \frac{1}{2}(\alpha_{+-} + \alpha_{-+})$ maps onto a line of infinitely many states $\{\beta'\}={\rm Sim}_{\mathcal{B}}(\alpha')$,
 where the different contexts in which $\alpha'$ is prepared (e.g.\ mixing $\alpha_{++}$ and $\alpha_{--}$, or mixing $\alpha_{+-}$ and $\alpha_{-+}$) necessitate the preparation of different states in $\mathcal{B}$.
This simulation is thus {\em preparation-multivalent}.
}
\end{figure}

As we will elaborate on in Section~\ref{sec:StandardContextuality} below, Spekkens' notion of generalized noncontextuality~\cite{Spekkens05} corresponds, in the discrete case, to univalent simulation in the special case $\mathcal{B}=\mathcal{C}_n$, i.e.\ where the fundamental theory is classical.
Hence, the above definition generalizes this notion to the simulation of any GPT by any other GPT.
Indeed, for much of this paper, we will be concerned with identifying theories that can beunivalently simulated by {\em quantum theory}. 
In \cref{sec:Plausibility}, we will discuss in more detail why univalence is a plausible assumption in our setting, generalizing arguments that have been put forward in favour of generalized noncontextuality.

As an example, let us consider an exact simulation of a {\em gbit} $\mathcal{A}=(\reals^3, E_A, \Omega_A)$,
 with state and effect spaces as defined in Subsection~\ref{SubsecGPT} above (\cref{eq:gbit_effects,eq:gbit_states}).
This theory can be simulated by $4$-level classical probability theory $C:=\mathcal{C}_4$ using an observation of~\citet{Holevo82}.
In particular, consider the map $e_i \mapsto f_i$ that acts to take the extremal gbit effects to
\begin{align}
f_{x+}\!=\!\left(\begin{small}\begin{array}{c} 1 \\ 1 \\ 0 \\ 0 \end{array}\end{small}\right)\!,
f_{x-}\!=\!\left(\begin{small}\begin{array}{c} 0 \\ 0 \\ 1 \\ 1 \end{array}\end{small}\right)\!,
f_{z+}\!=\!\left(\begin{small}\begin{array}{c} 1 \\ 0 \\ 1 \\ 0 \end{array}\end{small}\right)\!,
f_{z-}\!=\!\left(\begin{small}\begin{array}{c} 0 \\ 1 \\ 0 \\ 1 \end{array}\end{small}\right)\!,
\label{eq:embedded_gbit_effects}
\end{align}
and $u_A\mapsto u_B := \left(1, 1, 1, 1\right)\trans$. 
This assignment can be linearly extended to a map $\mathbb{R}^3\to\mathbb{R}^4$ which is positive and unital, and which hence defines for every gbit effect $e_A\in E_A$ a corresponding classical effect $e_C\in E_C$. 
However, the set of these effects $e_C$ is not linearly independent and spans only a subspace of $\mathbb{R}^4$.

Now, how do we simulate the gbit states such that probabilities are preserved?
Going through each of the four corners $\alpha_i$, we can pick a vector $\beta_i$ that satisfies \cref{eq:match_stats} for $\varepsilon=0$, when the effects are defined as in \cref{eq:embedded_gbit_effects}.
For these states, there is only one such choice, namely, to map each corner to an orthogonal classical state:
\begin{align}
\beta_{++} := \left(1, 0, 0, 0\right)\trans, \qquad
& \beta_{+-} := \left(0, 1, 0, 0\right)\trans, \nonumber \\
\beta_{-+} := \left(0, 0, 1, 0\right)\trans, \qquad
& \beta_{--} := \left(0, 0, 0, 1 \right)\trans,
\end{align}
as drawn in \cref{fig:Holevo}.

The assignment $\alpha_i\mapsto \beta_i$ for $i\in\{++,+-,-+,--\}$ cannot be extended to a linear map (the image of a square under any linear map cannot contain more than two linearly independent elements), and the trouble this causes is more obvious when we consider the simulation of the nonextremal gbit states.
Consider the gbit state $\alpha'=\left(1, 0, 0\right)\trans$ that is in the center of the gbit's square state space.
Any state of the form $\beta'=\left(\frac{\lambda}{2}, \frac{1-\lambda}{2}, \frac{1-\lambda}{2}, \frac{\lambda}{2},\right)\trans$ for $\lambda\in[0,1]$ will yield the correct statistics on the effects $e_C$, and so is a suitable candidate for simulation. 
Thus, we obtain a \emph{set} of states $\{\beta'\}={\rm Sim}_{\mathcal{C}}(\alpha')$, in line with \cref{eq:Mixing1}.
Suppose we consider $a'$ as the equal mixture of $\alpha_{++}$ and $\alpha_{--}$: then the set ${\rm Sim}_{\mathcal{B}}(\alpha')={\rm Sim}_{\mathcal{B}}(\frac{1}{2} \alpha_{++} + \frac{1}{2}\alpha_{--})$ must contain the mid-point of the line between $\beta_{++}$ and $\beta_{--}$ (midpoint of top line of tetrahedron in \cref{fig:Holevo}).
Conversely, the very same point $a'$ is also the equal mixture of $\alpha_{+-}$ and $\alpha_{-+}$, mandating that ${\rm Sim}_{\mathcal{B}}(\alpha')={\rm Sim}_{\mathcal{B}}(\frac{1}{2}\alpha_{+-} + \frac{1}{2}\alpha_{-+})$ must contain the mid-point of the line between $\beta_{+-}$ and $\beta_{-+}$ (midpoint of bottom line of tetrahedron in \cref{fig:Holevo}).
Indeed, writing $\alpha' := \frac{\lambda}{2}\left(\alpha_{++} + \alpha_{--}\right) + \frac{1-\lambda}{2}\left(\alpha_{+-}+\alpha_{-+}\right)$ for $\lambda\in[0,1]$,
 we see that this mandates the inclusion of every point in $\beta'$.
Indeed, for this gbit simulation, every nonextremal state maps to an infinite number of classical states. 
That is, this classical simulation of the gbit is preparation-multivalent.

In particular: suppose we prepare with $50\%$ probability $\alpha_{++}$, and $50\%$ probability $\alpha_{--}$, necessitating the preparation of $\beta_{++}$ or $\beta_{--}$ in the fundamental theory $\mathcal{B}$.
If we wanted to prepare the mixture of $50\%$ $\alpha_{+-}$ and $50\%$ $\alpha_{-+}$, we would have to prepare completely different states $\beta_{+-}$ and $\beta_{-+}$ in $\mathcal{B}$.
Thus, to know how $\alpha'$ is simulated in $\mathcal{B}$, knowledge of the operational statistics is insufficient: 
 we would also need to know the {\em context} ($\alpha_{++}$/$\alpha_{--}$ or $\alpha_{+-}$/$\alpha_{-+}$) in which it was prepared.
This arguably introduces an implausible amount of finetuning in the explanation of how the effective theory $\mathcal{A}$ is supposed to arise from the fundamental theory $\mathcal{B}$.

We will say more about this notion of univalence, its motivation, and its relation to generalized noncontextuality in \cref{sec:StandardContextuality,sec:Plausibility} below.

The above Holevo simulation of the gbit is not a rare pathological example: \emph{every} finite-dimensional GPT can be classically simulated in this way.

\begin{figure}[tbh]
\includegraphics[width=0.2\textwidth]{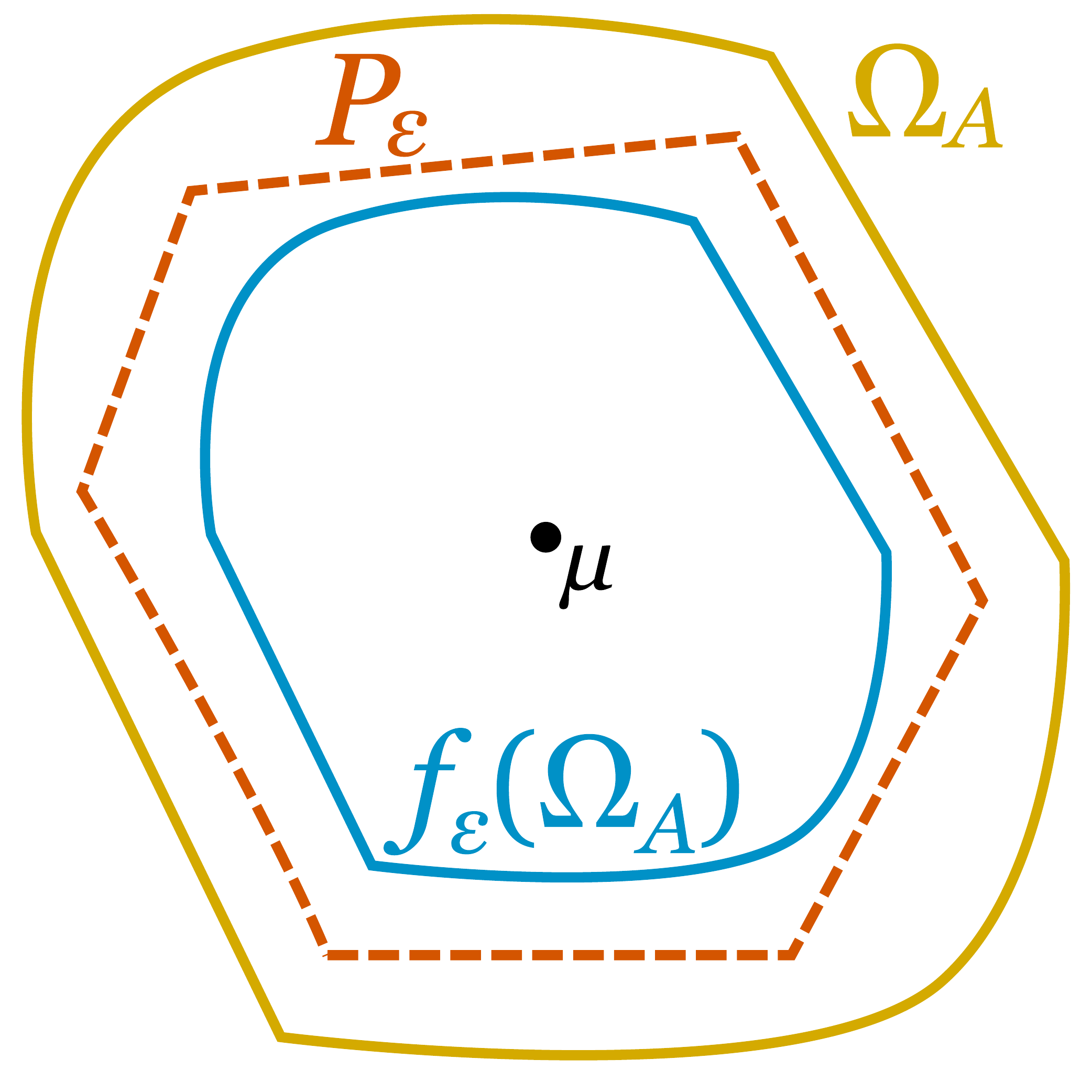}
\caption{
\label{fig:ShrinkConvex}
\caphead{Approximating convex sets by polytopes.} 
When map $f_\varepsilon$ ``shrinks'' convex set $\Omega_A$ towards interior point $\mu$, there always exists a polytope $P_\varepsilon$ between $\Omega_A$ and $f_\varepsilon(\Omega_A)$, the ``shadow'' of a classical GPT. Together with a bound on the number of vertices, this is proven in Appendix~\ref{app:Polytope}.
}
\end{figure}

\begin{lemma}
\label{lem:AllContextual}
Let $\mathcal{A}$ be any GPT. 
Then, for every $\varepsilon>0$, there is a measurement-univalent (but, in general, preparation-multivalent) $\varepsilon$-simulation of $\mathcal{A}$ by $\mathcal{C}_n$ (and thus by $\mathcal{Q}_n$) for some $\displaystyle n\leq\left\lceil \left(\frac c \varepsilon\right)^{(\dim A-2)/2}\right\rceil$, where $c>0$ is a constant that only depends on $\Omega_A$.
\end{lemma}
\begin{proof}
Pick some point $\mu$ in the relative interior~\cite{Webster94} of $\Omega_A$. 
Then the function $f_\varepsilon:\Omega_A\to\Omega_A$
\begin{align}
   f_\varepsilon(\omega_A):=\varepsilon\mu+(1-\varepsilon)\omega_A
\end{align}
``shrinks'' $\Omega_A$ towards $\mu$. 
Geometric intuition (\cref{fig:ShrinkConvex}) suggests that there exists a convex polytope $P_\varepsilon$ with all vertices in $\Omega_A\setminus f_\varepsilon(\Omega_A)$, such that $f_\varepsilon(\Omega_A)\subset P_\varepsilon\subset \Omega_A$. {}Lemma~\ref{LemPolytopeInBetween} in Appendix~\ref{app:Polytope} gives a rigorous proof that this is indeed the case, and gives the claimed bound on the number of vertices $n:=n_\varepsilon$.
Denote the vertices of $P_\varepsilon$ (in arbitrary order) by $v_{\varepsilon,1},v_{\varepsilon,2},\ldots,v_{\varepsilon,n_\varepsilon}$, and define the linear map $L_\varepsilon:\mathbb{R}^{n_\varepsilon}\to A^*$ via $L_\varepsilon e_i:=v_{\varepsilon,i}$ for $i=1,\ldots,n_\varepsilon$, where $e_i$ denotes the $i$th unit vector of $\mathbb{R}^{n_\varepsilon}$.
Consider the classical GPT $\mathcal{C}:=\mathcal{C}_{n_\varepsilon}$, then the polytope $P_\varepsilon$ is the image of the simplex $\Omega_C$ under $L_\varepsilon$~\cite{Ziegler95}.
For $\omega_A\in\Omega_A$ and $e_A\in E_A$, define the sets
\begin{align}
	{\rm Sim}_{\mathcal{C}}(\omega_A)&:=\{\omega_C\in\Omega_C\,\,|\,\, L_\varepsilon\omega_C = f_\varepsilon(\omega_A)\},\\
	{\rm Sim}_{\mathcal{C}}(e_A)&:= \{L_\varepsilon^*(e_A)\},
\end{align}
where $L_\varepsilon^*:A\to \mathbb{R}^{n_\varepsilon}$ is the dual of $L_\varepsilon$.
Since $f_\varepsilon(\omega_A)\in P_\varepsilon$, there must be at least one $\omega_C\in\Omega_C$ which is mapped to this point via $L_\varepsilon$, hence ${\rm Sim}_{\mathcal{C}}(\omega_A)$ is a nonempty subset of classical states. The set ${\rm Sim}_{\mathcal{C}}(e_A)$ contains a single element $e_C$, and it satisfies $(\omega_C,e_C)=(L_\varepsilon\omega_C,e_A)\in [0,1]$ for all $\omega_C\in\Omega_C$ since $L_\varepsilon\omega_C\in P_\varepsilon\subseteq \Omega_A$. Thus ${\rm Sim}_{\mathcal{C}}(e_A)\subset E_C$. 
Furthermore, \cref{eq:Mixing1,eq:ZeroToZero} follow from convex-linearity of $f_\varepsilon$ and linearity of $L_\varepsilon$.

Now, for $\omega_A\in\Omega_A$ and $e_A\in E_A$, pick any $\omega_C\in{\rm Sim}_{\mathcal{C}}(\omega_A)$ and $e_C\in {\rm Sim}_{\mathcal{C}}(e_A)$. Then
\begin{align}
(\omega_C,e_C)&=(L_\varepsilon \omega_C,e_A)=(f_\varepsilon(\omega_A),e_A) \nonumber \\
&= (\omega_A,e_A)+\varepsilon(\mu-\omega_A,e_A).
\end{align}
But $|(\mu-\omega_A,e_A)|\leq 1$, and so $|(\omega_C,e_C)-(\omega_A,e_A)|\leq\varepsilon$. This shows that the above maps define an $\varepsilon$-simulation of $\mathcal{A}$ by $\mathcal{C}_{n_\varepsilon}$.
\end{proof}
Classical simulations cannot in general preserve compositional structure: simulating two GPTs $\mathcal{A}_1$ and $\mathcal{A}_2$ as in Lemma~\ref{lem:AllContextual} will not in general yield a valid classical simulation of a given composite GPT $\mathcal{A}_1\mathcal{A}_2$. This is why Lemma~\ref{lem:AllContextual} is unrelated to, and in particular not in conflict with, Bell's theorem.

It turns out that univalent simulations have a particularly simple structure: they are (linear) \emph{embeddings}. 
We will now first define the notion of embedding and then prove this statement as a lemma.
\begin{definition}[Embedding]
\label{def:EmbeddingMap}
Let $\mathcal{A}=(A,\Omega_A,E_A)$ and $\mathcal{B}=(B,\Omega_B,E_B)$ be GPTs, and let $\varepsilon\geq 0$. A pair of linear maps $\Phi:A\to B$ and $\Psi:A^*\to B^*$ is said to be an {\bf $\varepsilon$-embedding} of $\mathcal{A}$ into $\mathcal{B}$ if
\begin{enumerate}[(i)]
\item $\Phi$ and $\Psi$ are {\em positive} and $\Psi$ is normalization-preserving, i.e.\ $\Phi(E_A)\subseteq E_B$ and $\Psi(\Omega_A)\subseteq \Omega_B$;
 \item $\Phi$ and $\Psi$ {\em preserve outcome probabilities up to $\varepsilon$}; i.e.\ $|(\omega,e)-(\Psi(\omega),\Phi(e))|\leq\varepsilon$ for all $e\in E_A$, $\omega \in \Omega_A$.
\end{enumerate}
 If in addition $\Phi(u_A)=u_B$, then we say that the embedding is {\bf unital}. An ($\varepsilon=0$)-embedding is also called an {\bf exact embedding}.
\end{definition}
This notion of approximate embeddings as pairs of structure-preserving (and thus \textbf{linear}) maps has already been introduced and studied by~\citet[p.\ 150]{Werner82} for the case that $\mathcal{B}$ is a quantum system and $\mathcal{A}$ a possibly infinite-dimensional classical system. 
Here we are concerned with general GPTs and finite-dimensional $\mathcal{A}$.
\begin{lemma}[Univalent simulations are embeddings]
\label{lem:Linear}
Every $\varepsilon$-embedding of $\mathcal{A}$ into $\mathcal{B}$ defines an univalent $\varepsilon$-simulation of $\mathcal{A}$ by $\mathcal{B}$, and vice versa.
\end{lemma}
The proof is straightforward and thus deferred to Appendix~\ref{app:EmbeddingProperties}. In the special case of $\mathcal{B}=\mathcal{C}_n$ and $\varepsilon=0$, univalence is equivalent to generalized noncontextuality~\cite{Spekkens05} (more on this in Section~\ref{sec:StandardContextuality}), and the equivalence of this to (linear) embeddings has already been shown in Refs.~\cite{SchmidSWKW21,Plavala21Incompatibility}. Since univalent simulations are a further generalization of this notion, it should not come as a surprise that we obtain a linear notion of embedding as well.

\begin{figure}[tbh]
\begin{centering}
\includegraphics[width=0.45\textwidth]{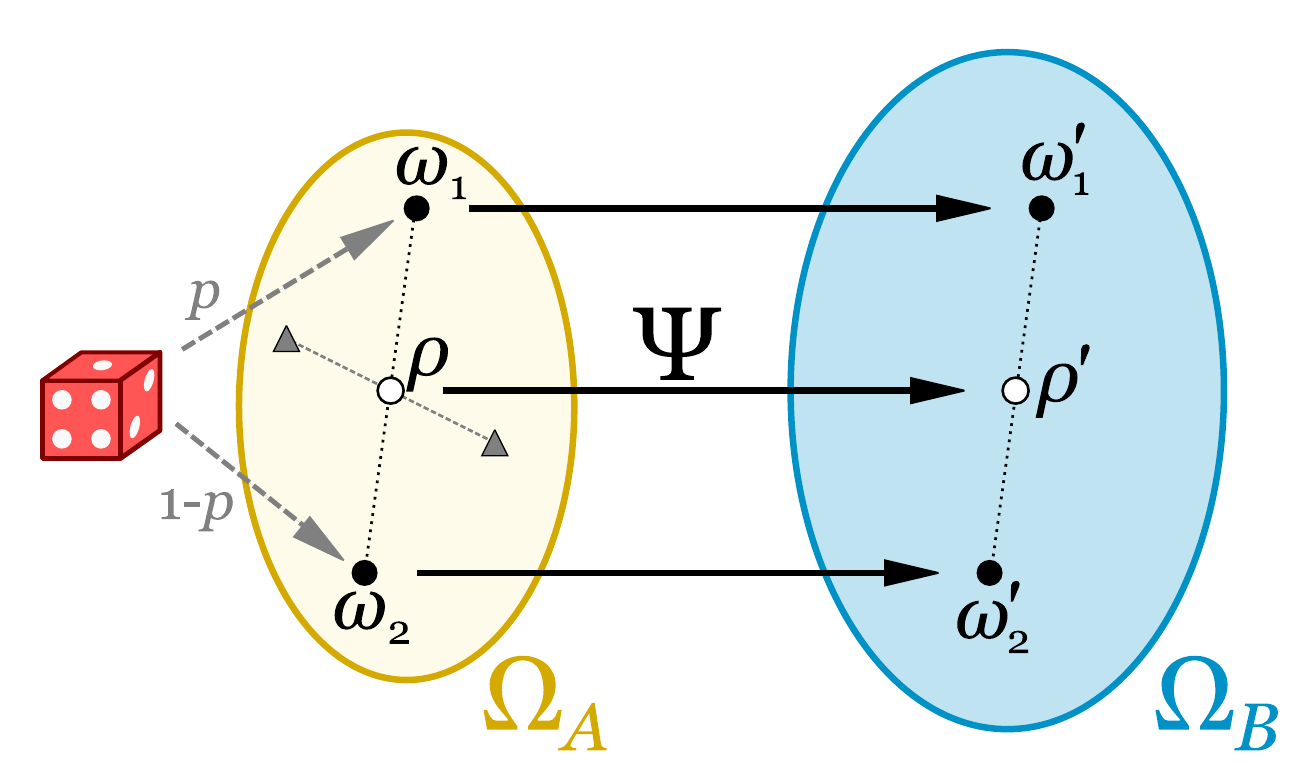}
\caption{%
\label{fig:ProcIndep}
\caphead{\mbox{Why embeddings must be linear.} }
Consider preparing the effective state $\rho:= p \omega_1+\left(1-p\right) \omega_2$ in $\mathcal{A}$ by randomly preparing either $\omega_1$ with probability $p$ or otherwise preparing $\omega_2$. 
For an univalent simulation where $\omega_i$ is uniquely simulated by $\omega'_i=\Psi(\omega_i)$, this will prepare the fundamental state $\omega'_1$ with probability $p$ and $\omega'_2$ with probability $1-p$, i.e.\ it will prepare $\rho'=p\omega'_1+(1-p)\omega'_2$. Hence $\Psi(\rho)=\rho'$, which is an instance of linearity.
}
\end{centering}
\end{figure}
It is important to understand that linearity of embeddings is not an \emph{assumption} of our approach, but a \emph{result}: it follows from the crucial property of operational theories to admit random preparations, as explained around Eq.~(\ref{eqFluctuations}). 
It implies that statistical mixtures of fundamental states are always valid simulations of the corresponding mixtures of effective states. 
This leads to condition~(\ref{eq:Mixing1}) in Definition~\ref{def:epssim}, which becomes linearity in the case of univalent simulations, as also illustrated in Figure~\ref{fig:ProcIndep}.

Furthermore, if we deliberately dropped the condition of linearity from the definition of an embedding (ignoring its central importance in expressing the possibility of random preparations!), then the resulting notion would become trivial: \emph{all} GPTs would then be embeddable into classical probability theory, as demonstrated in Lemma~\ref{lem:AllContextual}, and no experimental tests of contextuality \emph{or} of QT could rely on this nonlinear notion of embeddability.

It is clear that embeddings satisfy a transitivity property: for GPTs $\mathcal{A}$, $\mathcal{B}$ and $\mathcal{C}$, embedding $\mathcal{A}$ into $\mathcal{B}$ and then $\mathcal{B}$ into $\mathcal{C}$ defines an embedding of $\mathcal{A}$ into $\mathcal{C}$:
\begin{lemma}[Transitivity of embeddings]
\label{LemIterate}
Let $(\Phi,\Psi)$ define an $\varepsilon$-embedding of $\mathcal{A}$ into $\mathcal{B}$, and $(\Phi',\Psi')$ define a $\delta$-embedding of $\mathcal{B}$ into $\mathcal{C}$, where $\varepsilon,\delta\geq 0$. Then $(\Phi'\circ\Phi,\Psi'\circ\Psi)$ defines an $(\varepsilon+\delta)$-embedding of $\mathcal{A}$ into $\mathcal{C}$.
\end{lemma}
\noindent The proof is straightforward and thus omitted. 

Univalence thus extends transitively across different levels of description: think of $\mathcal{A}$ as an effective theory, $\mathcal{B}$ as a somewhat more fundamental (``intermediate'') theory, and $\mathcal{C}$ as the most fundamental among the three. 
If $\mathcal{A}$ has an univalent explanation in terms of $\mathcal{B}$, and so does $\mathcal{B}$ in terms of $\mathcal{C}$, then the effective theory $\mathcal{A}$ has an univalen explanation in terms of the fundamental one $\mathcal{C}$ (with the approximation errors adding up).

Univalence does \emph{not} imply that there cannot be distinct elements of the fundamental (simulating) GPT $\mathcal{B}$ that produce identical statistics in the effective (simulated) GPT $\mathcal{A}$. 
These elements can (and typically will) exist, but the point is that they are not necessary for the simulation. 
For example, let $\mathcal{A}=\mathcal{C}_2$ and $\mathcal{B}=\mathcal{C}_4$ be classical $2$- and $4$-level systems, respectively. 
Then the maps
\begin{align}
   \left(\begin{array}{c} p_1 \\ p_2\end{array}\right)\stackrel{\Psi}\mapsto \left(\begin{array}{c} p_1/2 \\ p_1/2 \\ p_2/2 \\ p_2/2\end{array}\right),\quad
   \left(\begin{array}{c}e_1 \\ e_2 \end{array}\right)\stackrel{\Phi}\mapsto \left(\begin{array}{c} e_1 \\ e_1 \\ e_2 \\ e_2 \end{array}\right)
   \label{eq:Coarsegraining}
\end{align}
define an exact embedding of $\mathcal{A}$ into $\mathcal{B}$, i.e.\ an exact univalent simulation of $\mathcal{A}$ by $\mathcal{B}$. 
We can interpret $\mathcal{A}$ as a coarse-graining of $\mathcal{B}$. Instead of simulating the state $(1,0)\trans\in\Omega_A$ via $\Psi((1,0)\trans)=(\frac{1}{2},\frac{1}{2},0,0)\trans$, we \emph{could} also simulate it via, for example, $(\frac{3}{4}, \frac{1}{4},0,0)\trans\in\Omega_B$, since this would reproduce the exact same probabilities on all simulated effects. 
Crucially, however, the simulation does not {\em require} such alternatives (cf.\ the Holevo projection). 
We will say more about coarse-graining processes and contextuality in \cref{sec:Decoherence}.

Embeddings have a rich mathematical structure, which we explore in the following lemmas.
These properties will become relevant in later sections of this article.

\begin{lemma}
\label{lem:Unital}
Let $\mathcal{A}=(A,\Omega_A,E_A)$ and $\mathcal{B}=(B,\Omega_B,E_B)$ be GPTs, and let $\varepsilon\geq 0$. If there exists an $\varepsilon$-embedding of $\mathcal{A}$ into $\mathcal{B}$, then there exists a $(2\varepsilon)$-embedding of $\mathcal{A}$ into $\mathcal{B}$ that is \emph{unital}.
\end{lemma}

\begin{lemma}
\label{lem:EmbeddingFacts}
For an exact embedding $\mathcal{A}\to\mathcal{B}$ with maps $\Phi:A\to B$ and $\Psi:A\conj \to B\conj$:
\begin{enumerate}[(i)]
\item $\Psi\conj\Phi=\mathbf{1}_A$, i.e.\ the dual of $\Psi$ is a left-inverse of $\Phi$, hence $\dim A \leq \dim B$. 
Likewise, $\Phi\conj\Psi = \mathbf{1}_{A\conj}$.
\item The map $P:=\Phi\Psi\conj: B \to B$ is a projection onto the image of $\Phi$ ($P^2 = P$ and $P(B) =\Phi(A)$). Furthermore, if the embedding is unital, we have $P u_B=u_B$, i.e.\ $P$ is unital too.
Similarly, $P\conj = \Psi\Phi\conj: B\conj \to B\conj$ is a projection onto the image of $\Psi$ ($P^{*2} = P\conj$ and $P^*(B\conj)=\Psi(A^*)$).
\end{enumerate}
\end{lemma}

\begin{lemma}
\label{lem:UnrestrictedFacts}
Let $\mathcal{A}$ be unrestricted.
For an exact embedding $\mathcal{A}\to\mathcal{B}$ with maps $\Phi:A\to B$ and $\Psi:A\conj \to B\conj$:
\begin{enumerate}[(i)]
\item $\Phi^*$, $\Psi^*$, $P :=\Phi\Psi\conj$ and $P\conj :=\Psi\Phi\conj$ are all positive maps,
\item $\Phi(A_+)=\Phi(A)\cap B_+=P(B_+)$.
\end{enumerate}
\end{lemma}
The proofs are given in \cref{app:EmbeddingProperties}.
For the complementary case of embeddings into infinite-dimensional classical systems, similar results are obtained in the forthcoming work of \citet{BarnumL20}.

\section{Embeddings into classical probability theory: \mbox{generalized noncontextuality}}
\label{sec:StandardContextuality}
In this section, we will briefly review Spekkens' generalized notion of contextuality~\cite{Spekkens05} and prove equivalence to ours in the special case of embeddings into (and simulations by) classical probability theory $\mathcal{C}_n$. 
We subsequently elucidate the relation between contextuality inequalities~\cite{KunjwalS15,MazurekPKRS16,KrishnaSW17,KunjwalS18} and our notion of approximate embeddings into $\mathcal{C}_n$.
\vspace*{1em}

\subsection{Equivalence of simulations and \mbox{ontological models}}
Recall from \cref{sec:PrepAndMeasure} that an operational theory describes a collection of procedures accessible in a laboratory, and specifies the probabilities $P(k|p,m)$ of obtaining outcome $k$ when measurement $m$ is performed after preparation $p$.
For any given operational theory, we can study \emph{ontological models} for that theory, and the properties of such models. 
As defined by \citet{Spekkens05}:

\emph{``An ontological model is an attempt to offer an explanation of the success of an operational theory by assuming that there exist physical systems that are the subject of the experiment. These systems are presumed to have attributes regardless of whether they are being subjected to experimental test, and regardless of what anyone knows about them. These attributes describe the real state of affairs of the system. Thus, a specification of which instance of each attribute applies at a given time we call the \emph{ontic state} of the system.''}

The ontic state of the given system will be denoted $\lambda$, and the set of all such states (formally, a measurable space) is $\Lambda$. 
Here, as in~\citet{SchmidSWKW21}, we restrict our attention to discrete $\Lambda$ for simplicity. 
An ontological model should reproduce the probabilistic predictions of the operational theory as follows:
To every preparation procedure $p$, there is a specific probability distribution $\mu_p$ over the ontic states, and to every measurement $m$ and outcome $k$, there is a specific function $\chi_{k,m}:\Lambda\to[0,1]$ with
 \begin{equation}
    P(k|p,m)=\sum_{\lambda\in\Lambda} \mu_p(\lambda)\chi_{k,m}(\lambda)
    \label{eqOntModel}
\end{equation}
 such that $\sum_{\lambda\in\Lambda}\mu_p(\lambda)=1$ and $\sum_k \chi_{k,m}(\lambda)=1$ for all $\lambda$. 

We can interpret an ontological model as specifying a way to reproduce the statistics of the operational theory in classical terms: the preparation $p$ amounts to sampling an ontic state $\lambda$ according to a certain probability distribution $\mu_p$, and the measurement procedure implements a (possibly noisy) classical measurement of that ontic state.
 We assume that ontological models are closed under statistical mixing, i.e.\ that we can implement one of two preparation procedures with 
some probability and obtain the corresponding convex combinations of the involved $\mu_p$. 
 We denote the resulting procedure by $q p+(1-q)p'$, where $0\leq q\leq 1$ is the probability of implementing $p$. 
Similar reasoning applies to measurements and their outcomes and their corresponding $\chi_{k,m}$.
We also assume that $\chi_{k,m}\equiv 0$ is a valid response function that describes an impossible outcome.

Recall furthermore from \cref{sec:PrepAndMeasure} that by identifying equivalence classes of preparation and measurement procedures, we can associate a GPT with an operational theory.
The notion of equivalence of procedures is also the main ingredient for Spekkens' definition of noncontextuality: 
 an ontological model of an operational theory is \emph{preparation--noncontextual} if equivalent preparations $p\sim p'$ yield identical distributions of ontic states, $\mu_p=\mu_{p'}$, and \emph{measurement--noncontextual} if equivalent outcome-measurement pairs~\cite{SchmidSWKW21} $(k,m)\sim (k',m')$ yield identical response functions, $\chi_{k,m}=\chi_{k',m'}$. 
The model is called \emph{noncontextual} if it is both preparation-- and measurement--noncontextual. 

It turns out that our notion of ``simulation by a classical GPT'' (special case $\mathcal{B}=\mathcal{C}_n$ of \cref{def:epssim}) is equivalent to that of an ontological model, and that the corresponding notions of contextuality are equivalent:
\begin{theorem}
\label{thm:Correspondence}
Every discrete ontological model of an operational theory defines an exact simulation of the corresponding GPT by some $\mathcal{C}_n$, and vice versa.
Moreover, the ontological model is preparation--noncontextual / measurement--noncontextual / noncontextual if and only if the corresponding simulation is preparation-univalent / measurement-univalent / univalent.
\end{theorem}
The proof is given in \cref{app:Correspondence}.
Essentially, the claim follows by associating each distinct distribution $\mu$ with a simulating state in $\mathcal{C}_n$, and each response function $\chi$ with a simulating effect.

This theorem implies a simple corollary that subsumes the main result of~\cite{SchmidSWKW21}: a GPT admits of a discrete ontological model (in the restricted sense of their definition, i.e.\ noncontextual) if and only if the GPT is simplex-embeddable (recall that the state spaces $\Omega_n$ of the classical GPTs $\mathcal{C}_n$ are simplices).
\begin{corollary}
An operational theory admits of a discrete noncontextual ontological model if and only if the corresponding GPT is embeddable into some $\mathcal{C}_n$.
\end{corollary}
This follows from Theorem~\ref{thm:Correspondence} because a univalent simulation is an embedding (see Lemma~\ref{lem:Linear} and \cref{def:EmbeddingMap}). 
Furthermore, our results on exact embeddings into quantum theory (in \cref{sec:QuantumExact} below) imply as a simple consequence a result that has also been found in~\cite{Shahandeh21,BarnumL20,SchmidSWKW21}:
\begin{corollary}
\label{col:Classical}
An unrestricted GPT can be embedded exactly into finite-dimensional classical theory, i.e.\ into some $\mathcal{C}_n$, if and only if it is classical.
\end{corollary}
A simple argument for this is given in Appendix~\ref{app:EmbeddingProperties}.

\subsection{Approximate embeddability and noncontextuality inequalities}
There has been a wave of recent interest on how contextuality (in the sense of~\citet{Spekkens05}) can be experimentally tested~\cite{KunjwalS15,MazurekPKRS16,KrishnaSW17,KunjwalS18}. 
This requires noncontextuality certificates that are robust to a certain amount of noise. 
One way to achieve this is via \emph{noncontextuality inequalities}, whose experimental violation (subject to certain assumptions~\cite{MazurekPKRS16}) rule out the existence of a noncontextual ontological model. 
We will now demonstrate that noncontextuality inequalities imply statements about the $\varepsilon$-embeddability of quantum theory (or other GPTs) into classical probability theory $\mathcal{C}_n$.

Consider the noncontextuality inequality derived by \citet{MazurekPKRS16}:
\begin{align}
A:=\frac{1}{6} \sum_{t\in\{1,2,3\}}\sum_{b\in\{0,1\}} P(b\,|\,p_{t,b}, m_t)\leq \frac{5}{6}.
\label{eq:ContIn}
\end{align}

Here, $p_{t,b}$ denotes six preparation procedures and $m_t$ three measurement procedures (with two possible outcomes $b\in\{0,1\}$) in an operational theory. 
By assumption, the three preparation procedures $p_t:=\frac{1}{2} p_{t,0}+\frac{1}{2} p_{t,1}$ (obtained by tossing a fair coin and implementing either $p_{t,0}$ or $p_{t,1}$) are operationally equivalent, i.e.\ statistically indistinguishable. 
Furthermore, $m=\frac {1}{3} m_1+\frac{1}{3}m_2+\frac{1}{3} m_3$ resembles a fair coin toss, i.e.\ yields outcomes $0$ or $1$ with equal probability regardless of the preparation.

\Citet{MazurekPKRS16} show that the existence of a noncontextual ontological model implies inequality~\eqref{eq:ContIn}. 
However, this inequality can be violated via preparations and measurements of a quantum bit, which admit a value of $A=1$. 
These preparations and measurements lie in an equatorial plane of the Bloch ball, and can hence be interpreted as procedures within quantum theory over the real numbers (i.e.\ as elements of a \emph{rebit}).

This contextuality inequality implies the nonexistence of an approximate embedding into classical probability theory:
\begin{example}
\label{ExRebit}
Let $\varepsilon<\frac{1}{6}$. 
Then the rebit (and thus, also the qubit) cannot be $\varepsilon$-embedded into any $\mathcal{C}_n$.	
\end{example}
\textit{Proof sketch}. Here we only summarize the proof; all the details are given in Appendix~\ref{SecProofRebit}. To the six preparation procedures, $p_{t,b}$, we associate six rebit states $\rho_{t,b}$, and to the outcomes $b$ of the measurements $m_t$, we associate the rebit effects $E_{t,b}$ such that $P(b'|p_{t,b},m_{t'})=\tr(\rho_{t,b} E_{{t'},b'})$, as in Ref.~\cite{MazurekPKRS16}. Consider any $\varepsilon$-embedding of the rebit into some $\mathcal{C}_n$. 
This defines classical states $\omega_{t,b}:=\Psi(\rho_{t,b})$ and effects $e_{t,b}:=\Phi(E_{t,b})$ such that $|(\omega_{t,b},e_{t',b'})-\tr(\rho_{t,b} E_{t',b'})|\leq\varepsilon$, and the linear maps $\Psi$ and $\Phi$ preserve the linear dependencies among the $\rho_{t,b}$ and among the $E_{t,b}$, i.e.\ the operational equivalences. 
But $\mathcal{C}_n$ certainly has a noncontextual ontological model (namely itself), hence
\begin{align}
\frac{5}{6} \geq \frac{1}{6} \sum_{t,b} (\omega_{t,b},e_{t,b})
\geq \frac{1}{6}\sum_{t,b}\left(\strut\tr(\rho_{t,b} E_{t,b})-\varepsilon\right)=1-\varepsilon.
\end{align}
Thus $\varepsilon\geq \frac 1 6$, and no $\varepsilon$-embedding is possible for any smaller value of $\varepsilon$.
\qed

Our arguments in \cref{sec:Experiment} below suggest that the converse is also true: proofs of nonexistence of $\varepsilon$-embeddings can in principle be used as experimental tests of contextuality. 
However, in this article, we focus on a different goal: rather than ruling out noncontextual ontological models, we exclude \emph{underlying quantum models} that would explain the experimental data by univalent simulation.

\section{Exact embeddings into \mbox{quantum theory}}
\label{sec:QuantumExact}
If we assume that the fundamental theory of nature is quantum, then all effective GPTs that we actually encounter in nature must be simulated by quantum theory. 
Thus, it is of particular interest to study embeddings into quantum theory. 
While we focus on embeddings into finite-dimensional quantum systems $\mathcal{Q}_n$, we will show in \cref{sec:Infinite} that exact embeddings into countably infinite-dimensional quantum systems $\mathcal{Q}_\infty$ can be approximated by almost-exact embeddings into $\mathcal{Q}_n$.

\subsection{Examples}
\label{sec:embed_examples}

\inlineheading{Example: Quantum theory.}
First, it can be useful to embed quantum theory of some finite dimension into quantum theory of a higher dimension.
For example, the quantum-error correcting Shor code~\cite{Shor95} maps a single logical qubit onto nine physical qubits ($\psi: \Csa{2} \to \bigotimes_{i=1}^9 \Csa{2}$) so as to allow for a random bit flip ($\sigma_x$) and/or phase flip ($\sigma_z$) on any of the nine physical qubits without affecting the encoded logical information.
Here both $\Psi$ and $\Phi$ take the form $X\mapsto V X V^\dagger$, where $V$ is an isometry, and thus satisfy the requirements of Lemma~\ref{lem:Linear} for $\varepsilon=0$.

More generally, identifying subspaces of larger quantum systems via the corresponding isometries defines nonunital exact embeddings. 
For example, $\mathcal{Q}_2$ (a qubit) is embedded into $\mathcal{Q}_3$ (a qutrit), by simply mapping the two qubit levels onto two of the three qutrit levels.

\inlineheading{Example: Classical theory.}
An important exact embedding is the inclusion of $n$-level classical probability theory within $n$-level quantum theory. 
There is a positive unital linear map $\Psi: \reals^n \to \Csa{n}$ -- onto the $n\times n$ diagonal matrices $\Psi: (p_1,\ldots,p_n)\trans\mapsto \rho := \sum_i p_i \ketbra{i}{i}$ for some choice of basis $\{\ket{i}\}_{i=1\ldots n}$. 
The same can be defined on effects, embedding them as diagonal POVM elements.
Here, the embedding is unital.

\inlineheading{Example: Real quantum theory.}
Real quantum theory is defined almost identically to complex quantum theory,
 except instead of using positive complex self-adjoint matrices $\Csap{n}$ to describe effects and states,
 real quantum theory uses the positive real symmetric matrices $\Rsap{n}$.
As each real symmetric matrix is also a complex self--adjoint matrix (with no imaginary elements), 
 real quantum theory can be trivially embedded by inclusion maps into complex quantum theory of the same dimension.

\inlineheading{Example: Spin-factors (``$d$-balls'').}
There are also more ``exotic'' GPTs that can be exactly embedded into finite-dimensional quantum theory.
For instance, the $d$-dimensional {\em spin-factor} models $\mathcal{B}_d=(\mathbb{R}^{d+1},\Omega_d,E_d)$, whose normalized states are given by $d$-dimensional balls $\Omega_d= \{ (1, \vec{x})\,\, |\,\, \lVert \vec{x} \rVert_2 \leq 1 \}\subseteq \reals\oplus\reals^d$ often arise as foil theories to quantum theory, generalizing the $3$-dimensional real ``Bloch ball'' representation of a qubit (see e.g.,\ \cite{PawlowskiW12,PaterekDB10,GarnerMD17,KrummM19}). 
The effect cone is $B_{d+}: = \{\left(n, \vec{x}\right)\trans \in \reals\oplus\reals^d \;|\; n\geq 0,\, \lVert \vec{x} \rVert_2 \leq n \}$ 
 and $\vec{u}_d=\left(1,0,\ldots 0 \right)\trans$. 
 This theory is defined to be unrestricted, and its states may also be associated with elements of $B_{d+} \subset \reals\oplus\reals^d$, with probabilities ascribed to states and effects via the inner product $\inn{( n_s, \vec{x}_s)}{(n_e, \vec{x}_e)} := n_1 n_2 + \vec{x}_s \cdot \vec{x}_e$ (where ``$\cdot$'' here is the usual Euclidean dot product). 
Since the cones of states and effects become identical under this inner product, the GPT is strongly self-dual.

A $d$-dimensional spin factor can be embedded into complex quantum theory via a map given by \citet{Tsirelson87} (as outlined by \citet{KleinmannOSW13}),
 specifically into the $m$-qubit system $\mathcal{Q}_{2^m}$ where $m = d/2$ if $d$ is even, and $m = \left(d+1\right)/2$ otherwise.
In particular, we define the matrices $\{\gamma_i\}_{i=1\ldots d}$ in $\Csa{2^m}$ as:
\begin{align}
\gamma_{i} := \begin{cases}
\underbrace{\sigma_z \otimes \ldots \otimes \sigma_z}_{\frac{i-1}{2}} \otimes \sigma_x \otimes \underbrace{\sigma_0 \otimes \ldots \otimes \sigma_0}_{m - \frac{i-1}{2}} 
& \mathrm{for~odd~} i, \\
\underbrace{\sigma_z \otimes \ldots \otimes \sigma_z}_{\frac{i}{2}-1} \otimes \sigma_y \otimes \underbrace{\sigma_0 \otimes \ldots \otimes \sigma_0}_{m - \frac{i}{2}}
& \mathrm{for~even~} i,
\end{cases}
\end{align}
where $\sigma_x, \sigma_y, \sigma_z$ are the $2\times2$ Pauli matrices, and $\sigma_0=\id_2$ is the $2\times2$ identity matrix.
Then, the embedding map $\Phi: \reals\oplus\reals^d \to \Csa{n}$ is specified (for $n\in\reals$, $\vec{x}\in \reals^d$) as:
\begin{align}
\Phi\left(n,\vec{x}\right)\trans:=  n\id_{2^m} + \sum_{i=1}^d x_i \gamma_i.
\end{align}
Likewise, the map $\Psi: \reals\oplus\reals^d \to \Csa{n}$ for embedding the spin-factor states is identically defined.
(See \cref{app:spin} for proof of positivity, and preservation of outcome probabilities). 
See also~\citet{BarnumGW20} for more details on this embedding and on the one below, and on the relation to Jordan algebras that we will turn to shortly.

\inlineheading{Example: Quaternionic quantum theory.}
The quaternions $\quat$ are a rank-$4$ associative, but noncommutative, algebra,
 admitting three imaginary units $i,j,k$ such that $ii = jj = kk = ijk = -1$.
Any $q\in\quat$ can be written $q = a + ib + jc + kd$ for $a,b,c,d\in\reals$,
 and the conjugate $q\conj$ is defined $q^* := a - ib - jc - kd$.
The set of $n\times n$ quaternionic self-adjoint matrices $\Hsa{n}$ is defined similarly to the complex Hermitian matrices: i.e.\ as the set of all matrices equal to their own transposed element-wise conjugate.
The cone of positive elements $\Hsap{n}$ is the subset of $\Hsa{n}$ with positive real eigenvalues.
{\em Quaternionic quantum theory}~\cite{Graydon11} is then defined to have effects and states both associated with elements of $\Hsap{n}$, 
 with outcome probabilities given via the self-dualizing inner product $\inn{\rho}{e} := \re{\tr\left(\rho e\right)}$ for $\rho, e\in\Hsap{n}$. 
The unit effect is the $n\!\times\!n$ identity matrix, and the normalized states accordingly have $\re{\tr\left(\rho\right)} = 1$.

Quaternionic quantum theory can be embedded into complex quantum theory of twice the dimension by expressing the elements of $\Hsa{n}$ as ``symplectic'' complex self-adjoint matrices.
Identifying the quaternionic unit $i$ with the complex imaginary unit $i$, any quaternionic self-adjoint matrix $Q\in\Hsap{n}$ can be written $Q = A+B j$ where $A$ and $B$ are $n\times n$ complex matrices, and such that the map $\Phi: \Hsap{n}\to\Csap{2n}$ on quaternionic effects:
\begin{align}
\label{eq:SympForm}
\Phi (Q):= \left(\begin{array}{cc}A & B \\ -B\conj & A\conj \end{array}\right),
\end{align}
together with a similarly defined $\Psi := \frac{1}{2} \Phi: \Hsap{n} \to \Csap{2n}$ on quaternionic states is an embedding into (complex) quantum theory.
(See \cref{app:quat} for proof of positivity, and preservation of outcome probabilities).

\inlineheading{Nonexample: The gbit.}
The Holevo simulation of the gbit by classical theory (and hence by quantum theory) was not an embedding, because it was preparation multivalent.
Could there be another way of embedding such a state space into $\mathcal{Q}_n$ for some choice of $n$?
We can see, somewhat straightforwardly, that the answer must be ``no'' for the gbit.

Suppose that there is an exact embedding of a gbit (recall \cref{fig:GPTCones}) into some $\mathcal{Q}_n$. 
Then the quantum states $\rho_{ij}:=\Psi(\alpha_{ij})$ ($i,j\in\{+,-\}$) are pairwise perfectly distinguishable. 
Thus, $\langle\rho_{ij},\rho_{kl}\rangle_{\rm HS}=\delta_{ik}\delta_{jl}$, where $\langle X,Y\rangle_{\rm HS}:={\rm tr}(X Y)$ is the Hilbert-Schmidt inner product on  $\mathbf{H}_n(\mathbb{C})$. 
Hence, these states span a four-dimensional subspace of $\mathbf{H}_n(\mathbb{C})$, which is in contradiction to the three-dimensionality of the span of the $\alpha_{ij}$. 
Thus, by contradiction, the gbit cannot be embedded exactly into quantum theory. 
We address the question of approximate embeddings in \cref{sec:ExGbit} below.

\subsection{All exact embeddings of unrestricted GPTs}
\label{SubsecAllExact}
Is there some deeper, structural reason, why gbits should fail to embed, while classical theory and the spin factors succeed?
It turns out -- via results by \citet{EffrosS79} (see also \citet{Werner82} and \citet{Idel13}) -- that the unrestricted GPTs that can be exactly embedded into quantum theory are rather limited and can be characterized algebraically, see Theorem~\ref{thm:Simulate} below. 
For completeness, we give a simplified version of the proof as it applies to our specific setting.

We begin by introducing a formal way of saying that an embedding into quantum theory should not be ``unnecessarily large'':

\begin{definition}[Minimal embedding]
\label{def:Min}
An $\varepsilon$-embedding of a GPT $\mathcal{A}$ into $n$-dimensional quantum theory $\mathcal{Q}_n$ is \textbf{minimal} if there does not exist any $m<n$ such that $\mathcal{A}$ can be $\varepsilon$-embedded into $\mathcal{Q}_m$.
\end{definition}
When we embed into quantum models, we may always choose the smallest possible Hilbert space dimension. 
Thus, we will restrict our attention to minimal embeddings. 
Consequently:
\begin{lemma}
\label{lem:QMinimal}
If an $\varepsilon$-embedding of a GPT $\mathcal{A}$ into $\mathcal{Q}_n$ is minimal, then there exists some state $\omega\in\Omega_A$ such that the quantum state $\Psi\!\left(\omega\right)$ has full rank.
\end{lemma}
\begin{proof}
Let $\omega\in \Omega_A$ such that $m:={\rm rank}(\Psi(\omega))$ is maximal, and suppose that $m<n$. 
Let $S:={\rm supp}(\Psi(\omega))$ (an $m$-dimensional subspace of $\comp^n$), and suppose there is some $\rho\in\Omega_A$ with ${\rm supp}(\Psi(\rho))\not\subseteq S$. 
Taking the orthocomplement of ${\rm ker}(\Psi(\frac 1 2 \omega+\frac 1 2\rho))\subseteq{\rm ker}(\Psi(\omega))$ yields ${\rm supp}(\Psi(\frac{1}{2} \omega+\frac{1}{2} \rho))\supseteq {\rm supp}(\Psi(\omega))=S$, but similarly, ${\rm supp}(\Psi(\frac{1}{2} \omega+\frac{1}{2} \rho))\supseteq {\rm supp}(\Psi(\rho))\ni x$ with $x\not\in S$. 
Hence ${\rm rank}(\Psi(\frac{1}{2} \omega+\frac{1}{2} \rho))> m$ which is a contradiction. 
Thus, ${\rm supp}(\Psi(\rho))\subseteq S$ for all $\rho\in \Omega_A$.

Now choose an arbitrary orthonormal basis $|s_1\rangle,\ldots,|s_m\rangle$ of $S$, and define the map $V:\mathbb{C}^n\to\mathbb{C}^m$ as $V:=\sum_{j=1}^m|j\rangle\langle s_j|$, where $\{|j\rangle\}_{j=1}^m$ is an orthonormal basis of $\mathbb{C}^m$. It follows that $V^\dagger V=\Pi_S$, the orthogonal projector onto $S$. Define $\tilde\Phi:A\to\mathbf{H}_m(\mathbb{C})$ via $\tilde\Phi(e):=V\Phi(e)V^\dagger$ and $\tilde\Psi:A^*\to\mathbf{H}_m(\mathbb{C})$ as $\tilde\Psi(\omega):=V\Psi(\omega)V^\dagger$. Then
\begin{align}
\tr(\tilde\Psi(\omega)\tilde\Phi(e))=\tr(\Pi_S \Psi(\omega)\Pi_S\Phi(e))=\tr(\Psi(\omega)\Phi(e)),\nonumber
\end{align}
and the maps $\tilde\Phi$ and $\tilde\Psi$ define an $\varepsilon$-embedding of $\mathcal{A}$ into $\mathcal{Q}_m$. 
This contradicts the assumption that the $\varepsilon$-embedding of $\mathcal{A}$ into $\mathcal{Q}_n$ is minimal, and hence $m=n$, and so $\Psi(\omega)$ has full rank.
\end{proof}
The following lemmas concern the interplay of the embedding projectors with the Jordan product $x\sjord{}y:=\frac{1}{2}(xy\!+\!yx)$ on the vector space of Hermitian matrices of $\mathcal{Q}_n$.
\begin{lemma}
\label{lem:EffSt}
For every minimal exact unital embedding of an unrestricted GPT $\mathcal{A}$ into finite-dimensional quantum theory $\mathcal{Q}_n=:\mathcal{B}$, the corresponding projector $P=\Phi\Psi\conj$ satisfies
\begin{align}
   P\!\left(x\sjord y\right)=x\sjord P\left(y\right) \mbox{ for all }x\in\Phi(A),\, y\in B.
\end{align}
Hence, $\Phi(A)\equiv P(B)$ is closed under the Jordan product $\sjord$, and $\left(P\!\left(B\right),\, \sjord\right)$ is a special Euclidean Jordan algebra.
\end{lemma}
\begin{proof}
First, we show that all $x\in\Phi(A)$ satisfy $P(x^2)=x^2$. Due to Lemma~\ref{lem:QMinimal}, there exists some full-rank fixed state$\rho=\Psi(\omega)$, and due to Lemma~\ref{lem:EmbeddingFacts}, we have $P^*(\rho)=\Psi\Phi^*\Psi(\omega)=\Psi(\omega)=\rho$; similar argumentation shows $P(x)=x$.
Hence $\tr[P(x^2)\rho]=\tr[x^2 P^*(\rho)]=\tr(x^2\rho)$,
 such that $\tr(\Delta\rho)=0$ for $\Delta:=P(x^2)-x^2$. 
According to Lemma~\ref{lem:UnrestrictedFacts}, $P$ is positive, therefore Kadison's inequality~\cite{Kadison52} implies $P(z^2)\geq P(z)^2$ for all $z$, and thus $\Delta\geq 0$.
Thus, $\tr(\Delta\rho)=0$ is only possible if $\Delta=0$ since $\rho$ is positive definite.

Now let $x\in\Phi(A)$, $y\in B$, and $t\in\mathbb{R}$ be arbitrary, and set $z:=tx+y$. We thus have $x=P(x)$ and $x^2=P(x^2)$.
Since $P$ is positive (Lemma~\ref{lem:UnrestrictedFacts}) and unital (Lemma~\ref{lem:EmbeddingFacts}), Kadison's inequality gives
$2t P(x\sjord y)+P(y^2)\geq 2 t x\sjord P(y)+P(y)^2$
for all $t\in\mathbb{R}$. 
But if $v=v^\dagger$ and $w=w^\dagger$ such that $tv+w\geq 0$ for all $t\in\mathbb{R}$, then $v=0$ (to see this, multiply from left and right by eigenvectors of $v$). 
Thus, the terms linear in $t$ must be equal, and so $P(x\sjord y)=x\sjord P(y)$.

If $x,y\in P(B)$ then $x\sjord y=x\sjord P(y)=P(x\sjord y)\in P(B)$, and hence $\left(P\!\left(B\right), \, \sjord\right)$ is a Jordan subalgebra of $\Csa{n}$, inheriting the properties of being special and Euclidean from $\Csa{n}$.
\end{proof}
Next we show that the image of the quantum effect cone under the positive projection $P$ is the cone of squares of the corresponding Jordan algebra:
\begin{lemma}
\label{lem:ConeOfSquares}
For every minimal exact unital embedding of an unrestricted GPT $\mathcal{A}$ into finite-dimensional quantum theory $\mathcal{Q}_n=:\mathcal{B}$, we have
\begin{align}
P\!\left(B_+\right) & = \{x^2\,\,|\,\,x\in P\!\left(B\right)\}.
\end{align}
\end{lemma}
\begin{proof}
The right-hand side equals the cone of squares $\mathcal{J}_+$ of $\left(P\!\left(B\right), \sjord\right)$ due to Lemma~\ref{lem:EffSt}.
To show $\mathcal{J}_+\subseteq P(B_+)$, let $y:=x^2$ with $x\in P(B)$. 
Then $0\leq y=x\sjord P(x)=P(x\sjord x)=P(y)$ (using Lemma~\ref{lem:EffSt}), and thus $y\in P(B_+)$. 

Meanwhile, using $\inn{x}{y} = \tr(xy)$ to identify $B$ with $B\conj$, we have $\langle a\sjord b,c\rangle=\langle a,b\sjord c\rangle$ for all $a,b,c$, and in particular for all $a,b,c\in P(B)$. 
Consequently, the cone $\mathcal{J}_+$ is self-dual under this inner product~\cite[III.2]{FarautK94} (i.e.,\ $\mathcal{J}_+=\mathcal{J}_+\conj$).
Let $y\in P(B_+)$.
Then, for all $x\in P(B)$, $\inn{x^2}{y} = \tr(x^2 y)\geq 0$ since $x^2\geq 0$ and $y\geq 0$, and thus $y\in \mathcal{J}_+^* \equiv \mathcal{J}_+$, and thus $P(B_+)\subseteq \mathcal{J}_+$.
Hence, $P\!\left(B_+\right) = \mathcal{J_+} = \{x^2\,\,|\,\,x\in P\!\left(B\right)\}$.
\end{proof}
This allows us to classify all unrestricted GPTs that have an exact univalent simulation by finite-dimensional quantum theory:
\begin{theorem}
\label{thm:Simulate}
An unrestricted GPT can be exactly embedded into finite-dimensional quantum theory if and only if it corresponds to a {\em special Euclidean Jordan algebra}.
\end{theorem}
\begin{proof}
For the {\em only if} direction, we can choose a minimal embedding $\Phi: A\to \Csa{n}$, and Lemma~\ref{lem:Unital} shows that we can choose it to be unital.
From Lemma~\ref{lem:ConeOfSquares}, it follows that $\Phi(A_+)=\{x^2\,\,|\,\,x\in\Phi(A)\}$, hence $\mathcal{A}$ is order-isomorphic to the GPT of the special Euclidean Jordan algebra $(P(B),\sjord)$.
For the {\em if} direction, such algebras can be exhaustively listed~\cite{JordanvNW34},
 and appropriate embeddings exist for these~\cite{Tsirelson87,BarnumGW20,KleinmannOSW13} and their direct sums.
\end{proof}

In other words, the examples in \cref{sec:embed_examples} and their direct sums are in fact the {\em only} unrestricted GPTs that can be exactly embedded into quantum theory.

\subsection{Decoherence, noise, and coarse-grainings}
\label{sec:Decoherence}
Suppose we can prepare any state and measure any effect of $n$-level quantum theory $\mathcal{Q}_n=(\mathbf{H}_n(\mathbb{C}),\Omega_n,E_n)$, but there is some unavoidable noise, described by a trace-preserving quantum channel $\mathcal{N}$, happening in between the preparation and the measurement. 
Let us assume that $\mathcal{N}$ is ``nonsingular'', in the sense that its image has full dimension, i.e.\ $\mathcal{N}(\mathbf{H}_n(\mathbb{C}))=\mathbf{H}_n(\mathbb{C})$. 
The states and effects in this situation will be described by an effective GPT
\begin{align}
   \mathcal{Q}_n^{\mathcal{N}}:=(\mathbf{H}_n(\mathbb{C}),\mathcal{N}(\Omega_n),E_n).
\end{align}
That is, the effective set of states is not $\Omega_n$, but the ``noisy'' set of states $\mathcal{N}(\Omega_n)$. 
Since we assume that this set of states still spans all of $\mathbf{H}_n(\mathbb{C})$, all effects in $E_n$ can still be statistically distinguished from each other by the values they take on the states, which is necessary for $\mathcal{Q}_n^{\mathcal{N}}$ to be a valid GPT. 
\begin{lemma}
\label{lem:Restricted}
Quantum theory under nonsingular non\-unitary noise $\mathcal{N}$, i.e.\ $\mathcal{Q}_n^{\mathcal{N}}$, is a \emph{restricted} GPT which can be embedded exactly into $\mathcal{Q}_n$.
\begin{proof}
Choosing $\Phi$ and $\Psi$ as the identity maps defines the corresponding embedding. If $\mathcal{D}$ is not unitary, then $\mathcal{D}(\Omega_n)\subsetneq \Omega_n$, and thus the resulting set of states is not maximal given the set of effects, i.e.\ $\mathcal{Q}_n^{\mathcal{N}}$ is restricted.
\end{proof}
\end{lemma}
For nonsingular nonunitary qubit channels $\mathcal{N}$, the Bloch ball of states is effectively mapped to a smaller ellipsoid inside the ball~\cite{NielsenC00}, which represents the set of states of the resulting GPT $\mathcal{Q}_2^{\mathcal{D}}$. 
Lemma~\ref{lem:Restricted} tells us that these naturally occurring GPTs admit of univalent quantum simulations --- in this sense, noise does not introduce the necessity of multivalence for explaining the statistics.

We do not currently know whether {\em all} singular quantum channels (i.e.\ channels whose image is a proper subspace of $\mathbf{H}_n(\mathbb{C})$) lead to effective GPTs that are embeddable. 
However, one special class of channels of particular interest does: \emph{complete decoherence processes} and \emph{coarse-graining processes} $\mathcal{D}$.
Intuitively, complete decoherence is a relaxation process that affects a physical system in the long time limit (in practice, often after a very short time) such that ``decohering twice is the same as decohering once'', i.e.\ $\mathcal{D}^2=\mathcal{D}$. 
For example, the process that removes the off-diagonal elements of a density matrix is of this form.

Similarly, coarse-graining processes are described by maps of this kind. 
Recall the example of \cref{eq:Coarsegraining} for the case of classical probability theory: we can think of the bit $\mathcal{A}$ as arising from two bits $\mathcal{B}$ by the map 
\begin{align}
P:=\Phi\Psi^*=\left(\begin{array}{cccc} 1/2 & 1/2 & 0 & 0 \\ 1/2 & 1/2 & 0 & 0 \\ 0 & 0 & 1/2 & 1/2 \\ 0 & 0 & 1/2 & 1/2\end{array}\right),
\end{align}
which randomizes the four configurations in groups of two, and $P^2=P$.

Complete decoherence and coarse-graining processes cannot introduce multivalence: they lead to effective GPTs which are exactly embeddable, i.e.\ univalently simulable. 
We can prove this in a more general case where the underlying GPT $\mathcal{A}=(A,\Omega_A,E_A)$ is not necessarily quantum.  
In this case, a complete decoherence or coarse-graining process is a linear map that is positive and normalization-preserving ($\mathcal{D}(\Omega_A)\subset\Omega_A$), has positive adjoint ($\mathcal{D}^*(E_A)\subset E_A$ such that valid effects are mapped to valid effects in the Heisenberg picture), and is idempotent ($\mathcal{D}^2=\mathcal{D}$)~\cite{RichensSS17}.
Then we get:
\begin{lemma}
\label{LemDecoherence}
Let $\mathcal{A}=(A,\Omega_A,E_A)$ be any GPT (for example quantum theory), and $\mathcal{D}:A^*\to A^*$ a complete decoherence or coarse-graining process. 
The resulting effective GPT
\begin{align}
   \mathcal{A}^{\mathcal{D}}=(\mathcal{D}^*(A),\mathcal{D}(\Omega_A),\mathcal{D}^*(E_A))
   \label{DefDecoherence}
\end{align}
can be embedded exactly into $\mathcal{A}$, i.e.\ admits of a univalent simulation by $\mathcal{A}$. Moreover, if $\mathcal{A}$ is unrestricted then $\mathcal{A}^{\mathcal{D}}$ is also unrestricted.
\begin{proof}
First we have to check that the GPT $\mathcal{A}^{\mathcal{D}}$ is well-defined: we need to be able to understand $\mathcal{D}(\Omega_A)$ as a subset of the dual space of $\mathcal{D}^*(A)$. For this, it is sufficient to show that $\mathcal{D}^*(A)$ can naturally be understood as the dual space of $\mathcal{D}(A^*)$. Now, to every linear functional on $\mathcal{D}(A^*)$, there is some (in general non-unique) vector $e_A\in A$ such that the functional equals $\omega'_A\mapsto (\omega'_A,e_A)$. Using idempotence of $\mathcal{D}$, we obtain
\begin{align}
   (\omega'_A,e_A)=(\mathcal{D}(\omega'_A),e_A)=(\omega'_A,\mathcal{D}^*(e_A)).
\end{align}
Hence the vectors in $\mathcal{D}^*(A)$ reproduce all linear functionals on $\mathcal{D}(A^*)$ under the pairing $(\cdot,\cdot)$. Thus, $\mathcal{A}^{\mathcal{D}}$ is well-defined.
We can now simply define $\Phi:\mathcal{D}(A^*)\to A$ and $\Psi:\mathcal{D}(A^*)\to A^*$ as the inclusion maps. Since $\mathcal{D}$ and $\mathcal{D}^*$ are positive and normalization-preserving, so are $\Psi$ and $\Phi$, and $(\Psi(\omega'_A),\Phi(e'_A))=(\omega'_A,e'_A)$ for all $\omega'_A\in\mathcal{D}(\Omega_A),e'_A\in\mathcal{D}^*(E_A)$ by construction.

Finally, suppose that $\mathcal{A}$ is unrestricted, and that $e\in\mathcal{D}^*(A)$ is any vector such that $(\omega_A^{\mathcal{D}},e)\geq 0$ for all $\omega_A^{\mathcal{D}}\in\mathcal{D}(\Omega_A)$. 
Then, we have for all $\omega_A\in\Omega_A$:
\begin{align}
   0\leq (\mathcal{D}(\omega_A),e)=(\omega_A,\mathcal{D}^*(e))=(\omega_A,e),
\end{align}
hence $e\in A_+$, and thus $e=\mathcal{D}^*(e)\in\mathcal{D}^*(A_+)$, which is the effect cone of $\mathcal{A}^{\mathcal{D}}$.
\end{proof}
\end{lemma}
Lemmas~\ref{lem:Restricted} and~\ref{LemDecoherence}, together with the earlier Lemma~\ref{LemIterate} on the transitivity of embeddings, show that multivalence, our generalization of contextuality, cannot result from a natural class of physical processes. 
This hints at a \emph{resource-theoretic} perspective on this generalization of contextuality, an idea that has already been successfully explored for standard notions of contextuality~\cite{HowardWVE14,BermejoVegaDBOR17,Amaral19}.

In the case of quantum theory, the projections $P$ that have appeared in \cref{SubsecAllExact} (for example, in Lemma~\ref{lem:EffSt}) are complete decoherence processes, yielding the special Euclidean Jordan algebras $\mathcal{J}$ as the resulting effective GPTs $\mathcal{Q}_n^{\mathcal{D}}$ for $\mathcal{D}=P^*$. 
However, only some of them are actually physically realizable, i.e.\ completely positive~\cite{Stinespring55,Choi75,NielsenC00}: those that appear in the case that $\mathcal{J}$ corresponds to standard complex quantum theory with superselection rules (which includes the case of classical probability theory $\mathcal{C}_n$). 
Namely, these are the maps that remove the off-diagonal elements of a density matrix (in the case of $\mathcal{C}_n$) or project onto block matrices (pinching maps~\cite{Bhatia97}). 
In the more general context of operator algebras, these maps are known as \emph{conditional expectations}~\cite{Kadison04}.

For all other special Euclidean Jordan algebras $\mathcal{J}$, the corresponding projector $P$ is positive but not completely positive~\cite{Idel13}.
 And there can be no other completely positive processes that yield the GPTs of such Jordan algebras since the fixed-point sets of quantum channels are known to be ${}\conj$-subalgebras of $M_n(\mathbb{C})$, and thus isomorphic to standard complex quantum theory with superselection rules~\cite[Thm.\ 6.12]{Wolf12} (see also~\cite{Asias02}). 
 This proves the following:
\begin{corollary}
The only GPTs which can result from a physically realizable complete decoherence or coarse-graining process from quantum theory are classical probability theory and standard complex quantum theory with superselection rules.
\end{corollary}
This would hint to us why -- even though the exotic quaternionic, real and spin factor GPTs are in principle supported by standard complex quantum theory -- we rarely encounter them in nature.

\subsection{Embeddings into infinite-dimensional \mbox{quantum theory}}
\label{sec:Infinite}
To discuss embeddability into an infinite-dimensional quantum system $\mathcal{Q}_\infty$, corresponding to a separable Hilbert space $\mathcal{H}$, we must first address some topological subtleties of infinite--dimensional quantum theory. 
Unlike $\mathcal{Q}_n$ for finite $n$, the sets of unnormalized states and effects of $\mathcal{Q}_\infty$ cannot be treated as identical.
Rather, the effects of $\mathcal{Q}_\infty$ are elements of the bounded operators $\mathcal{B}(\mathcal{H})$, but the states are trace-class operators, i.e.\ elements of $\mathcal{T}(\mathcal{H})=\{T\in\mathcal{B}(\mathcal{H})\,\,|\,\,{\rm tr}(|T|)<\infty\}$. 
In particular, they are elements of the subsets of self-adjoint elements of both spaces which we denote $\mathcal{T}_{\rm sa}(\mathcal{H})$ and $\mathcal{B}_{\rm sa}(\mathcal{H})$. 
Thus, the linear maps which define an $\varepsilon$-embedding are of the form $\Psi:A^*\to\mathcal{T}_{\rm sa}(\mathcal{H})$ and $\Phi:A\to\mathcal{B}_{\rm sa}(\mathcal{H})$. As we always assume that $A$ is finite-dimensional, both maps are automatically continuous (with respect to the norms $\|T\|_1={\rm tr}\sqrt{T^\dagger T}$ on $\mathcal{T}(\mathcal{H})$ and $\|T\|_\infty:=\sup_{\|\psi\|=1}\|T\psi\|$ on $\mathcal{B}(\mathcal{H})$).
\begin{lemma}
\label{lem:Infinity}
Let $\mathcal{A}$ be a GPT that can be $\varepsilon$-embedded into $\mathcal{Q}_\infty$ for some $\varepsilon\geq 0$.
Then, for every $\delta>0$, there exists some $n\in\mathbb{N}$ such that $\mathcal{A}$ can be $(\varepsilon+\delta)$-embedded into $\mathcal{Q}_n$.
\end{lemma}
\begin{proof}
Choose an arbitrary orthonormal basis $\{|i'\rangle\}_{i\in\mathbb{N}}$ of $\mathcal{H}$. 
For every $n\in\mathbb{N}$, define $P_n:\mathcal{H}\to\mathbb{C}^{n+1}$ as $P_n:=\sum_{i=1}^n |i\rangle\langle i'|$, where $\{|i\rangle\}_{1\leq i \leq n+1}$ is an arbitrary orthonormal basis of $\mathbb{C}^{n+1}$. This is a linear map, and an elementary calculation shows $\|P_n\psi\|^2\leq \|\psi\|^2$, hence $P_n$ is bounded. Set $\Pi_n:=P_n^\dagger P_n=\sum_{j=1}^n|j'\rangle\langle j'|$, and define $\tilde\Psi:\mathcal{T}_{\rm sa}(\mathcal{H})\to \mathbf{H}_{n+1}(\mathbb{C})$ via
\begin{align}
   \tilde\Psi(\rho) & :=P_n \rho P_n^\dagger + \tr[(\mathbf{1}-\Pi_n)\rho]|n+1\rangle\langle n+1|,
\end{align}
and $\Tilde\Phi:\mathcal{B}_{\rm sa}(\mathcal{H})\to\mathbf{H}_{n+1}(\mathbb{C})$ as $\tilde\Phi(E):=P_n E P_n^\dagger$. 
Finally, consider $f_n:\Omega_A\to\mathbb{R}$, defined as $f_n(\omega):=\tr(\Pi_n \Psi(\omega)\Pi_n)$. 
It is not difficult to check that $f_{n+1}(\omega)\geq f_n(\omega)$ and $\lim_{n\to\infty}f_n(\omega)=\tr\Psi(\omega)=1$.
 Since $\Omega_A$ is compact, Dini's theorem implies that the convergence is uniform, i.e.\ there is a sequence $(\varepsilon_n)_{n\in\mathbb{N}}$ of nonnegative real numbers with $\varepsilon_n\stackrel{n\to\infty}\longrightarrow 0$ and $f_n(\omega)\geq 1-\varepsilon_n$ for all $\omega\in\Omega_A$. 
 But then, the gentle measurement lemma~\cite{Winter99,Wilde13} implies that
\begin{align}
   \|\Psi(\omega)-\Pi_n\Psi(\omega)\Pi_n\|_1\leq 2\sqrt{\varepsilon_n}\mbox{ for all }\omega\in\Omega_A.
\end{align}
A straightforward calculation shows that $\tr[\tilde\Psi(\rho)\tilde\Phi(E)]=\tr(\Pi_n \rho \Pi_n E)$ for all $\rho\in\mathcal{T}_{\rm sa}(\mathcal{H})$ and $E\in\mathcal{B}_{\rm sa}(\mathcal{H})$. Thus
\begin{align}
\Delta&:= |\tr(\tilde\Psi\circ\Psi(\omega)\cdot\tilde\Phi\circ\Phi(e))-(\omega,e)| \nonumber\\
&\leq|\tr(\tilde\Psi\circ\Psi(\omega)\cdot\tilde\Phi\circ\Phi(e))-\tr(\Psi(\omega)\Phi(e))|\nonumber\\
& \quad + |\tr(\Psi(\omega)\Phi(e))-(\omega,e)|\nonumber\\
&\leq |\tr(\Pi_n \Psi(\omega)\Pi_n\Phi(e))-\tr(\Psi(\omega)\Phi(e))|+\varepsilon\nonumber\\
&\leq \|\Pi_n \Psi(\omega)\Pi_n-\Psi(\omega)\|_1\cdot \|\Phi(e)\|_\infty + \varepsilon\nonumber\\
&\leq 2\sqrt{\varepsilon_n}+\varepsilon.
\end{align}
Thus, the maps $\tilde\Psi\circ\Psi:A^*\to\mathbf{H}_{n+1}(\mathbb{C})$ and $\tilde\Phi\circ\Phi:A\to\mathbf{H}_{n+1}(\mathbb{C})$ define a $(2\sqrt{\varepsilon_n}+\varepsilon)$-embedding of $\mathcal{A}$ into $\mathcal{Q}_{n+1}$ (it is easy to see that they also map the elements of $E_A$ to valid quantum effects and the elements of $\Omega_A$ to normalized quantum states). 
For every fixed $\delta>0$, we can choose $n$ large enough such that $2\sqrt{\varepsilon_n}\leq\delta$, and the claim follows.
\end{proof}
In other words: if a GPT can be embedded into a countably infinite-dimensional quantum system (either exactly or approximately), with a bit of additional error, it can be embedded into a finite-dimensional quantum system.

\section{Approximate embeddings into quantum theory}
\label{sec:QuantumApprox}
Theorem~\ref{thm:Simulate} gives strong constraints on what can be {\em exactly} embedded into quantum theory
 -- but how robust is this to noise and error?
In the following, we shall consider how to rule out that a GPT can be embedded -- even approximately -- into quantum theory.
In this section, we will discuss the mathematical criteria on a theory's embeddability, before describing more explicity in \cref{sec:Experiment} how this is related to experimental tests of quantum mechanics.

\subsection{Example: The gbit}
\label{sec:ExGbit}
We have seen (in \cref{sec:embed_examples}) that there are no {\em exact} embeddings of the gbit into quantum theory -- but what if we are willing to tolerate a degree of error in our embedding? 
What is the minimum such error we must allow -- even if we admit embeddings into quantum systems of arbitrarily high dimension?

Recall the gbit effects (\cref{eq:gbit_effects}) and states (\cref{eq:gbit_states}),
 and suppose that for some $\varepsilon>0$ there is an $\varepsilon$-embedding into finite-dimensional quantum theory, such that $\rho_{ij} := \Psi\left(\alpha_{ij}\right)$ for $i,j\in\{+,-\}$,
 and $E_x := \Phi\left(e_{x+}\right)$ and $E_z := \Phi\left(e_{z+}\right)$.
Although for the gbit, $\alpha_{ij}(e_k)$ (for $k={x+,x-,z+,z-}$) evaluates to either $0$ or $1$, for an $\varepsilon$--embedding, we tolerate up to $\varepsilon$ of error.
That is, for the quantum--embedded gbit states and effects:
\begin{align}
\tr\left(E_x \rho_{++}\right) \geq 1-\varepsilon, &\qquad \tr\left(E_x \rho_{+-}\right) \geq 1-\varepsilon, \nonumber \\
\tr\left(E_x \rho_{-+}\right) \leq \varepsilon, &\qquad \tr\left(E_x \rho_{--}\right) \leq \varepsilon,
\end{align}
with similar inequalities for $E_z$.

\begin{figure}[tbh]
\begin{centering}
\includegraphics[width=0.35\textwidth]{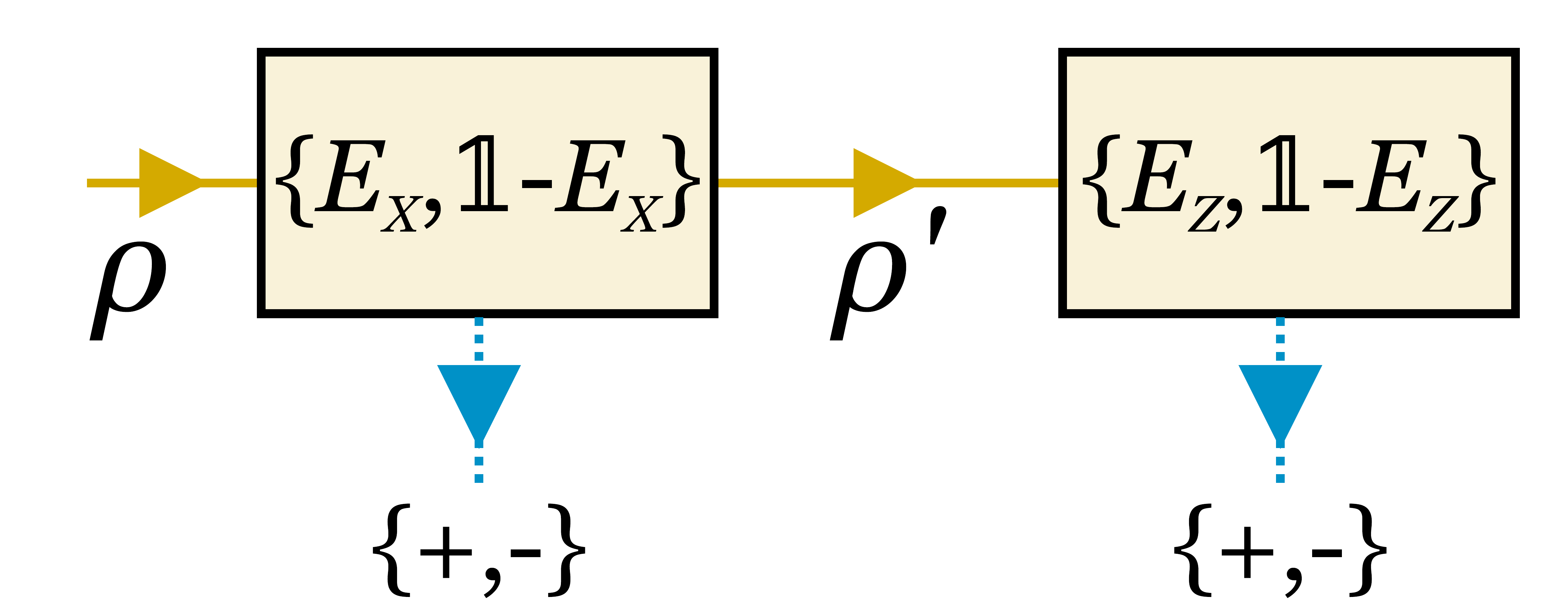}
\caption{%
\label{fig:device}
\caphead{\mbox{Quantum device $\mathcal{D}$.}} 
State $\rho$ is first measured by $\{E_x, \id-E_x\}$ with respective outcomes $+$ and $-$ and post-measurement state $\rho'$. 
Then $\rho'$ is measured by $\{E_z, \id-E_z\}$ also with respective outcomes $+$ and $-$. The required output statistics of this device on embeddings of a gbit allow us to bound the accuracy with which a gbit can be embedded into quantum theory.
}
\end{centering}
\end{figure}

Consider now the quantum device $\mathcal{D}$ (\cref{fig:device}), where quantum input state $\rho$ is first measured with the POVM $\{E_x, \id - E_x\}$ (with outcomes $+$ and $-$ respectively), yielding post-measurement state $\rho'=\sqrt{E}\rho\sqrt{E}/\tr(E\rho)$, where $E=E_x$ if the outcome is $+$ and $\id-E_x$ otherwise. 
Subsequently, $\rho'$ is measured with the POVM $\{E_z, \id-E_z\}$, also with respective outcomes $+$ and $-$.
First, consider when $\rho_{++}$ is input to $\mathcal{D}$.
With probability $P_1(+|\rho_{++}) \geq 1-\varepsilon$, the first outcome is $+$.
From the gentle measurement lemma~\cite{Winter99,Wilde13}, one can bound the change in post-measurement state for the case that outcome $+$ is obtained, namely $\|\rho_{++} - \rho'_{++} \|_1 \leq 2\sqrt{\varepsilon}$, where $\|\cdot\|_1$ is twice the trace distance, and hence:
 \begin{align}
| \tr\left(\rho'_{++} E_z\right) -  \tr\left(\rho_{++} E_z\right) | \leq \sqrt{\varepsilon}.
\end{align}
Thus, the joint probability of outcome $++$ from $\mathcal{D}$ is:
\begin{align}
P(++|\rho_{++}) & = P_1(+|\rho_{++}) P_2(+|\rho'_{++}) \nonumber \\
& \geq \left(1-\varepsilon\right) \tr\left(\rho'_{++} E_z\right) \nonumber\\
& \geq \left(1-\varepsilon\right)\left(1-\varepsilon-\sqrt{\varepsilon}\right)
\end{align}
By equivalent logic, $P(ij|\rho_{ij}) \geq \left(1-\varepsilon\right)\left(1-\varepsilon-\sqrt{\varepsilon}\right)$ for the other $i,j\in\{+,-\}$.

Suppose we input the state $\sigma := \Psi\left(\alpha'\right)$ into $\mathcal{D}$, 
 where $\alpha' := \frac{1}{2}\left(\alpha_{++} + \alpha_{--}\right) = \frac{1}{2}\left(\alpha_{+-} + \alpha_{-+}\right)$ is the state in the center of the gbit's square state space.
Then we can calculate the expected behaviour in two ways: 
 Either we use the decomposition $\sigma= \frac{1}{2}\left(\rho_{++} + \rho_{--}\right)$, such that
\begin{align}
P(++|\sigma) \geq \frac{1}{2}P(++|\rho_{++})\geq\frac{1-\varepsilon}{2}\left(1\!-\!\varepsilon\!-\!\sqrt{\varepsilon}\right);
\label{eq:ppp_lower}
\end{align}
 or we use $\sigma= \frac{1}{2}\left(\rho_{-+} + \rho_{+-}\right)$, such that
\begin{align}
P(++|\sigma) & = \frac{1}{2}P(++|\rho_{-+}) + \frac{1}{2}P(++|\rho_{+-}) \nonumber \\
&\leq \frac{1}{2} P_1(+|\rho_{-+})+\frac{1}{2} P_1(+|\rho_{+-})P_2(+|\rho'_{+-})\nonumber\\
& = \frac{1}{2}\tr\left(E_x\rho_{-+}\right) + \frac{1}{2}\tr\left(E_x\rho_{+-}\right) \tr\left(E_z\rho'_{+-}\right) \nonumber \\
& \leq \frac{1}{2} \varepsilon+\frac 1 2 \left(\tr(E_z\rho_{+-})+\sqrt{\varepsilon}\right)\nonumber \\
& \leq  \frac{1}{2}\varepsilon + \frac{1}{2} \left(\varepsilon + \sqrt{\varepsilon}\right).
\label{eq:ppp_upper}
\end{align} 

For the gbit embedding to satisfy both lower (\cref{eq:ppp_lower}) and upper (\cref{eq:ppp_upper}) bounds on the behaviour of $P(++|\sigma)$,
 we thus require $4\varepsilon + 2\sqrt{\varepsilon} - \varepsilon \sqrt{\varepsilon} - \varepsilon^2 \geq 1$,
 which solves to $\varepsilon \geq 0.101416$.
That is, no matter the dimension of the quantum system we use, our embedding of a gbit must have at least around $10\%$ error.

Taking also Lemma~\ref{lem:Infinity} into account, we have thus proven the following:
\begin{example}
\label{LemApproxGbit}
Let $\varepsilon\leq 0.1014$. Then the gbit cannot be $\varepsilon$-embedded into any $\mathcal{Q}_n$ or $\mathcal{Q}_\infty$.
\end{example}

This example provides some additional intuition on why the gbit embedding has to be somewhat noisy. 
The constraint that the equal mixture $\sigma$ of $\rho_{++}$ and $\rho_{--}$ is statistically identical to the equal mixture of $\rho_{+-}$ and $\rho_{-+}$ arises from the demand that the quantum simulation is univalent.
Meanwhile, the requirement to replicate gbit behaviour also requires that these four states have as distinguishable behaviour as possible when input to $\mathcal{D}$.
A degree of noise is thus required to satisfy both these constraints simultaneously.
Contrast this noisy embedding with the multivalent behaviour of the exact Holevo simulation. 
There, the two alternatives of how to prepare the gbit state $a'$ as mixtures, i.e.\ the two contexts, are encoded onto entirely different states, hence enabling the possibility of entirely different behaviour for each context when the preparation is acted on by $\mathcal{D}$.

\subsection{Using nonlocality to certify nonembeddability}
\label{SubsecNLNE}
The above example gives us a lower bound on the required error to embed a gbit, but its derivation is very specific to the gbit's geometry.
In the following subsection, we will provide a general prescription for obtaining such bounds for a larger class of GPTs via concepts from the study of Bell nonlocality.

It may seem surprising at first that the study of \emph{bipartite} correlations says anything about the $\varepsilon$-embeddability of \emph{single} GPT systems into quantum theory. 
But both embeddability and Bell nonlocality study \emph{dimension-independent} problems: 
is there \emph{any} dimension $n$ such that we can embed $\mathcal{A}$ into $\mathcal{Q}_n$; or, what is the maximum over all dimensions $n$ of the local quantum systems for a certain Bell correlation? 
This hints why insights into the latter can be useful for the study of the former.

We begin by defining a notion of bipartite states on pairs of GPTs. 
(Here, we ignore a large part of theory about composition in GPTs, and focus only on those aspects that are relevant for the study of embeddings.)

\begin{definition}[Bipartite states]
\label{def:BipartiteStates}
Let $\mathcal{A}$ and $\mathcal{B}$ be GPTs. 
A \textbf{bipartite state} on $\mathcal{A}\mathcal{B}$ is a bilinear map $\omega_{AB}:A\times B\to\mathbb{R}$ which is normalized and positive, i.e.
\begin{itemize}
	\item $\omega_{AB}(u_A,u_B)=1$,
	\item $\omega_{AB}(e_A,f_B)\geq 0$ for all $e_A\in \bar E_A, f_B\in \bar E_B$,
\end{itemize}
where $\bar E_A$ is the set of all $e\in A$ with $0\leq (\omega,e)\leq 1$ for all $\omega\in\Omega_A$. (Clearly $E_A\subset\bar E_A$, and these sets agree if $\mathcal{A}$ is unrestricted.) A special case are the \emph{product states} $\omega_{AB}=\omega_A\otimes\varphi_B$ for $\omega_A\in\Omega_A,\varphi_B\in\Omega_B$, acting as $\omega_A\otimes\varphi_B(e_A,f_B)=\omega_A(e_A)\varphi_B(f_B)$. A state $\omega_{AB}$ is \textbf{separable} if it can be written as a convex combination of product states, and otherwise it is \textbf{entangled}.
\end{definition}
Since the set of product states is compact, so is their convex hull (the set of separable states). 
The set of all bipartite states, being closed and bounded, is also compact. 
We will use bipartite states $\omega_{AB}$ only as calculation tools, without any claim of direct physical relevance.

We need a few more concepts before we can return to the question of $\varepsilon$-embeddability of GPTs. 
The next definition gives a possible way to quantify entanglement in bipartite states:
\begin{definition}[Robustness of entanglement]
Let $\mathcal{A}$ and $\mathcal{B}$ be GPTs and $\omega_{AB}$ a bipartite state on $\mathcal{A}\mathcal{B}$. 
Then we define its {\bf robustness of entanglement} as
\begin{align}
   \mathcal{R}(\omega_{AB}):=&\frac 1 2\inf\left\{\sum_{i=1}^n|\lambda_i|\,\,\left|\,\, n\in\mathbb{N}, \exists \varphi_A^{(i)}\in\Omega_A,\psi_B^{(i)}\in\Omega_B:\right.\right. \nonumber\\
& \qquad \left.\omega_{AB}=\sum_{i=1}^n \lambda_i \varphi_A^{(i)}\otimes\psi_B^{(i)}, \sum_{i=1}^n \lambda_i=1\right\}-\frac 1 2.
\label{eq:DefE}
\end{align}
\end{definition}
If $\mathcal{A}$ and $\mathcal{B}$ are quantum systems and if $\omega_{AB}$ is a valid state on their standard composite, then $\mathcal{R}(\omega_{AB})$ equals the standard quantum definition of robustness of entanglement by~\citet{VidalTarrach}. This follows from item (iii) of Lemma~\ref{lem:Eprops} below, which shows that our definition is also a special case of Takagi and Regula's robustness definition for resource theories in GPTs~\cite{TakagiRegula}.
\begin{lemma}
\label{lem:Eprops}
Robustness of entanglement $\mathcal{R}$ has the following properties:
\begin{itemize}
\item[(i)] The infimum of Eq.~(\ref{eq:DefE}) is attained as a minimum. 
Furthermore, it is sufficient to optimize over pure $\varphi_A^{(i)}$ and $\psi_B^{(i)}$.
\item[(ii)] $\mathcal{R}(\omega_{AB})\geq 0$, and $\mathcal{R}(\omega_{AB})=0$ if and only if $\omega_{AB}$ is a separable state.
\item[(iii)] $\mathcal{R}(\omega_{AB})=\min\{\varepsilon\geq 0\,\,|\,\, (1+\varepsilon)\varphi_{AB}^S-\varepsilon \psi_{AB}^S=\omega_{AB},\quad\varphi_{AB}^S,\psi_{AB}^S\mbox{ are separable states}\}$.
\item[(iv)] $\mathcal{R}$ is convex, i.e.\ if $0\leq\mu\leq 1$, then
\begin{align}
\mathcal{R}(\mu \omega_{AB}+(1\!-\!\mu)\omega'_{AB})\leq \mu \mathcal{R}(\omega_{AB})+(1\!-\!\mu)\mathcal{R}(\omega'_{AB}).
\end{align}
\item[(v)] $\mathcal{R}$ is continuous.
\end{itemize}
\end{lemma}
The proof is given in Appendix~\ref{SecProoflem:Eprops}.

In the following, we will not be so interested in the entanglement of single states, but in the maximal value that this quantity can take for a given GPT:
\begin{definition}
\label{def:SelfEntangle}
For any GPT $\mathcal{A}$, we define its \textbf{self-entangleability} as
\begin{align}
   \mathcal{R}(\mathcal{A}):=\max_{\omega_{AA}}\mathcal{R}(\omega_{AA}),
\end{align}
where the maximum is over the bipartite states on two copies of $\mathcal{A}$. 
\end{definition}
We make a few immediate remarks:
First, it is sufficient to maximize over the \emph{pure} bipartite states, since $\mathcal{R}$ is convex (Lemma \ref{lem:Eprops}.iv).
Second, if $\mathcal{A}$ is a classical GPT then all states in $\mathcal{AA}$ are separable, and so $\mathcal{R}(\mathcal{C}_n)=0$.

\begin{lemma}
\label{lem:gbitSE}
The self-entangleability of the gbit $\mathcal{A}$ is $\mathcal{R}(\mathcal{A})=\frac 1 2$.
\begin{proof}
To see this, we note that for gbits, which are unrestricted, the bipartite states (as per Definition~\ref{def:BipartiteStates}) are those of the maximal tensor product of two gbits. This contains just two classes of pure state: product states and PR-boxes~\cite{PopescuR94}.
Moreover, there are only finitely many such states,
 and by exhaustive computer--assisted algebra, we can thus minimize \cref{eq:DefE}.
This yields a minimum decomposition into four separable states with $\lambda_1\!=\!\lambda_2\!=\!\lambda_3\!=\!\frac{1}{2}$ and $\lambda_4\!=-\frac{1}{2}$. 
\end{proof}
\end{lemma}

As in quantum theory, we can study Bell correlations corresponding to bipartite states.
\begin{definition}[Bell correlations]
Let $\mathcal{A}$ and $\mathcal{B}$ be GPTs. Suppose that there exist finite sets $X,X',Y,Y'$ and, for all $a\in X',b\in Y',x\in X, y\in Y$, effects $e_{a|x}\in E_A$, $f_{b|y}\in E_B$ with $\sum_a e_{a|x}=u_A$ and $\sum_b f_{b|y}=u_B$. Then, for any bipartite state $\omega_{AB}$, the probability table
\begin{align}
   P(a,b|x,y):=\omega_{AB}(e_{a|x},f_{b|y})
\end{align}
is called a \textbf{behaviour on $\mathcal{AB}$}~\cite{KhalfinT85}. 
For given finite sets $X,X',Y,Y'$, a \textbf{Bell functional} is a linear functional on the corresponding set of behaviours, i.e.
\begin{align}
   B[P]=\sum_{a,b,x,y}b_{a,b,x,y}P(a,b|x,y),
\end{align}
with $b_{a,b,x,y}\in\mathbb{R}$. 
We write $|B|:=\sum_{a,b,x,y}|b_{a,b,x,y}|$, and the maximum of $B[P]$ over all behaviours $P$ of $\mathcal{AB}$ will be denoted $B_{\mathcal{AB}}$.
Furthermore, we define the \textbf{maximal quantum value} of $B$ as $B_{\mathcal{Q}}:=\sup_{P_Q} B[P_Q]$, where we optimize over all behaviours $P_Q$ that can arise on the tensor product of two finite-dimensional Hilbert spaces $A$ and $B$ as $P_Q(a,b|x,y)=\tr(\rho_{AB}\, e_{a|x}\otimes f_{b|y})$. 
The \textbf{maximal classical value} $B_{\mathcal{C}}$ is analogously defined as $B_{\mathcal{C}}:=\sup_{P_C}B[P_C]$, with $P_C$ constrained to arise from a joint classical probability distribution.
\end{definition}

The following lemma gives a first hint at the relation between embeddability and Bell nonlocality: if a GPT is $\varepsilon$-embeddable into quantum theory, then its Bell correlations are $\mathcal{O}(\varepsilon)$-close to those of quantum theory. This reformulates and generalizes the results of~\cite{KleinmannOSW13,BarnumBBEW10}.
\begin{lemma}
\label{lem:BellTables}
Let $\mathcal{A}$ be a GPT for which there exists a unital $\varepsilon$-embedding into some quantum system $\mathcal{Q}_n$, where $n\in\mathbb{N}$. 
Suppose that we have a behaviour $P(a,a'|x,x')$ on two copies of $\mathcal{A}$. 
Then there exists a quantum behaviour $P_Q(a,a'|x,x')$ on two $n$-dimensional Hilbert spaces with
\begin{align}
   |P(a,a'|x,x')-P_Q(a,a'|x,x')|\leq 2\varepsilon[1+2\,\mathcal{R}(\mathcal{A})]
\end{align}
for all $a$, $a'$, $x$, and $x'$, where $\mathcal{R}(\mathcal{A})$ is the self-entangleability of $\mathcal{A}$.
\end{lemma}
The proof is given in \cref{app:BellTableProof}.
The assumption of unitality can be dropped at the expense of replacing $\varepsilon$ by $2\varepsilon$ (Lemma~\ref{lem:Unital}), and the case $n=\infty$ can be treated by replacing $\varepsilon$ by $\varepsilon+\delta$ for some $\delta>0$ (Lemma~\ref{lem:Infinity}).

This lemma allows us to obtain bounds on the embeddability of a broad class of GPTs. 
Before turning to the general result, let us again consider the special case of the gbit. 
While the next example gives us a statement that is strictly weaker than \cref{LemApproxGbit}, it demonstrates the main idea and the strategy for the general proof.
\begin{example}
\label{eg:TsirelsonGbit}
The Tsirelson bound implies that the gbit cannot be $\varepsilon$-embedded into any $\mathcal{Q}_n$ or $\mathcal{Q}_\infty$ if $\varepsilon\leq 0.00915$.
\begin{proof}
Let $\mathcal{A}$ be the gbit. 
To every behaviour $P(a,a'|x,x')$ on two gbits $\mathcal{AA}$, with $x,x',a,a'\in\{0,1\}$, we can associate a winning probability for the CHSH game~\cite{ClauserHSH69,VerSteegW09}:
\begin{align}
   P^{\rm win}=\frac{1}{4}\sum_{x,x'} \; \sum_{a,a':a\oplus a'=x\cdot x'}P(a,a'|x,x'),
\end{align}
where $\oplus$ denotes addition modulo two. 
The composite state space of two gbits admits Popescu-Rohrlich (PR) correlations~\cite{PopescuR94} $P_{\rm PR}(a,a'|x,x')$ which give a winning probability of unity, i.e.\ $P^{\rm win}_{\rm PR}=1$.
On the other hand, \emph{every} quantum Bell behaviour $P_Q$ gives a value of $P^{\rm win}_Q\leq \frac{1}{2}+\frac 1 {2\sqrt{2}}$, the Tsirelson bound~\cite{Cirelson80}.

Suppose that there exists an $\varepsilon$-embedding of the gbit into some $\mathcal{Q}_n$ or $\mathcal{Q}_\infty$. 
Then, according to Lemmas~\ref{lem:Unital} and~\ref{lem:Infinity}, for every $\delta>0$, there exists a unital $(2\varepsilon+\delta)$-embedding of $\mathcal{A}$ into some $\mathcal{Q}_n$. But then, Lemma~\ref{lem:BellTables} implies that there exists a quantum behaviour $P_Q$ with
\begin{align}
   |P_{\rm PR}(a,a'|x,x')-P_Q(a,a'|x,x')|\leq 4 (2\varepsilon+\delta)
\end{align}
for all $a,a',x,x'$ (recall Lemma~\ref{lem:gbitSE} that $\mathcal{R}(\mathcal{A})=\frac 1 2$).
But $P^{\rm win}$ is a linear combination of eight such probabilities, each with a prefactor of $1/4$. Thus
\begin{align}
   \frac{1}{2} -\frac 1 {2\sqrt{2}}\leq P^{\rm win}_{\rm PR}-P^{\rm win}_Q\leq \frac 1 4\cdot 8\cdot 4(2\varepsilon+\delta).
\end{align}
Since this is true for all $\delta>0$, it must also be true for $\delta=0$, hence $\varepsilon\geq (2-\sqrt{2})/64\approx 0.00915$. 
\end{proof}
\end{example}

This proof strategy can be generalized in obvious ways beyond the gbit.  This leads us to the main result of this section:
\begin{theorem}
\label{thm:QuantumFiniteEps}
Suppose that $\mathcal{A}$ is a GPT that ``admits of post-quantum self-correlations'' in the following sense: there is some Bell functional $B$ such that $B_{\mathcal{A A}}>B_{\mathcal{Q}}$, i.e.\ some state on two copies of $\mathcal{A}$ violates the corresponding Bell inequality by more than any bipartite quantum state. Then, for every
\begin{align}
\label{eq:QuantumEpsBound}
   \varepsilon<\frac{B_{\mathcal{A A}}-B_{\mathcal{Q}}}{4|B|(1+2\,\mathcal{R}(\mathcal{A}))},
\end{align}
the GPT $\mathcal{A}$ cannot be $\varepsilon$-embedded into any $\mathcal{Q}_n$ or $\mathcal{Q}_\infty$. 
\end{theorem}
This theorem allows us to obtain bounds on the quantum embeddability for all GPTs that admit of post-quantum self-correlations. 
For example, it is known that all unrestricted GPTs where the set of normalized states $\Omega$ is a regular polygon with an even number of vertices admit of post-quantum CHSH Bell-correlations~\cite{JanottaGBB11}. 
Theorem~\ref{thm:QuantumFiniteEps} shows that these GPTs $\mathcal{A}$ cannot be perfectly embedded into quantum theory, and that this result is robust up to some $\varepsilon>0$ that we can compute.

We do not currently know whether Theorem~\ref{thm:QuantumFiniteEps} produces nontrivial bounds for all GPTs that are not embeddable into quantum theory. 
The theorem applies to restricted and unrestricted GPTs alike, but we expect that it yields stronger results for GPTs that are close to unrestricted. 
However, this theorem cannot be straightforwardly generalized to more than two parties or to embeddings into theories which are not classical or quantum. We comment more on this, and on the relation to a result by \citet{Pusey18}, in Appendix~\ref{sec:NoGen}.

Finally, we remark that this strategy can also be adapted to test embeddings into \emph{classical} probability theory, since bipartite states on two classical systems (as defined in Definition~\ref{def:BipartiteStates}) cannot violate Bell inequalities.
\begin{theorem}
\label{thm:ClassicalFiniteEps}
Suppose that $\mathcal{A}$ is a GPT that admits of nonlocal self-correlations, i.e.\ there is some Bell functional $B$ with $B_{\mathcal{AA}}>B_{\mathcal{C}}$. Then, for every
\begin{align}
\label{eq:ClassicalEpsBound}
   \varepsilon<\frac{B_{\mathcal{A A}}-B_{\mathcal{C}}}{4|B|(1+2\,\mathcal{R}(\mathcal{A}))},
\end{align}
the GPT $\mathcal{A}$ cannot be $\varepsilon$-embedded into any $\mathcal{C}_n$ -- that is, $\mathcal{A}$ does not admit of an $\varepsilon$-accurate discrete noncontextual ontological model.
\end{theorem}

\section{Testing quantum theory experimentally}
\label{sec:Experiment}
One way to probe the fundamental probabilistic structure of nature is to perform \emph{theory--agnostic tomography}. 
The essential idea is to isolate a physical system in the laboratory, implement as many preparation and measurement procedures as possible, 
 and fit a GPT to the resulting data. 
This notion has been introduced, thoroughly analyzed, and experimentally implemented in Refs.~\cite{MazurekPKRS16,MazurekPRS21,GraboweckyPCSR21}.

Theory--agnostic tomography can either be used to certify the nonclassicality of a physical system or, more specifically, to verify that the corresponding GPT is one predicted by quantum theory. 
In Refs.~\cite{MazurekPRS21, GraboweckyPCSR21} such tomography of the polarization degree of freedom of a photon yielded a GPT close to the qubit $\mathcal{Q}_2$, and a photonic three-level system turned out to be well-described by the qutrit $\mathcal{Q}_3$.

However, \citet{MazurekPRS21} (also~\cite{GraboweckyPCSR21}) suggest that their experiments suffer from an unavoidable drawback: \emph{``The most significant loophole in this sort of analysis is that the set of preparations and measurements one implements might fail to be tomographically complete for the system of interest.''} 
As the argument goes, failing to implement a tomographically complete set may lead the experimenter to underestimate the true dimension of the GPT, to incorrectly infer operational equivalences, and thus to incorrectly verify or falsify underlying classical or quantum models. \Citet{Pusey18} argues that, in principle, this problem cannot be completely avoided, although there are partial mitigations~\cite{PuseyRM19}.

Here, we suggest a different strategy: instead of testing whether the experimentally determined GPT is quantum, test whether it admits of a plausible quantum simulation. 
This drops the assumption of tomographic completeness (the experiment might give us an ``effective GPT'' that is only a shadow of a higher-dimensional quantum system), and replaces it by an assumption of \emph{plausibility} that we will identify with \emph{univalence}.
To see why this strategy suggests itself, it is helpful to dissect the subtleties that underlay the usual notion of tomographic completeness.

\subsection{Tomographic completeness and the notion of a physical system}
\label{sec:Completeness}
One major strength of generalized contextuality is its operational flavour: in principle, it is defined entirely in terms of notions of direct relevance for laboratory physics. 
There is one notable exception, however: the notion of a \emph{physical system} is taken as a primitive in the framework of operational theories. 
This is particularly evident in category-theoretic formulations of convex-operational theories~\cite{ChiribellaAP10,CoeckeK17} where the theories are defined as a collection of objects (the physical systems) and morphisms (transformations between the systems).
One may then think of quantum theory as a process theory of Hilbert spaces (or operator algebras), and conclude that every given physical system must be described by a GPT that corresponds exactly to the states and effects of a mathematical object of this kind. 
This picture of ``wires'' and ``gates'' hints at an intuition that physical systems should be thought of as well-delineated \emph{objects}, \emph{defined} by their spatial or spatiotemporal location, that come with an attached GPT which we may or may not be able to probe or understand with our current experimental capabilities.

However, many physical systems like the polarization degree of freedom of a photon are arguably not of this kind: they are not well-circumscribed spatially localized objects, but collections of phenomena \emph{defined} by their experimental context. 
Electron spin, for example, corresponds to a two-dimensional subspace of an infinite-dimensional Hilbert space $\mathcal{H}$; under Lorentz boosts, this subspace moves continuously through $\mathcal{H}$, mixing with the momentum degree of freedom. 
Hence, the laboratory rest frame, together with the coupling to the local electromagnetic field, defines this physical system in the first place.

Even in the ``object'' view of physical systems, tomographic completeness is contingent on the way that we define the boundaries of the physical system of interest. 
For example, doing tomography on the polarization of a photon necessarily ignores all other aspects of its wavefunction (such as orbital angular momentum); all physical objects are embedded in larger environments that are disregarded when an experimenter decides to probe a specific localized degree of freedom only.

We therefore make a suggestion that has recently been proposed independently in a similar way in Ref.~\cite{GittonWoods22}: take operational reasoning one step further, and think of physical systems as being \emph{defined} by an experimental scenario. 
Concretely, an experimenter chooses a setup that allows for the reliable implementation of a finite number of preparation and measurement procedures, such that all measurements can be applied to all preparations. For interpretive reasons, we may imagine an ``effective physical system'' as the object that is sent from the preparation to the measurement devices --- or rather as the bundle of statistical properties of this object that are probed in the experiment. 
However, this interpretation is irrelevant for the analysis of the experiment.

Instead of pondering whether the set of implemented procedures is tomographically complete, under this view, it represents a complete set \emph{by definition of the effective physical system}. 
Furthermore, the operational theory that describes the specific set of preparation and measurement procedures gives rise to an \emph{effective GPT $\mathcal{A}'$} that describes this system (for the choice of nomenclature see Figure~\ref{fig:ApproxStatistics} below), and the physicist is guaranteed to be able to experimentally determine some approximation $\mathcal{A}''$ of $\mathcal{A}'$ (see \cref{fig:ApproxStatistics}).

\begin{figure}[tbh]
\begin{centering}
\includegraphics[width=0.235\textwidth]{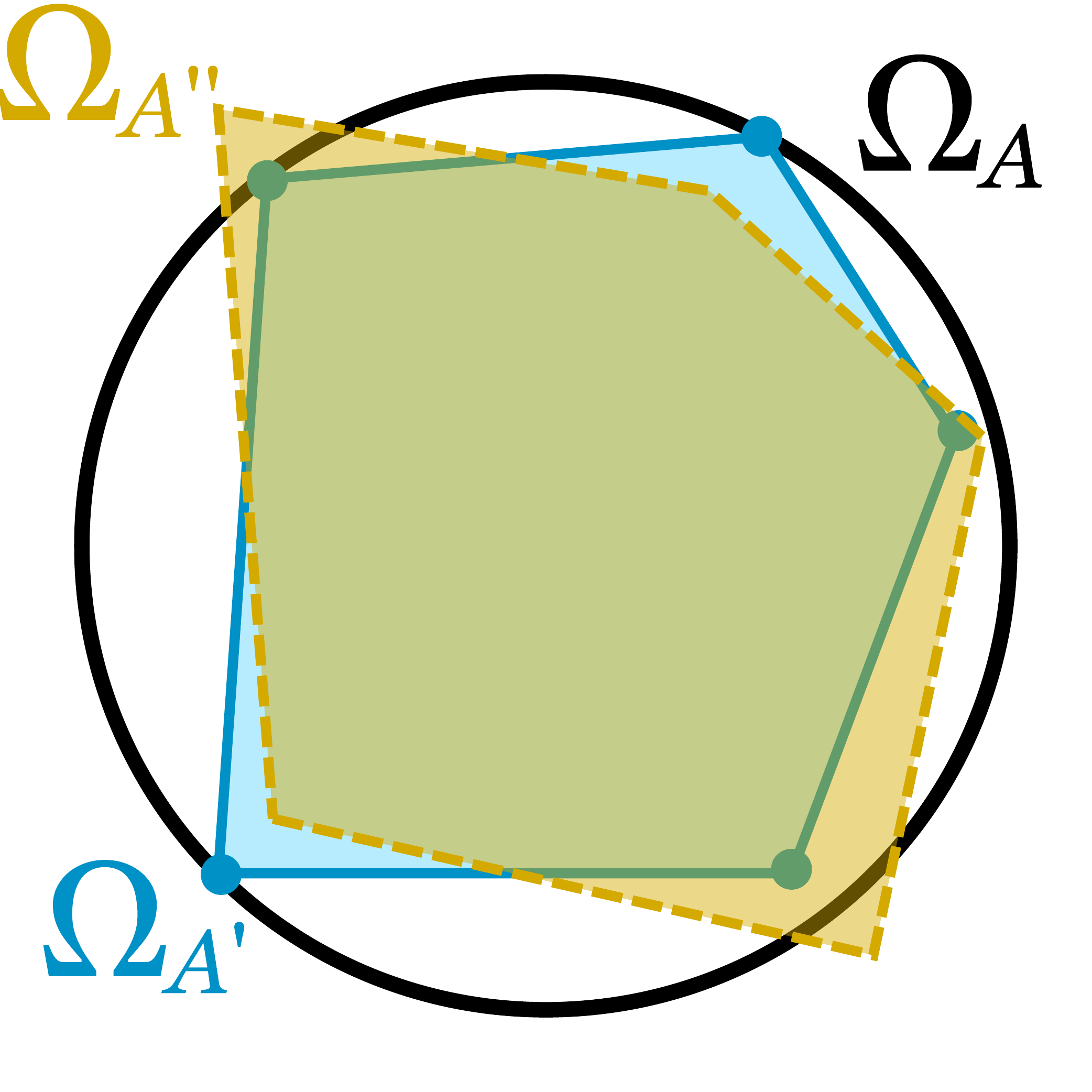}
\caption{%
\label{fig:ApproxStatistics}
\caphead{Approximately embedded state spaces.} 
Suppose all in-principle available preparations and measurements of an experimental scenario are described by an effective GPT $\mathcal{A}$ (here: the black circle of states $\Omega_A$).
Due to resource restrictions, experiments can only implement a finite subset of those procedures, inducing another GPT $\mathcal{A}'$ (here: blue polygon) that can be perfectly embedded into $\mathcal{A}$, assuming univalence.
Moreover, due to only collecting finite statistics, the actual GPT that is experimentally determined, $\mathcal{A}''$, is itself an approximation of $\mathcal{A}'$ (here: yellow dashed polygon of measured states $\Omega_{A''}$). 
There is an $\varepsilon$--embedding of $\mathcal{A}''$ into $\mathcal{A}'$, and hence also into $\mathcal{A}$, where $\varepsilon$ is a function of the experimental error due to finite statistics.
}
\end{centering}
\end{figure}

\subsection{Plausibility of assuming univalence}
\label{sec:Plausibility}
Suppose, then, that we perform theory--agnostic tomography as described above, and do this on an effective physical system. 
Which resulting GPTs $\mathcal{A}''$  can we plausibly expect to find if nature is fundamentally described by quantum theory?

As we have shown in Lemma~\ref{lem:AllContextual}, a priori, \emph{all} GPTs can be simulated by quantum (and even classical) theory, but such simulations will in general have to be \emph{multivalent}, i.e.\ involve a specific sort of finetuning~\cite{WoodS15}. This finetuning could be implemented, for example, by a malicious adversary with supreme technological capabilities. For instance, this adversary could have deliberately manipulated the laboratory such that the experimenter has only access to the subspace of measurements of a classical four-level system $\mathcal{C}_4$ that gives rise to the gbit $\mathcal{A}'$ via the Holevo projection, as explained in Figure~\ref{fig:Holevo}.

However, in any non-adversarial scenario where an experimenter probes his system in a predefined number of ways that are determined by resource requirements and not by deliberate choice, in an experiment that they have built essentially from scratch, we would not expect such finetuning to take place. Indeed, the Holevo projection is equivalent to a preparation-contextual ontological model, and such models are considered implausible in the scientific community concerned with generalized noncontextuality~\cite{Spekkens05,MazurekPKRS16,Spekkens19}. 
The plausibility of univalent simulations relies on very similar arguments; and thus we stipulate:

\begin{assumption}[Multivalent simulations are implausible.]
Multivalent simulations are implausible for essentially the same reasons that contextual ontological models are implausible. 
In a nutshell, such simulations involve an unlikely amount of fine-tuning.
\end{assumption}

A wide range of arguments has been put forward in favor of the assumption of noncontextuality for classical ontological models of operational theories (and thus for contextuality being a certificate of nonclassicality)~\cite{Spekkens05,MazurekPKRS16,Spekkens19}, and we argue that they also apply to our more general situation of simulation by GPTs (with univalence generalizing noncontextuality).
Essentially, a \emph{multivalent} simulation of an effective GPT $\mathcal{A}$ by a fundamental GPT $\mathcal{B}$ is implausible because it would involve \emph{distinct} fundamental states (or effects) of $\mathcal{B}$ that are necessary to explain the statistics of $\mathcal{A}$, but whose difference is exactly ``washed out'' in the effective theory. 

This is arguably similarly conspiratorial~\cite{Spekkens05} as explanations for the violation of Bell inequalities that involve nonlocal hidden variable models: 
 such models postulate instantaneous signalling at the ontological level, but the signalling is miraculously invisible on the phenomenal level such that we are not able to construct a Bell telephone. 
We think that this is implausible (and so are multivalent simulations) because they violate the following methodological principle, generalizing Leibniz' doctrine of the ``identity of the indiscernibles''~\cite{Spekkens19,MazurekPKRS16} as a motivation of noncontextuality: \emph{processes that are statistically indistinguishable in an effective theory should not require explanation by multiple distinguishable processes in a more fundamental theory.} 
This postulate is in spirit similar to the motivations for model selection suggested in Ref.~\cite{Daley}.

More concretely, several technical results of this article that support this view: multivalence cannot arise from noise, complete decoherence or coarse-graining processes (Lemmas~\ref{lem:Restricted} and~\ref{LemDecoherence}), and univalence extends transitively across different levels of physical description (Lemma~\ref{LemIterate}).
Multivalent simulation is a phenomenon of finetuning~\cite{WoodS15} that we do not expect to arise naturally. 
As such, we assert that, outside of deliberately engineered (e.g.\ adversarial) scenarios, only univalent simulations are plausible.

\subsection{Approximate embeddings account for experimental imperfections}
\label{SubsecApproxExp}
Let us now discuss how to implement a test of quantum theory via theory-agnostic tomography under these assumptions, and how to deal with unavoidable experimental imperfections.

The physicist begins by designing an experimental scenario, which involves the description of a set of preparation and measurement procedures that can in principle be implemented in the given scenario. This defines an operational theory, giving rise to an effective GPT $\mathcal{A}$ (in the terminology of \citet{MazurekPRS21}, the ``true'' GPT) that describes the statistical relations of all in principle available preparation and measurement procedures of the given experimental scenario. 
If the experimenter could implement all preparation and measurement procedures an infinite number of times to collect perfect statistics, assuming perfect stability of the experimental setup, this is the unique GPT that they would infer from the resulting data. 

In practice, experiments~\cite{MazurekPRS21,GraboweckyPCSR21} will have at least two drawbacks as compared to this idealized description. 
First, it will in general not be possible to implement \emph{all} of the (potentially infinitely many) possible preparations or measurements. 
For example, think of a measurement being determined by the direction $\vec{n}$ of inhomogeneity of a magnetic field, as in a Stern-Gerlach experiment. 
We know that our experimental scenario allows in principle an uncountably-infinite set of measurements indexed by $\vec{n}$, but we will actually only ever be able to measure in a finite number of such directions. 
This means that the experimenter will ultimately only be able to implement a finite subset of states and measurements.

To this end, the physicist will have to decide \emph{before} the experiment on a \emph{finite} number of preparation and measurement procedures that they will implement. 
The experimenter is supposed to choose this subset of procedures in terms of \emph{resource requirements}: bounding the amount of (for example) energy, time, or resolution should determine the ``cut-off'', i.e.\ the decision of which procedures to implement. 
Crucially, the physicist should not cherry-pick a subset of procedures by deliberately leaving out some other procedures that are as easy to implement as the chosen ones: this could potentially open the door to undesired finetuning. 
Needless to say, the choice of procedures should in particular not be done \emph{after} the experiment, e.g.\ by discarding some of the data selectively; the general standards of experimental physics apply here too!

In the Stern-Gerlach example, the experimenter could begin by deciding on a discretization that samples $N$ approximately equally spaced points on the unit sphere, and implements the corresponding projective measurements in these directions. 
The distance between the points (and hence the number $N$) could be determined by the accuracy of magnet rotations available to the experimenter, under the resource constraints of the experiment.

Consider the effective GPT $\mathcal{A}'$ obtained from the operational theory that describes this \emph{finite} set of preparation and measurement procedures; due to finiteness, it will generally be a \emph{restricted} GPT with sets of states and effects that are polytopes. 
This effective GPT $\mathcal{A}'$ of implemented procedures must then be simulated by the more exhaustive GPT $\mathcal{A}$ of all available procedures. 
If the experimenter adheres to the fair practices we just described, we expect that this simulation should not be finetuned: in other words, it should be \emph{univalent}, and $\mathcal{A}'$ should have an exact embedding into $\mathcal{A}$.

The second drawback is that the experimenter can only ever collect a finite amount of statistics.
The frequencies that they measure in the experiment will only be approximations to the actual probabilities. 
Thus, they will actually obtain an approximation $\mathcal{A}''$ of $\mathcal{A}'$, and the states and effects of this approximation will not in general be subsets of the states and effects of $\mathcal{A}'$. 
Indeed, in the experiment on theory-agnostic tomography on photon polarization, \citet{MazurekPRS21} found that some of the experimentally determined state vectors ``stick out'' of the Bloch ball of $\mathcal{A}'=\mathcal{A}$, corresponding to nonstates -- an obvious artefact of finite statistics.

We can say more about how $\mathcal{A}''$ approximates $\mathcal{A}'$ by building on an observation of \Citet{MazurekPRS21}: in their specific experiment, shrinking the state and effect vectors of $\mathcal{A}''$ by a small amount (0.14\%) embeds them in the qubit state space $\mathcal{A}$. 
In general, we expect that there are maps $\Phi:A''\to A'$ and $\Psi:(A'')^*\to (A')^*$ that are close to the identity map, ``shrinking'' the state and effect spaces $\Omega_{A''}$ and $E_{A''}$ by a small amount, and embedding them into $\Omega_{A'}$ and $E_{A'}$. 
Hence, these maps define an $\varepsilon$-embedding of $\mathcal{A}''$ into $\mathcal{A}'$, and it is clear that the embedding becomes more exact (i.e.\ $\varepsilon\searrow 0$) when the number of experimental runs is increased and more and more data is collected.

While the exact error analysis depends on the specific details of the protocol that is used to produce the estimate $\mathcal{A}''$, we can imagine that the experimenter chooses a number of repetitions that is large enough to conclude with high confidence that the experimentally determined GPT $\mathcal{A}''$ can be $\varepsilon$-embedded into $\mathcal{A}'$, and thus into $\mathcal{A}$, for some small value of $\varepsilon$.
With the null-hypothesis that $\mathcal{A}$ can be exactly embedded into some $\mathcal{Q}_n$ or $\mathcal{Q}_\infty$, i.e.\ admits of an univalent quantum simulation,
 this implies that $\mathcal{A}''$ must also be $\varepsilon$-embeddable into quantum theory by virtue of Lemma~\ref{LemIterate}. 
This can be tested: if it turned out (say, numerically) that the experimentally determined GPT $\mathcal{A}''$ does not admit of an $\varepsilon$-embedding into any $\mathcal{Q}_n$ or $\mathcal{Q}_\infty$, then the experiment would challenge QT -- and falsify it under the assumption of univalence.

\subsection{The experimental test in a nutshell}
\label{sec:Nutshell}
We are now in a position to provide a scheme of an experimental test of quantum theory.
In particular, suppose we start with the null hypothesis that the effective behaviour of some physical system we are testing {\em can} be explained fundamentally by quantum theory on a finite-dimensional or separable Hilbert space, i.e.\ that its effective GPT $\mathcal{A}$ can be simulated by some fundamental $\mathcal{Q}_n$ or $\mathcal{Q}_\infty$.
The following outlines a way to falsify this hypothesis, under the assumption (motivated above) that the simulation must be univalent, i.e.\ correspond to an embedding.

\begin{enumerate}
\item Identify the effective physical system to probe. 
That is, decide on a finite set of available preparation procedures $\{p\}$ and measurement procedures (with potential measurement outcomes $\{(m,k)\}$) to be implemented, while adhering to the fair practices described in Subsection~\ref{SubsecApproxExp}.
\item By performing all pairs of selected preparation and measurement procedures a large number of times, collect as much data as feasible to build an operational description (i.e.\ the set of frequencies approximating $\{P(k|p,m)\}$).
\item From these operational statistics, reconstruct an effective GPT $\mathcal{A}''$ (using, e.g.,\ the methodology in Refs.~\cite{GraboweckyPCSR21,MazurekPRS21}).
\item From the statistics, infer some value $\varepsilon>0$ (as small as possible) such that, with high confidence, $\mathcal{A}''$ can be $\varepsilon$-embedded into the unknown $\mathcal{A}'$.
As we have argued in the previous subsection, this $\varepsilon$ results from the inaccuracy of finite statistics, and can thus be estimated with suitable statistical methods.
\item \label{step:certify}Certify with numerical methods (based on analytic results) that $\mathcal{A}''$ cannot be $\varepsilon$-embedded into quantum theory.

If this succeeds, then the assumption of univalence implies that  our effective system represents an object whose statistical properties cannot be explained by quantum theory.
\end{enumerate}
The certification in step 5 can be done by various methods. 
One method follows from our results in Subsection~\ref{SubsecNLNE}:
\begin{enumerate}[\ref{step:certify}.1]
\item \label{step:BellStuff} Over the bipartite states (\cref{def:BipartiteStates}) of two copies of $\mathcal{A}''$, search (numerically -- the actual Bell test does not need to be performed) for a Bell functional whose maximum $B_{\mathcal{A}''\mathcal{A}''}$ exceeds the quantum value $B_\mathcal{Q}$. 
This need not be the maximum, but the greater the violation, the more tolerance to experimental error is admissible to falsify the null hypothesis.
\item \label{step:SelfEntangle} Calculate the self-entanglability (\cref{def:SelfEntangle}) $\mathcal{R}(\mathcal{A''})$ or an upper bound thereof.
\item \label{step:ApplyToUs} Check that inequality~(\ref{eq:QuantumEpsBound}) is satisfied, such that Theorem~\ref{thm:QuantumFiniteEps} implies that $\mathcal{A}''$ cannot be $\varepsilon$-embedded into any $\mathcal{Q}_n$ or $\mathcal{Q}_\infty$.
\end{enumerate}

A few clarifications are in order. 
First, the failure to find any Bell functional on $\mathcal{A}''\mathcal{A}''$ that violates the quantum bound does not imply that the effective theory $\mathcal{A}$ {\em can} be embedded in quantum theory (i.e.\ and hence can arise from univalent simulation) -- at least not for given finite $\varepsilon>0$. 
Whether such failure in the limit $\varepsilon\searrow 0$ guarantees exact embeddability is an open problem.

In general, our nonlocality test for certifying nonembeddability (steps 5.1--5.3) is not optimal, as the comparison of Examples~\ref{LemApproxGbit} and~\ref{eg:TsirelsonGbit} for the gbit demonstrates.
One may hence replace it by a more practical and efficient algorithm which, however, has yet to be developed. 
In some cases, it may be possible to use certain specific features of $\mathcal{A}''$ directly to certify nonembeddability; again, the gbit Example~\ref{LemApproxGbit} in \cref{sec:ExGbit} demonstrates this. 
 
Our nonlocality test has a useful spin-off: comparing $B_{\mathcal{A}''\mathcal{A}''}$ with $B_{\mathcal{C}}$ instead of $B_{\mathcal{Q}}$ in step 5.1 (and using Eq.~(\ref{eq:ClassicalEpsBound}) of Theorem~\ref{thm:ClassicalFiniteEps} in step 5.3) allows the experimenter to exclude approximate embeddings into \emph{classical theory}. 
That is, finding nonlocal self-correlations demonstrates the nonexistence of a discrete noncontextual ontological model, i.e.\ the presence of generalized contextuality.
One may thus test classical and quantum embeddability together in the same experiment.

\subsection{Towards a concrete implementation}
A concrete starting point for implementing, testing and refining the scheme proposed above could lie in the analysis of existing data~\cite{Fadel} from experiments with Bose-Einstein condensates (BECs). \citet{Schmied} have demonstrated the existence of Bell correlations between the spins of about 480 atoms in a BEC via collective measurements. Since it is impossible to actually perform a Bell test with local measurements in the BEC, this result relied on strong additional assumptions such as the correct calibration of the measurement devices and quantum theory's algebra of Pauli matrices.

The BEC is, on the one hand, a very high-dimensional quantum system, such that no tomographically complete data can be obtained; and, on the other hand, it behaves physically like a \emph{single} system well-described by a prepare-and-measure scenario. This makes it an ideal testbed for our proposal, which can potentially certify nonclassicality and quantumness under fewer assumptions as follows. First, perform GPT tomography not of the full BEC (which is too high-dimensional), but of an \emph{effective physical subsystem} that is defined by a finite number of preparations of spin-squeezed states and collective spin measurements used in the detection of Bell correlations in Ref.~\cite{Schmied}. Then, test whether the resulting GPT $\mathcal{A}''$ can be $\varepsilon$-embedded into classical theory $\mathcal{C}_n$ or quantum theory $\mathcal{Q}_n$. Anticipating a positive answer to the latter question, use this to determine the smallest $n$ such that $\mathcal{A}''$ is $\varepsilon$-embeddable into $\mathcal{Q}_n$.

The previous detection of Bell correlations suggests that $\mathcal{A}''$ should not be embeddable into any $\mathcal{C}_n$. This would in itself be a very interesting experimental result, certifying a notion of nonclassicality that is strictly stronger than Wigner negativity~\cite{SpekkensNeg08}. Moreover, our test of embeddability into $\mathcal{Q}_n$ could give us a lower bound on the Hilbert space dimension $n$ that is necessary to explain the data, similar to a dimension witness~\cite{Gallego2010}. We expect that the data analysis methods of~\cite{MazurekPRS21,GraboweckyPCSR21} have to be slightly extended to account for peculiarities of this more complex quantum system, for example due to the fact that $\mathcal{A}''$ will not be unrestricted, and since the embedding error parameter $\varepsilon$ has to be suitably related to the experimental uncertainties. We therefore leave a more detailed elaboration or implementation of this proposal to future work.

\section{Discussion and outlook}
In this article, we have introduced the notion of simulation of one general probabilistic theory by another, motivated  by the question of which effective GPTs we can plausibly expect to find in prepare-and-measure experiments if nature is fundamentally quantum. We have identified \emph{univalence} as an important property of simulations that we should naturally expect to hold, generalizing Spekkens' notion of noncontextuality~\cite{Spekkens05}, and leading to the notion of linear \emph{embeddings} of one GPT into another.
We have proven a multitude of results on exact and approximate embeddings, in particular into quantum theory. 
For example, we have shown that the special Euclidean Jordan algebras~\cite{BarnumGW20} (including real and quaternionic quantum theory) are the unique unrestricted generalized probabilistic theories (GPTs) that can be embedded exactly into quantum theory, we have elucidated the relation between noncontextuality inequalities~\cite{MazurekPKRS16} and approximate embeddings into classical theory, and we have constructed a method for proving the nonexistence of an approximate embedding into quantum or classical theory by using known results on Bell nonlocality.

Based on recent proposals for theory--agnostic tomography~\cite{MazurekPRS21,GraboweckyPCSR21}, we have then used these insights to suggest a novel scheme for an experimental test of quantum theory. 
Our proposal arguably avoids the tomographic completeness loophole of previous approaches because our notion of embedding is specifically targeted to the study of \emph{effective GPTs} that we may obtain as ``shadows'' of higher-dimensional quantum or other theories. 
To this end, we have argued that generalized noncontextuality should be further generalized to a methodological principle that constrains how effective theories can have plausible explanations in terms of more fundamental theories. 
Altogether, this reasoning suggests that effective GPTs found in the laboratory in any nonadversarial scenario should have approximate embeddings into quantum theory. 
Experimental evidence for the opposite would therefore challenge quantum theory.

Our work raises several interesting questions. 
From a practical point of view, it would be useful to find computationally more efficient and optimal algorithms for the certification of nonembeddability into quantum theory. 
Our nonlocality test shows that such certification is in principle possible for a large class of (possibly restricted) GPTs, but it is in general far from optimal (see Examples~\ref{LemApproxGbit} and \ref{eg:TsirelsonGbit}). 
Furthermore, it is an open question whether it yields nonzero bounds for all GPTs that cannot be perfectly embedded, and it involves steps that we expect to be computationally hard, including the search for post-quantum Bell inequality violations on two copies of the system. 
It would also be interesting to generalize our approach to the recently proposed scenario of ``GPT fragments''~\cite{SelbySWSKS21b} where preparations are allowed to succeed with probability strictly less than one.

Our proposed experimental scheme to test quantum theory invites further study, too. In particular, the statistical analysis of fitting a GPT to the data, and the estimate of $\varepsilon$ as a ``distance'' between $\mathcal{A}''$ and $\mathcal{A}'$ in step 4 of our protocol in Subsection~\ref{sec:Nutshell}, deserves further exploration. We hope that our proposal leads to actual implementations, which may in turn inform theoretical efforts to find more practical methods for its analysis.

The interpretation of such experiments --- in particular, the idea that their results can potentially challenge quantum theory --- relies on the claim that simulations of effective theories by fundamental theories should be expected to be univalent. In the special case that the fundamental theory is classical, this reduces to the assumption of noncontextuality, which has been thoroughly explored and motivated in a large body of literature, e.g.~\cite{Spekkens05,MazurekPKRS16,Spekkens19}. While our notion of univalence has arguably similarly compelling motivations as noncontextuality (after all, the latter is essentially a special case of the former), it is still a novel notion that should be analyzed further. This does not only involve conceptual analysis, but also the proof of theorems which show that multivalence is very hard to obtain. In this paper, we have begun to establish results along this line, by showing that multivalence cannot arise from noise, complete decoherence or coarse-graining processes (Lemmas~\ref{lem:Restricted} and~\ref{LemDecoherence}), and univalence extends transitively across different levels of physical description (Lemma~\ref{LemIterate}). A resource-theoretic perspective might provide the suitable language for further exploration of this line of argumentation.

Our results also raise some interesting questions in mathematical physics. 
For example, every algorithm that yields a nontrivial bound on the embeddability for, say, all nonembeddable unrestricted GPTs must necessarily recognize the special Euclidean Jordan algebras as singular cases in which this bound becomes trivial. 
This hints at some deeper mathematical structures at the intersection of embeddability and operator algebras. 
Does our nonlocality test already satisfy this desideratum? 
This would imply an unexpected relation between Jordan algebras and Bell inequalities. 
We hope that our work initiates further study of these surprising interrelations.

\hidetoc
\vbox{
\section*{Note}
This article supersedes the earlier preprint \href{https://arxiv.org/abs/2004.06136}{arXiv:2004.06136}, which contains a very condensed proof of Theorem~\ref{thm:Simulate}.
}

\section*{Acknowledgments}
\vbox{
We are grateful for helpful discussions with Howard Barnum, David Schmid, Farid Shahandeh, John H.\ Selby, Robert W.\ Spekkens, Reinhard Werner, and Alex Wilce. We acknowledge support from the Austrian Science Fund (FWF) via project P 33730-N. This research was supported in part by Perimeter Institute for Theoretical Physics. Research at Perimeter Institute is supported by the Government of Canada through the Department of Innovation, Science, and Economic Development, and by the Province of Ontario through the Ministry of Colleges and Universities.
}

\restoretoc

\bibliographystyle{unsrtnat}
\bibliography{embedding}
\balancecolsandclearpage

\appendix

\section*{Appendix}
\hidetoc

\section{Verification of the geometric intuition \mbox{in Figure~\ref{fig:ShrinkConvex}}}
\label{app:Polytope}
\begin{applemma}
Let $C\subset\mathbb{R}^d$ be a compact convex set of full dimension $d$ such that $0$ is in the interior of $C$. 
Then, for every $\lambda\in (0,1)$, there exists a polytope $P$ with at most $\left\lceil \left(\frac{c}{1-\lambda}\right)^{(d-1)/2}\right\rceil$ vertices such that $\lambda C \subset P \subset C$, where $c>0$ is a constant that only depends on $C$.
\begin{proof}
\label{LemPolytopeInBetween}
Let us first show that a nonzero neighbourhood of $\lambda C$ must be contained in $C$. 
Since $0$ is in the interior of $C$, there exists some $\alpha>0$ such that all $x\in\mathbb{R}^d$ with $\|x\|\leq\alpha$ are contained in $C$. We claim that
\begin{equation}
   \lambda C +\alpha(1-\lambda)B^d\subset C.
   \label{eqToShow}
\end{equation}
To show this, suppose that $y\in \lambda C+\alpha(1-\lambda)B^d$, that is, that there exists some $x\in C$ with $\|y-\lambda x\|\leq\alpha(1-\lambda)$. Let $x':=x+\frac{y-x}{1-\lambda}$, then
\begin{align}
\|x'\|&= \left\|\left(x+\frac{y-x}{1-\lambda}\right)-\left(x+\frac{\lambda x - x}{1-\lambda}\right)\right\| \nonumber \\
&= \frac 1 {1-\lambda}\|y-\lambda x\|\leq \alpha
\end{align}
and thus $x'\in C$. Since $y=\lambda x+(1-\lambda)x'$, and since $C$ is convex, it follows that $y\in C$. This proves \cref{eqToShow}.

Now we use known results (see e.g.~\cite{Bronstein08}) on the approximation of convex bodies, i.e.\ compact convex sets with nonempty interior, by polytopes. The Hausdorff metric for convex bodies $C,D$ is given by
\begin{align}
\rho_H(C,D)=\min\{\kappa\,\,|\,\, C\subset D+\kappa B^d,D\subset C+\kappa B^d\},
\end{align}
where $B^d$ denotes the unit ball in $\mathbb{R}^d$. Now~\cite[Subsection 4.1]{Bronstein08} shows (see also \citet{GruberK82}) that there is a constant $c(C)>0$ such that for all $n\in\mathbb{N}$, there exists a polytope $P_n\supset C$ with $n$ vertices which gives a good outer approximation of $C$, in the sense that
\begin{align}
   \rho_H(C,P_n)\leq\frac{c(C)}{n^{2/(d-1)}}=:\delta_n.
\end{align}
This implies $C\subset P_n\subset C+\delta_n B^d$. Now choose $n$ large enough such that $\delta_n\leq \alpha(1-\lambda)/\lambda$; a possible choice is $\displaystyle n:=\left\lceil \left(\frac{c}{1-\lambda}\right)^{(d-1)/2}\right\rceil$, where $c=c(C)/\alpha$. Then, using Eq.~(\ref{eqToShow}), we get the set inclusions
\begin{align}
   \lambda_C \subset \lambda P_n\subset \lambda C+\lambda\delta_n B^d\subset C.
\end{align}
That is, $P:=\lambda P_n$ satisfies the condition of the lemma. 
\end{proof}
\end{applemma}

\section{Proof of embedding properties}
\label{app:EmbeddingProperties}
\subsection{Proof of Lemma~\ref{lem:Linear}}
\begin{proof}
First, consider an univalent $\varepsilon$-simulation of $\mathcal{A}$ by $\mathcal{B}$. Let $d:=\dim A$, and pick $d$ linearly independent states $\omega_1^A,\ldots,\omega_d^A\in\Omega_A$. 
Then there are $d$ states $\omega_1^B,\ldots,\omega_d^B$ such that ${\rm Sim}_{\mathcal{B}}(\omega_i^A)=\{\omega_i^B\}$ for all $i$. 
Define $\Psi:A^*\to B^*$ as the linear extension of $\Psi(\omega_i^A)=\omega_i^B$ for $i=1,\ldots,d$. 
If $\omega_A\in C:={\rm conv}\{\omega_1^A,\ldots,\omega_d^A\}$, i.e.\ $\omega_A=\sum_{i=1}^d \lambda_i \omega_i^A$ for suitable $\lambda_i\geq 0$, $\sum_i\lambda_i=1$, then
\begin{align}
   {\rm Sim}_{\mathcal{B}}(\omega_A)=\sum_{i=1}^d \lambda_i\Omega_B(\omega_i^A)=\left\{\sum_{i=1}^d \lambda_i \omega_i^B\right\}=\{\Psi(\omega_A)\}.
\end{align}
Now suppose $\omega_A\in\Omega_A\setminus C$. 
Pick any state $\varphi_A$ in the relative interior of $C$, and consider the line connecting $\varphi_A$ and $\omega_A$. 
On it, we can find some $\rho_A\in C \setminus\{\varphi_A\}$, i.e.\ there is some $0<\lambda<1$ such that $\rho_A=\lambda\omega_A+(1-\lambda)\varphi_A$. Thus
\begin{align}
   \{\Psi(\rho_A)\}&={\rm Sim}_{\mathcal{B}}(\rho_A)=\lambda{\,\rm Sim}_{\mathcal{B}}(\omega_A)+(1-\lambda){\rm Sim}_{\mathcal{B}}(\varphi_A) \nonumber\\
   &=\lambda\,{\rm Sim}_{\mathcal{B}}(\omega_A)+(1-\lambda)\{\Psi(\varphi_A)\},
\end{align}
and from this it is elementary to infer that ${\rm Sim}_{\mathcal{B}}(\omega_A)=\{\Psi(\omega_A)\}$. Hence $\Psi(\Omega_A)\subseteq \Omega_B$, and $\Psi$ is a positive and normalization-preserving linear map.

The argumentation for effects is similar, applying the above construction to the convex hull $C$ of $d$ linearly--independent effects and the zero effect. Finally, the preservation of outcome probabilities up to $\varepsilon$ follows directly from the definition of a simulation.

Conversely, given the linear maps $\Psi$ and $\Phi$ of an $\varepsilon$-embedding, we obtain a univalent $\varepsilon$-simulation via ${\rm Sim}_{\mathcal{B}}(\omega_A):=\{\Psi(\omega_A)\}$ and ${\rm Sim}_{\mathcal{B}}(e_A):=\{\Phi(e_A)\}$.
\end{proof}
\subsection{Proof of Lemma~\ref{lem:Unital}}
\begin{proof}
Consider the maps $\Psi$ and $\Phi$ that embed $\mathcal{A}$ into $\mathcal{B}$. Let $u'_B:=\Phi(u_A)$, then $u'_B\in E_B$ and $\Delta:=u_B-u'_B\in E_B$. Due to the definition of $\varepsilon$-embedding, we have
\begin{align}
   |(\Psi(\omega_A),\Phi(u_A))-(\omega_A,u_A)|\leq\varepsilon,
\end{align}
and thus $(\Psi(\omega_A),u'_B)\in [1-\varepsilon,1]$ for all $\omega_A\in\Omega_A$. Since $(\Psi(\omega_A),u_B)=1$, we have $(\Psi(\omega_A),\Delta)\leq\varepsilon$ for all $\omega_A\in\Omega_A$. Now pick an arbitrary state $\xi_A\in\Omega_A$, and define
\begin{align}
   \Tilde\Phi(e_A):=\Phi(e_A)+(\xi_A,e_A)\Delta.
\end{align}
Suppose that $e_A\in E_A$, then $u_A-e_A\in E_A$, and so $\Phi(u_A-e_A)\in E_B$, i.e.\ $f_B:=u'_B-\Phi(e_A)\in E_B$, and so $E_B\ni u_B-f_B=\Phi(e_A)+\Delta$. Therefore, $0\leq \tilde\Phi(e_A)\leq u_B-f_B\leq u_B$, and so $\tilde\Phi(e_A)\in E_B$. 
Moreover, $\tilde\Phi(u_A)=u_B$, and the maps $\Psi$ and $\tilde\Phi$ satisfy
\begin{align}
& |(\omega,e)-(\Psi(\omega),\Tilde\Phi(e))|  \nonumber \\
& \leq |(\omega,e)-(\Psi(\omega),\Phi(e))|+|(\Psi(\omega),\Phi(e))-(\psi(\omega),\tilde\Phi(e))| \nonumber \\
& \leq \varepsilon + |e(\xi_A)(\Psi(\omega),\Delta)|\leq 2\varepsilon.
\end{align}
This proves the claim.
\end{proof}

\subsection{Proof of Lemma~\ref{lem:EmbeddingFacts}}
\begin{proof}
{\bf (i):} 
We have $(\omega,e)=(\Psi(\omega),\Phi(e))=(\omega,\Psi^* \Phi(e))$ for all $e\in E_A, \omega\in \Omega_A$. Since $E_A$ spans $A$, and since $\Omega_A$ spans $A^*$, we get $\Psi\conj\Phi=\mathbf{1}_A$.

{\bf (ii):} First, $P^2 = \Phi\Psi\conj\Phi\Psi\conj = \Phi\Psi\conj = P$. 
As $\Psi\conj\Phi \equiv \mathbf{1}_A$, we get $P \Phi\!\left(a\right) = \Phi \Psi\conj  \Phi \!\left(a\right) = \Phi\!\left(a\right)$ for all $a\in A$. 
This shows that $\Phi(A)\subseteq {\rm im}\, P$. Conversely, if $b\in {\rm im}\, P$ define $a:=\Psi^*(b)$, then $\Phi(a)=\Phi\Psi^*(b)=Pb=b$, i.e.\ $b\in\Phi(A)$.
To show unitality in the case of a unital embedding, apply $\Phi$ to both sides of $\Psi\conj\Phi\!\left(u_A\right) = u_A$ to yield $\Phi \Psi\conj \Phi\!\left(u_A\right) = \Phi\!\left(u_A\right)$,
  then use $\Phi\!\left(u_A\right) = u_B$ to conclude $\Phi \Psi\conj u_B = u_B$.
Similar reasoning establishes $P\conj$ as a projector onto $\Psi\!\left(A\conj\right)$.
\end{proof}

\subsection{Proof of Lemma~\ref{lem:UnrestrictedFacts}}
\begin{proof}
{\bf (i):} Let $b\in B_+$. Since $\Psi$ is positive, we have $(\Psi(a),b) \geq 0$
 and hence $(a,\Psi^*(b))\geq 0$ for all $a\in A\conj_+$. 
Since $\mathcal{A}$ is unrestricted, this implies that $\Psi\conj(b)\in A_+$, and hence $\Psi^*$ is positive. 
Similarly, let $\omega\in B_+^*$. Since $\Phi$ is positive, we have $(\omega,\Phi(a))\geq 0$ and thus $(\Phi^*(\omega),a)\geq 0$ for all $a\in A_+$.
Since $\mathcal{A}$ is unrestricted, it follows that $\Phi^*(\omega)\in A_+^*$, i.e.\ $\Phi^*$ is positive.
As the composition of positive maps is positive, $P$ and $P\conj$ are hence also positive.

{\bf (ii):}  
 $\Phi(A_+)\subseteq \Phi(A)\cap B_+$ follows from the positivity of $\Phi$, and if $b\in\Phi(A)\cap B_+=P(B)\cap B_+$ then $b=Pb\in P(B_+)$. 
 For the converse inclusions, we have $P(B_+)\subseteq B_+$ due to positivity of $P$ and $P(B_+)\subseteq P(B)=\Phi(A)$. 
 To see that $\Phi(A)\cap B_+\subseteq \Phi(A_+)$, let $a\in A$ and $\Phi(a)\in B_+$, then for all $\omega\in A_+\conj$, we have $(\omega,a)=(\Psi(\omega),\Phi(a))\geq 0$, hence $a\in A_+$.
\end{proof}

\subsection{Proof of Corollary~\ref{col:Classical}}
\begin{proof}
Suppose that $\mathcal{A}$ can be embedded into some $\mathcal{C}_n$ via maps $\Phi,\Psi$. 
Since $\mathcal{C}_n$ can be embedded into $\mathcal{Q}_n$ via some $\Phi',\Psi'$, this gives us an embedding of $\mathcal{A}$ into $\mathcal{Q}_n$ via $\Phi'\circ\Phi,\Psi'\circ\Psi$. 
Due to Theorem~\ref{thm:Simulate}, $\mathcal{A}$ must correspond to a special Euclidean Jordan algebra. 
But Lemma~\ref{lem:UnrestrictedFacts}(ii) tells us that $\Phi(A_+)=\Phi(A)\cap C_+$, 
 where $C_+$ is the polyhedral cone~\cite{AliprantisT07} of classical effects. 
Hence $A_+$ must be a polyhedral cone too, i.e.\ $A_+$ contains only a finite number of extremal rays. 
The only Jordan-algebraic effect cones with finitely many extremal rays are the classical ones.
\end{proof}

\section{Proof of Theorem~\ref{thm:Correspondence}}
\label{app:Correspondence}
\begin{proof}
Suppose that $\mathcal{A}$ is a GPT that is exactly simulated by $\mathcal{C}_n$. 
Let us construct a corresponding operational theory for which the simulation constitutes an ontological model. 
The preparation procedures $p$ are identified with the elements of ${\rm Sim}_{\mathcal{C}}(\omega_A)$ for all $\omega_A\in\Omega_A$. 
The effects of the operational theory are identified with the elements of ${\rm Sim}_{\mathcal{C}}(e_A)$ for all $e_A\in E_A$, and the measurement procedures correspond to finite collections of such effects that add up to unit probability on all preparations.

Let us construct an ontological model for this theory. 
The set of ontic states is defined as $\Lambda:=\{1,2,\ldots,n\}$. 
Now, given any preparation procedure $p$ and measurement procedure and outcome $(k,m)$, there is a corresponding classical state $\omega_C\in{\rm Sim}_{\mathcal{C}}(\omega_A)$ for some $\omega_A\in\Omega_A$ and a corresponding classical effect $e_C\in{\rm Sim}_{\mathcal{C}}(e_A)$ for some $e_A\in E_A$ such that $p=\omega_C$ and $(k,m)=e_C$, i.e.\ $P(k|p,m)=(\omega_C,e_C)=\sum_\lambda (\omega_C)_\lambda (e_C)_\lambda$ with $\sum_\lambda (\omega_C)_\lambda=1$ and $(e_C)_\lambda\in[0,1]$. 
We therefore identify preparation procedures with the distributions on ontic states that they generate, and outcomes of measurement procedures with the response function that they implement. 
Consider two preparation procedures $p=\omega_C\in{\rm Sim}_{\mathcal{C}}(\omega_A)$ and $p'=\omega'_C\in{\rm Sim}_{\mathcal{C}}(\omega'_A)$. 
They are equivalent if and only if
\begin{align*}
p\sim p' \hspace{-1.5em}& \\
&\Leftrightarrow (\omega_C,e_C)=(\omega'_C,e_C)\mbox{ for all }e_C\mbox{ in all }{\rm Sim}_{\mathcal{C}}(e_A)\\
&\Leftrightarrow (\omega_A,e_A)=(\omega'_A,e_A)\mbox{ for all }e_A\\
&\Leftrightarrow \omega_A=\omega'_A\\
&\Leftrightarrow p,p'\in{\rm Sim}_{\mathcal{C}}(\omega_A)\mbox{ for the same }\omega_A.
\end{align*}
That is, the sets ${\rm Sim}_{\mathcal{C}}(\omega_A)$ define the equivalence classes of preparations; and similarly, the sets ${\rm Sim}_{\mathcal{C}}(e_A)$ define the equivalence classes of outcome-measurement pairs. 
Thus, equivalent preparation procedures lead to identical distributions over ontic states if and only if $|{\rm Sim}_{\mathcal{C}}(\omega_A)|=1$ for all $\omega_A$ (preparation--univalence as per \cref{def:epssim}), 
and equivalent outcome-measurement pairs lead to identical response functions on ontic states if and only if $|{\rm Sim}_{\mathcal{C}}(e_A)|=1$ for all $e_A$ (measurement--univalence as per \cref{def:epssim}).

Conversely, suppose that we have an operational theory, the corresponding GPT $\mathcal{A}$, and a discrete ontological model of the operational theory. Let $n:=|\Lambda|$, and then rename the elements of $\Lambda$ such that $\Lambda=\{1,2,\ldots,n\}$. 
States $\omega_A\in\Omega_A$ are equivalence classes of preparation procedures; similarly, effects $e_A\in E_A$ are equivalence classes of outcome-measurement pairs. 
Thus, for all $\omega_A\in\Omega_A$ and all $e_A\in E_A$, we can define
\begin{align}
{\rm Sim}_{\mathcal{C}}(\omega_A)&:=\{\mu_p\,\,|\,\, p\in\omega_A\},\\
{\rm Sim}_{\mathcal{C}}(e_A)&:= \{\chi_{k,m}\,\,|\,\, (k,m)\in e_A\}.
\end{align}
These are nonempty sets, and $|{\rm Sim}_{\mathcal{C}}(\omega_A)|=1$ for all $\omega_A$ if and only if the ontological model is preparation--noncontextual; similarly, $|{\rm Sim}_{\mathcal{C}}(e_A)|=1$ for all $e_A$ if and only if the ontological model is measurement--noncontextual.
 It remains to show that these assignments define an exact simulation of $\mathcal{A}$ by $\mathcal{C}_n$. 
 To this end, let $\omega_C\in{\rm Sim}_{\mathcal{C}}(\omega_A)$ and $e_C\in {\rm Sim}_{\mathcal{C}}(e_A)$; it follows that $\omega_C=\mu_p$ for some $p\in\omega_A$, and $e_C=\chi_{k,m}$ for some $(k,m)\in e_A$. Thus
\begin{align}
   (\omega_C,e_C)=\sum_{\lambda\in\Lambda}\mu_p(\lambda)\chi_{k,m}(\lambda)=(\omega_A,e_A)
\end{align}
which proves Eq.~(\ref{eq:match_stats}). 
Furthermore, if $\mu_p\in{\rm Sim}_{\mathcal{C}}(\omega_A)$ and $\mu_q\in{\rm Sim}_{\mathcal{C}}(\varphi_A)$, then
\begin{align}
   \lambda\mu_p+(1-\lambda)\mu_q&=\mu_{\lambda p+(1-\lambda)q}\nonumber\\
&\in{\rm Sim}_{\mathcal{C}}(\lambda \omega_A+(1-\lambda)\varphi_A)
\end{align}
which implies Eq.~(\ref{eq:Mixing1}); an analogous argument applies to effects. 
Finally, by definition, $\chi\equiv 0$ is a valid response function, i.e.\ is actually being implemented for some (probability zero) outcome $k$ of some measurement $m$, i.e.\ $0=\chi_{k,m}\in {\rm Sim}_{\mathcal{C}}(0)$. 
We have thus verified all the properties of an exact simulation.
\end{proof}

\section{Detailed proof of Example~\ref{ExRebit}}
\label{SecProofRebit}
The proof sketch given in the main text works without further arguments if we assume that the $\varepsilon$-embedding is \emph{unital}, i.e.\ $\Phi(\mathbf{1})=(1,\ldots,1)\trans$. Let us first assume unitality and recapitulate the proof steps of Ref.~\cite{MazurekPKRS16} before dropping this assumption.

Our rebit states and effects satisfy
\begin{align}
\frac 1 2 \rho_{1,0}+\frac 1 2 \rho_{1,1}=\frac 1 2 \rho_{2,0}+\frac 1 2 \rho_{2,1}=\frac 1 2 \rho_{3,0}+\frac 1 2 \rho_{3,1}, &\\
\frac 1 3 E_{1,0}+\frac 1 3 E_{2,0} + \frac 1 3 E_{3,0}=\frac 1 2 \mathbf{1}=\frac 1 3 E_{1,1}+\frac 1 3 E_{2,1} + \frac 1 3 E_{3,1}, &\\
E_{t,0}+E_{t,1} = \mathbf{1}. &
\end{align}
Since the embedding maps $\Psi$ and $\Phi$ are linear,
\begin{align}
\frac 1 2 \omega_{1,0}+\frac 1 2 \omega_{1,1}=\frac 1 2 \omega_{2,0}+\frac 1 2 \omega_{2,1}=\frac 1 2 \omega_{3,0}+\frac 1 2 \omega_{3,1}=:\omega, &\\
\frac 1 3 e_{1,0}+\frac 1 3 e_{2,0}+\frac 1 3 e_{3,0} = \frac 1 2 \mathbf{1}=\frac 1 3 e_{1,1}+\frac 1 3 e_{2,1}+\frac 1 3 e_{3,1},&\\
e_{t,0}+e_{t,1}=\mathbf{1}&.
\end{align}
where we have used the notation $\mathbf{1}=(1,1,\ldots,1)\trans$ and unitality via $\Phi(\mathbf{1})=\mathbf{1}$. Following Ref.~\cite{MazurekPKRS16}, for any triple of numbers $\{a_t\}=(a_1,a_2,a_3)$, we define the function $f\{a_t\}:=\frac 1 3 \sum_{t=1}^3 \max\{a_t,1-a_t\}$, and it is easy to check that this function is convex, i.e.\ for $0\leq\mu\leq 1$ we have
\begin{align}
   f\{\mu a_t+(1-\mu)b_t\}\leq \mu\, f\{a_t\}+(1-\mu)f\{b_t\}.
\end{align}
Since states (and effects) like $\omega$ are vectors in $\mathbb{R}^n$, we will use the notation $\omega^\lambda$ for the components of that vector. We then have
\begin{align}
   e_{t,b}^\lambda\leq \eta_t^\lambda:=\max_{b'} e_{t,b'}^\lambda,
\end{align}
and we get the chain of inequalities
\begin{align}
A&:= \frac 1 6 \sum_{t,b} e_{t,b}\cdot \omega_{t,b} = \frac 1 6 \sum_{t,b,\lambda} e_{t,b}^\lambda \omega_{t,b}^\lambda
\leq \frac 1 3 \sum_{t,\lambda} \eta_t^\lambda \frac 1 2 \sum_b \omega_{t,b}^\lambda\nonumber\\
&= \frac 1 3 \sum_{t,\lambda} \eta_t^\lambda \omega^\lambda=\sum_\lambda \omega^{\lambda} \frac 1 3 \sum_t \eta_t^\lambda\leq \max_\lambda \frac 1 3 \sum_t \eta_t^\lambda.\label{eqMax}
\end{align}
So far, we have not used unitality, but we will do so now. Unitality implies that
\begin{align}
   \frac 1 3 \sum_t \eta_t^\lambda = \frac 1 3 \sum_t \max\{e_{t,0}^\lambda,1-e_{t,0}^\lambda\}=f\{e_{t,0}^\lambda\},
\end{align}
and that $\frac 1 3 \sum_t e_{t,0}^\lambda=\frac 1 2$. 
Hence
\begin{align}
   A\leq \max_\lambda f\{e_{t,0}^\lambda\}\leq \max_{\{a_t\}\in\mathcal{C}} f\{a_t\},
\end{align}
where $\mathcal{C}$ is the compact convex set of $(a_1,a_2,a_3)\trans$ with $0\leq a_t\leq 1$ and $\frac 1 3 \sum_t a_t=\frac 1 2$. Since $f$ is a convex function, its maximum is attained at one of the extremal points of $\mathcal{C}$ (like $(1,\frac 1 2,0)$), which implies $A\leq \frac 5 6$.

Let us now drop the assumption of unitality and go beyond the arguments of Ref.~\cite{MazurekPKRS16}. For the effects, we know that
\begin{align}
   \frac 1 3 \sum_t e_{t,0}^\lambda = \frac 1 3 \sum_t e_{t,1}^\lambda,\quad 0\leq e_{t,0}^\lambda, 0\leq e_{t,1}^\lambda,\,\, e_{t,0}^\lambda+e_{t,1}^\lambda \leq 1.
\end{align}
Set $\delta_t^\lambda:=1-(e_{t,0}^\lambda+e_{t,1}^\lambda)$, which is a non-negative number, then $\tilde e_{t,b}^\lambda:=e_{t,b}^\lambda+\frac 1 2 \delta_t^\lambda$ is non-negative, and
\begin{align}
   \tilde e_{t,0}^\lambda+\tilde e_{t,1}^\lambda=1,\quad \frac 1 3 \sum_t \tilde e_{t,0}^\lambda=\frac 1 3 \sum_t \tilde e_{t,1}^\lambda=\frac 1 2.
\end{align}
Since Eq.(\ref{eqMax}) still holds, we have
\begin{align}
   A&\leq \max_\lambda \frac 1 3 \sum_t \max\{e_{t,0}^\lambda,e_{t,1}^\lambda\}\leq \max_\lambda \frac 1 3 \sum_t \max\{\tilde e_{t,0}^\lambda,\tilde e_{t,1}^\lambda\}\nonumber\\
   &\leq \max_{\{a_t\}\in\mathcal{C}} f\{a_t\}=\frac 5 6.
\end{align}
Together with the arguments of the main text, this concludes the proof of Example~\ref{ExRebit}.

\clearpage
\section{Proof of embeddings \\into quantum theory}
\subsection{Spin-factor / $d$-balls} 
\label{app:spin}

\begin{applemma}
Consider the map $\Phi$ from $\reals\oplus \reals^d$ into $\Csa{2^m}$ where $m := d/2$ if $d$ is even or $m := \left(d+1\right)/2$ otherwise,
 such that (for $n\in\reals$, $\vec{x}\in \reals^d$)
\begin{align}
\Phi\left(n,\vec{x}\right) := n\id_{2^m} + \sum_{i=1}^d x_i \gamma_i,
\end{align}
where, for $i=1,\ldots,d$:
\begin{align}
\gamma_{i} := \begin{cases}
\underbrace{\sigma_z \ldots \sigma_z}_{\frac{i-1}{2}} \sigma_x \underbrace{\sigma_0 \ldots \sigma_0}_{m - \frac{i-1}{2}} 
& \mathrm{for~odd~} i \\
\underbrace{\sigma_z \ldots \sigma_z}_{\frac{i}{2}-1} \sigma_y \underbrace{\sigma_0 \ldots \sigma_0}_{m - \frac{i}{2}}
& \mathrm{for~even~} i,
\end{cases}
\end{align}
where $\sigma_x, \sigma_y, \sigma_z$ are the $2\times 2$ complex Pauli matrices, and $\sigma_0 = \id_2$.
Likewise, identically define the map $\Psi$ from spin-factor states to $\Csa{2^m}$.
Then:
\begin{enumerate}[(i)]
\item \label{it:ball:Lin} $\Phi$ and $\Psi$ are $\reals$-linear,
\item \label{it:ball:Unital} $\Phi$ is unital,
\item \label{it:ball:Pos} $\Phi$ and $\Psi$ are positive,
\item \label{it:ball:Preserve} $\Phi$ and $\Psi$ jointly preserve outcome probabilities,
\item \label{it:ball:Norm} $\Psi$ preserves the normalization of states.
\end{enumerate}
Hence, $\Phi$ and $\Psi$ define a unital embedding into quantum theory, as per \cref{def:EmbeddingMap}.
\begin{proof}
\Cref{it:ball:Lin} holds by inspection: the map $\Phi$ (and hence $\Psi$) only contains linear contributions of the elements of $(n,\vec{x}) \in \reals\oplus \reals^d$.
\Cref{it:ball:Unital} also follows straightforwardly, noting $\Phi\left(\vec{u}_d\right) = \Phi\left((1, \vec{0})\right) = \id_{2^m}$.

\Cref{it:ball:Pos}.  To verify positivity, we use the proof from \citet{Tsirelson87} via the hints in \citet{KleinmannOSW13}.
Take an element $\vec{v} := \left(n, \vec{x}\right) \in B_{d+}$, and consider the square
\begin{align}
\left[\Phi\left(\vec{v}\right)\right]^2 = n^2\, \id_{2^m} + 2 n \sum_{i=1}^d x_i\gamma_i + \sum_{j=1}^d\sum_{k=1}^d x_j x_k \gamma_j \gamma_k.
\end{align} 
Using $\gamma_j\gamma_k + \gamma_k\gamma_j = 2\delta_{jk}\id_{2^m}$ (as except when $j=k$, $\gamma_j$ and $\gamma_k$ anti-commute), then
\begin{align}
\left[\Phi\left(\vec{v}\right)\right]^2 
 = \left( n^2 + \sum_{i=1}^d {x_i}^2\right) \id_{2^m} + 2n \sum_{i=1}^d x_i \gamma_i.
\end{align}
Then, using 
 $2n\Phi\left(\vec{v}\right) = 2n^2\id_{2^m}+ 2 n\sum_{i=1}^d x_i \gamma_i$, 
 we have:
\begin{align}
\Phi\left(\vec{v}\right)
 = \frac{1}{2n} \left(\Phi\left(\vec{v}\right)\right)^2  + \frac{1}{2n} \left( n^2 - \sum_i {x_i}^2 \right) \id_{2^m}.
 \end{align}
Since $\Phi\left(\vec{v}\right)$ is self adjoint, $\Phi\left(\vec{v}\right)^2 = \Phi\left(\vec{v}\right) \Phi\left(\vec{v}\right)^\dag$ is always positive. 
Meanwhile,  since $n^2\geq \sum_{i=1}^d {x_i}^2$ by definition of $\vec{v}\in B_{d+}$, the second term is also positive.
The sum of two positive matrices is itself positive, and hence $\Phi: B_{d+} \to \Csap{2^m}$. 
The same argument can be made for map $\Psi$.

\Cref{it:ball:Preserve}. To demonstrate outcome probabilities are preserved, consider an effect $e:=(n_e, \vec{e}) \in B_{d+}$ and a state $s=(n_s, \vec{s})\in B_{d+}$.
The direct outcome probability of this pair is:
\begin{align}
\label{eq:ball:ProbDirect}
\inn{s}{e} = n_e n_s + \sum_{i=1}^d e_i s_i.
\end{align}
Compare this with the corresponding quantum outcome probability associated with $\Phi(e)$ and $\Psi(s)$:
\begin{align}
\tr\left(\Phi(e)\Psi(s)\right) & = \tr\left(\left(n_e\id + \sum_{i=1}^d e_i \gamma_i\right) \left(n_s\id + \sum_{j=1}^d s_j \gamma_j\right)\right) \nonumber \\
& = \tr\left(n_e n_s\id +  n_s \sum_{i=1}^d e_i \gamma_i + n_e\sum_{j=1}^d s_j \gamma_j \right. \nonumber \\
& \left.\hspace{4em} + \sum_{i=1}^d\sum_{j=1}^d e_i s_j \gamma_i \gamma_j\right) \nonumber \\
& = n_e n_s + \sum_{i=1}^d\sum_{j=1}^d e_i s_j \delta_{ij} = n_e n_s + \sum_{i=1}^d e_i s_i,
\end{align}
where we use the linearity of the trace, that Pauli matrices are traceless, and that $\tr\left(\gamma_i \gamma_j\right) = \delta_{ij}$.
This matches \cref{eq:ball:ProbDirect}, such that the maps $\Phi$ and $\Psi$ hence preserve outcome probabilities.
\Cref{it:ball:Norm}.
Finally, to see that $\Psi$ preserves the normalization of states, note that $\tr\Phi(1,\vec x)=1$.
\end{proof}
\end{applemma}

\subsection{Quaternionic quantum theory} 
\label{app:quat}
We denote by $\imhead$, $\jmhead$ and $\kmhead$ the taking of the real-value of the quaternionic elements $i$, $j$ and $k$ respectively, such that together with the real part $\rehead$, $Q = \re{Q} + i \im{Q} + j \jm{Q} + k \km{Q}$.
\begin{applemma}
\label{thm:FwdQ}
Let
 $\Phi$ be a map from $Q\in\Hsa{n}$ to $M\in\comp^{2n\times 2n}$ 
such that with  $A := \re{Q} + i\im{Q}$ and $B := \jm{Q} + i\km{Q}$.
\begin{equation}
\label{eq:SympFormApp}
\Phi(Q):= M= \left(\begin{array}{cc}A & B \\ -B\conj & A\conj \end{array}\right),
\end{equation}
and similarly define $\Psi := \frac{1}{2} \Phi$ to act on quaternionic quantum states (also elements of $\Hsa{n}$).
Then \begin{enumerate}[(i)]
\item \label{it:quat:Lin} $\Phi$ and $\Psi$ are $\reals$-linear,
\item \label{it:quat:SA} $\Phi\left(Q\right)\in\Csa{2n}$ and $\Psi\left(Q\right)\in\Csa{2n}$,
\item \label{it:quat:Pos} $\Phi$ and $\Psi$ are positive,
\item \label{it:quat:Unital} $\Phi$ is unital,
\item \label{it:quat:Preserve} $\Phi$ and $\Psi$ jointly preserve outcome probabilities.
\item \label{it:quat:norm} $\Psi$ preserves the normalization of states.
\end{enumerate}
Hence, $\Phi$ and $\Psi$ are a unital embedding into complex quantum theory, as per \cref{def:EmbeddingMap}.
\begin{proof}
\Cref{it:quat:Lin} follows by noting that the functions $\rehead$, $\imhead$, $\jmhead$, $\kmhead$ are real-valued and $\reals$-linear; so a map onto a matrix whose elements are given by linear combinations of these functions (i.e.\ $\Phi$ or $\Psi$) will also be $\reals$-linear.

To see \cref{it:quat:SA}, note that if $Q\in\Hsa{n}$, then $Q\ct := \re{Q}\trans-i\im{Q}\trans-j\jm{Q}\trans-k\km{Q}\trans = \re{Q} + i\im{Q}+j\jm{Q} + k \km{Q} = Q$. 
Comparing coefficients, $\re{Q}=\re{Q\trans}$, whereas $\im{Q} = -\im{Q}\trans$, $\jm{Q}\trans = -\jm{Q}$, and $\km{Q}\trans = -\km{Q}$.
Thus, consider $A=\re{Q} + i\im{Q}$; $A\ct = \re{Q}\trans - i\im{Q}\trans = \re{Q}+i\im{Q} = A$. Likewise $A\conj = (A\conj)\ct$.
Meanwhile, $B=\jm{Q}+i\km{Q}$,
 such that $B\trans = \jm{Q}\trans + i\km{Q}\trans = -\jm{Q} - i\km{Q}\trans = -B$.
This implies that $-B\conj = \left(B\trans)\right)\conj = B\ct$.
Thus, $M\in\Csa{2n}$.
Similar reasoning applies for $\Psi$.

\Cref{it:quat:Pos} follows by noting that the spectrum of $M$ is the same as $Q$, but each eigenvalue has twice the multiplicity~\cite{LeoS00}.
It then follows that if all eigenvalues of $Q$ are nonnegative, so too will be the eigenvalues of $M$, and $\Phi$ is hence a positive map.
The same logic applies for $\Psi$ (since the factor $\frac{1}{2}$ is positive). 
This also proves \cref{it:quat:norm}, since
\begin{align}
   \tr(\Psi(E))=\sum_{i=1}^{2n}\lambda_i(\Psi(E))=\frac{1}{2} \cdot 2 \cdot \sum_{i=1}^n \lambda_i(E)=\tr E,
\end{align}
where $\lambda_i(E)$ labels the eigenvalues of $E$.

\Cref{it:quat:Unital} holds since for $Q=\id_n$, $A=A\conj=\id_n$ and $B=-B\conj=0$, and hence $M =\Phi(Q) =\id_{2n}$.

Finally, to show \cref{it:quat:Preserve}, consider $Q, E \in\Hsap{n}$, and write $Q = A_Q + B_Q j$ and $E = A_E + B_E j$ for complex matrices $A_Q, B_Q, A_E, B_E$, as per \cref{eq:SympFormApp}.
The quaternionic quantum probabilities are given by 
\begin{align}
(Q, E) & := \re{\tr\left(Q E\right)} \nonumber \\
&= \re{\tr\left(\left(A_Q + B_Q j\right)\left(A_E + B_E j\right)\right)} \nonumber \\
&= \re{\tr\left(A_Q A_E + B_Q j A_E + A_Q B_E j + B_Q j B_E j\right)} \nonumber \\
&= \re{\tr\left(A_Q A_E + B_Q {A_E}\conj j + A_Q B_E j - B_Q B_E\conj\right)} \nonumber \\
&= \tr\left(\re{A_Q A_E - B_Q B_E\conj}\right),
\label{eq:RePartOfQTrace}
\end{align}
(where we use $ij = -ji$ so that $A j = j A\conj$, $j j = -1$, and that the middle expressions in $j$ have no real components).
Now, consider $\Phi(Q) \Psi(E)$: 
\begin{align}
\frac{1}{2}
\left(\begin{array}{cc}
A_Q & B_Q \\
-{B_Q}\conj & {A_Q}\conj
\end{array}\right)
\left(\begin{array}{cc}
A_E & B_E \\
-{B_E}\conj & {A_E}\conj
\end{array}\right) &
\nonumber \\
& \hspace{-14em} = 
\frac{1}{2}
\left(\begin{array}{cc}
A_Q A_E - B_Q {B_E}\conj & A_Q B_E + B_Q {A_E}\conj \\
-{B_Q}\conj A_E - {A_Q}\conj {B_E}\conj & -{B_Q}\conj B_E + {A_Q}\conj {A_E}\conj
\end{array}\right),
\end{align}
such that the quantum trace,
\begin{align}
\tr\left(\Phi(Q) \Psi(E)\right) \hspace{-3em} & \nonumber \\
& = \frac{1}{2} \tr\left( A_Q A_E - B_Q {B_E}\conj -{B_Q}\conj B_E + {A_Q}\conj {A_E}\conj\right) \nonumber \\
& = \tr\left(\re{ A_Q A_E - B_Q {B_E}\conj}\right),
\end{align}
matches \cref{eq:RePartOfQTrace}, and the outcome probabilities are thus jointly preserved.
\end{proof}
\end{applemma}

\section{Proof of results on approximate embeddings}
\subsection{Proof of Lemma~\ref{lem:Eprops}}
\label{SecProoflem:Eprops}
\begin{proof}
First, note that every separable state $\omega_{AB}$ can be written as a convex combination of product states, hence $\mathcal{R}(\omega_{AB})=0$. 
Let us now prove (iii). 
If $\omega_{AB}$ is separable, then (iii) is trivially true, so suppose that $\omega_{AB}$ is entangled.
Let $E$ be the set of all $\varepsilon\geq 0$ such that there exist separable states $\varphi^S_{AB}$ and $\psi^S_{AB}$ with $\omega_{AB}=(1+\varepsilon)\varphi_{AB}^S-\varepsilon \psi^S_{AB}$. 
Let us first show that $E$ is topologically closed. 
Consider any sequence $\{\varepsilon_n\}_{n\in\mathbb{N}}$ with $\varepsilon_n\in E$ which converges to some real number $\varepsilon$.
Denote the corresponding separable states by $\varphi_n$ and $\psi_n$. 
Since the set of separable states is compact, there exists a subsequence $\{n_k\}_{k\in\mathbb{N}}$ and two separable states $\varphi,\psi$ with $\varphi_{n_k}\stackrel{k\to\infty}\longrightarrow\varphi$ and $\psi_{n_k}\stackrel{k\to\infty}\longrightarrow\psi$. 
But then, $\omega_{AB}=(1+\varepsilon)\varphi-\varepsilon\psi$, i.e.\ $\varepsilon\in E$.

Consider some choice of $\{\lambda_i\}_{i=1,\ldots,n}$ in~\cref{eq:DefE}, and write $\{1,\ldots,n\}=I\cup I^C$, where $I=\{i\,\,|\,\,\lambda_i\geq 0\}$ and $I^C=\{i\,\,|\,\,\lambda_i<0\}$. 
We have
\begin{align}
   \omega_{AB}=
   (1+\varepsilon)\sum_{i\in I} \frac{\lambda_i}{1+\varepsilon}\varphi_A^{(i)}\otimes\psi_B^{(i)}-\varepsilon\sum_{i\in I^C}\frac{|\lambda_i|}{\varepsilon}\varphi_A^{(i)}\otimes\psi_B^{(i)},
\end{align}
where $\varepsilon:=\sum_{i\in I^C} |\lambda_i|$ cannot be zero because otherwise $\omega_{AB}$ would be separable. 
This shows that $\varepsilon\in E$; in particular, $E$ is not empty and has (since it is closed) a minimum $\min E$. 
Noting that $\sum_{i=1}^n|\lambda_i|-1=2\varepsilon\geq 2\min E$, this inequality remains true after taking the infimum over all decompositions as in~\cref{eq:DefE}, i.e.\ $2\,\mathcal{R}(\omega_{AB})\geq 2\min E$. 
On the other hand, let $\varepsilon:=\min E$, then there exist separable states $\varphi_{AB}^S$ and $\psi_{AB}^S$ with $\omega_{AB}=(1+\varepsilon)\varphi_{AB}^S-\varepsilon\psi_{AB}^S$. 
Every separable state has a finite convex decomposition into product states, $\varphi_{AB}^S=\sum_{i=1}^m \mu_i \varphi_A^{(i)}\otimes\psi_B^{(j)}$ and $\psi_{AB}^S=\sum_{i=m+1}^n \nu_i \varphi_A^{(i)}\otimes\psi_B^{(j)}$. 
Set
\begin{align}
   \lambda_i:=\left\{
      \begin{array}{cl}
      	   (1+\varepsilon)\mu_i & \mbox{if }1\leq i \leq m\\
      	   -\varepsilon\nu_i & \mbox{if }m<i\leq n,
      \end{array}
   \right.
\end{align}
then we obtain a decomposition as in~\cref{eq:DefE} (in particular $\sum_{i=1}^n\lambda_i=1$), and $2\,\mathcal{R}(\omega_{AB})\leq\sum_{i=1}^n|\lambda_i|-1=2\varepsilon=2\min E$. 
This proves (iii), and it also shows that the infimum in~\cref{eq:DefE} is a minimum. This also implies the remainder of claim (i): if the optimal decomposition involves mixed states $\varphi_A^{(i)}$ or $\psi_B^{(i)}$, then we can convexly decompose those into pure states and obtain another decomposition into an in general larger number of pure product states, which is however still optimal, i.e.\ gives the same $\sum_i |\lambda_i|$. Moreover, this immediately implies (ii).

Finally, to prove (iv), consider two states $\omega_{AB}$ and $\omega_{AB}'$ with optimal decompositions
\begin{align}
\omega_{AB}&=\sum_{i=1}^m\lambda_i \varphi_A^{(i)}\otimes\psi_B^{(i)},\quad 2\,\mathcal{R}(\omega_{AB})=\sum_{i=1}^m |\lambda_i|-1,\nonumber\\
\omega'_{AB}&=\sum_{i=m+1}^n \lambda'_i \varphi_A^{(i)}\otimes\psi_B^{(i)},\quad 2\,\mathcal{R}(\omega'_{AB})=\sum_{i=m+1}^n |\lambda'_i|-1,
\end{align}
and $\sum_{i=1}^m\lambda_i=\sum_{i=m+1}^n \lambda'_i=1$. Let $\mu\in[0,1]$ and $\varphi_{AB}:=\mu \omega_{AB}+(1-\mu)\omega'_{AB}$. Define
\begin{align}
   \kappa_i:=\left\{
      \begin{array}{cl}
      	  \mu\lambda_i & \mbox{if }1\leq i\leq m\\
      	  (1-\mu)\lambda'_i & \mbox{if }m<i\leq n.
      \end{array}
   \right.
\end{align}
Then $\varphi_{AB}=\sum_{i=1}^n \kappa_i \varphi_A^{(i)}\otimes\psi_B^{(i)}$ and $\sum_{i=1}^n \kappa_i=1$. Thus,
\begin{align}
2\,\mathcal{R}(\varphi_{AB})&\leq \sum_{i=1}^n |\kappa_i|-1\nonumber\\
&= \mu\sum_{i=1}^m |\lambda_i|+(1-\mu)\sum_{i=m+1}^n |\lambda'_i|-1\nonumber\\
&= 2\mu \mathcal{R}(\omega_{AB})+2(1-\mu)\mathcal{R}(\omega'_{AB}).
\end{align}
Finally, to see that $\mathcal{R}$ is continuous, note that the set $S$ of separable states is a compact convex set of full dimension $\dim A \dim B-1$ in the affine space of linear functionals $\omega_{AB}$ which satisfy $\omega_{AB}(u_A,u_B)=1$. Introduce an arbitrary Euclidean inner product on the full linear space of such functionals, and use the resulting norm on the affine subspace of normalized functionals. In this subspace, there exists some $r>0$ such that $S$ contains a full ball $B$ of radius $r$. Suppose we have two bipartite states $\omega_{AB},\omega'_{AB}$ with $\delta:=\|\omega_{AB}-\omega'_{AB}\|$. In the ball $B$, we can find two separable states $\tau_{AB}^S$ (the center) and $\kappa_{AB}^S$ (a suitable point on the surrounding sphere) such that $(\tau_{AB}^S-\kappa_{AB}^S)/r=(\omega'_{AB}-\omega_{AB})/\delta$. Consider the optimal decomposition of $\omega_{AB}$ as in (iii) of this lemma, i.e.\ $\mathcal{R}(\omega_{AB})=\varepsilon$ and $(1+\varepsilon)\varphi_{AB}^S-\varepsilon\psi_{AB}^S=\omega_{AB}$. Then
\begin{align}
\omega'_{AB}&=\frac \delta r(\tau_{AB}^S-\kappa_{AB}^S)+(1+\varepsilon)\varphi_{AB}^S-\varepsilon \psi_{AB}^S\nonumber\\
&=\left(1+\varepsilon+\frac\delta r\right)\left(\frac{1+\varepsilon}{1+\varepsilon+\frac\delta r}\varphi_{AB}^S+\frac{\frac \delta r}{1+\varepsilon+\frac\delta r}\tau_{AB}^S\right)\nonumber\\
& \quad- \left(\varepsilon+\frac\delta r\right)\left(\frac\varepsilon {\varepsilon+\frac\delta r}\psi_{AB}^S+\frac{\frac\delta r}{\varepsilon+\frac\delta r}\kappa_{AB}^S\right).
\end{align}
Thus, $\mathcal{R}(\omega'_{AB})\leq\mathcal{R}(\omega_{AB})+\frac{\delta} r$. Exchanging the roles of $\omega_{AB}$ and $\omega'_{AB}$ above yields the same inequality, but with $\omega_{AB}$ and $\omega'_{AB}$ exchanged, and therefore $|\mathcal{R}(\omega_{AB})-\mathcal{R}(\omega'_{AB})|\leq \|\omega_{AB}-\omega'_{AB}\|/r$.
\end{proof}

\subsection{Proof of Lemma~\ref{lem:BellTables}}
\label{app:BellTableProof}
\begin{proof}
Set $B=B'=\mathcal{Q}_n$. 
Let $\Phi$ and $\Psi$ be the maps that define the $\varepsilon$-embedding of $\mathcal{A}$ into $\mathcal{Q}_n$. Let $e_B\in E_B$. Since $\Psi$ is positive, we have for all $\omega_A\in A_+^*$
\begin{align}
   0\leq (\Psi(\omega_A),e_B)=(\omega_A,\Psi^*(e_B)).
\end{align}
The special case $e_B=u_B$ yields $\Psi^*(u_B)=u_A$, since $\Psi$ maps normalized states of $A$ to normalized states of $B$. For general $e_B$, if $\mathcal{A}$ is unrestricted, the above implies that $\Psi^*(e_B)\in E_B$ and hence $\Psi^*$ is a positive map -- but $\mathcal{A}$ may not be unrestricted. However, we obtain that $\Psi^*:B\to A$ maps elements of $B_+$ onto vectors that are \emph{nonnegative on all states of $A$}: in other words, $\Psi^*(e_B)\in \bar E_B$, using the notation of Definition~\ref{def:BipartiteStates}. Define $\rho_{BB'}:B\times B'\to\mathbb{R}$ via
\begin{align}
   \rho_{BB'}(e^B,f^{B'}):=\omega_{A A'}(\Psi^*(e^B),\Psi^*(f^{B'})),
\end{align}
where $\omega_{AA'}$ is the bipartite state that generates the behaviour $P(a,a'|x,x')=\omega_{AA'}(e_{a|x},f_{a'|x'})$.
This bilinear map is unital, i.e.\ $\rho_{BB'}(u_B,u_B)=1$, and positive by construction. Hence, due to the results of \Citet{BarnumBBEW10} (see also \Citet{KleinmannOSW13}), there exists a density operator $\sigma_{BB'}$ on $B\otimes B'$ and a positive unital automorphism $\tau$ on $B'$ such that
\begin{align}
   \rho_{BB'}(e^B,f^{B'})={\rm tr}(\sigma_{BB'} e^B\otimes\tau(f^{B'})) \mbox{ for all }e^B,f^B\in B_+.
\end{align}
Set $e_{a|x}^B:=\Phi(e_{a|x})$ and $f_{a'|x'}^{B'}:=\Phi(f_{a'|x'})$. Since the embedding is unital, we have $\sum_a e_{a|x}^B=\Phi(u_A)=u_B$ and $\sum_{a'} f_{a'|x'}^{B'}=\Phi(u_A)=u_{B'}$, hence for every $x$ and $x'$, this defines valid local quantum measurements with POVM elements summing to the identity. Define the quantum behaviour
\begin{align}
   P_Q(a,a'|x,x'):={\rm tr}(\sigma_{BB'} e_{a|x}^B\otimes \tau(f_{a'|x'}^{B'})).
\end{align}
Now choose $\lambda_i$, $\varphi_A^{(i)}$, $\psi_{A'}^{(i)}$ as those that minimize the expression in~(\ref{eq:DefE}) for $\omega_{A  A'}$, i.e.\ $\sum_i\lambda_i=1$, $\sum_i \lambda_i\varphi_A^{(i)}\otimes\psi_{A'}^{(i)}=\omega_{ A A'}$, and $\sum_i |\lambda_i|=2\,\mathcal{R}(\omega_{ A A'})$. Then
\begin{align}
P_Q(a,a'|x,x')&=\rho_{BB'}(e_{a|x}^B,f_{a'|x'}^{B'})\nonumber\\
&=\omega_{A A'}(\Psi^* \Phi(e_{a|x}),\Psi^*\Phi(f_{a'|x'}))\nonumber\\
&=\sum_i \lambda_i (\varphi_A^{(i)},\Psi^*\Phi(e_{a|x}))(\psi_{A'}^{(i)},\Psi^*\Phi(f_{a'|x'}))\nonumber\\
&=\sum_i \lambda_i (\Psi(\varphi_A^{(i)}),\Phi(e_{a|x}))(\Psi(\psi_A^{(i)}),\Phi(f_{a'|x'})).
\end{align}
On the other hand, we have
\begin{align}
   P(a,a'|x,x')=\sum_i \lambda_i (\varphi_A^{(i)},e_{a|x})(\psi_A^{(i)},f_{a'|x'}).
\end{align}
Since we are dealing with an $\varepsilon$-embedding, we have $|(\omega_A,e_A)-(\Psi(\omega_A),\Phi(e_A))|\leq \varepsilon$, and hence $P$ and $P_Q$ are close to each other. Concretely,
\begin{align}
\Delta&:=|P(a,a'|x,x')-P_Q(a,a'|x,x')|\nonumber\\
&\leq\sum_i |\lambda_i| \left| (\varphi_A^{(i)},e_{a|x})(\psi_A^{(i)},f_{a'|x'})-\right.\nonumber\\
& \qquad\quad\, \left. -(\Psi(\varphi_A^{(i)}),\Phi(e_{a|x}))(\Psi(\psi_A^{(i)},\Phi(f_{a'|x'}))\right|\nonumber\\
&\leq  \sum_i |\lambda_i|\cdot 2\varepsilon=(1+2\,\mathcal{R}(\omega_{ A  A'}))2\varepsilon\nonumber\\
&\leq [1+2\,\mathcal{R}(\mathcal{ A})]2\varepsilon.
\end{align}
This proves the claim.
\end{proof}

\section{Why Theorem~\ref{thm:QuantumFiniteEps} cannot be straightforwardly generalized}
\label{sec:NoGen}

One might conjecture that Theorem~\ref{thm:QuantumFiniteEps} and Theorem~\ref{thm:ClassicalFiniteEps} have a couple of immediate generalizations. 
First, instead of embedding into classical or quantum theory, we might consider embeddings into some other GPT $\mathcal{B}$.
Generalizing Theorem~\ref{thm:QuantumFiniteEps}: if
\begin{align}
   \varepsilon<\frac{B_{\mathcal{AA}}-B_{\mathcal{BB}}}{4|B|(1+2\,\mathcal{R}(\mathcal{A}))}
   \label{eq:TheoremGeneralized}
\end{align}
then no $\varepsilon$-embedding of $\mathcal{A}$ into $\mathcal{B}$ is possible.
This statement is mathematically correct and can be used to bound the approximate embeddability of $\mathcal{A}$ into $\mathcal{B}$. However, in doing so, one has to be careful to remember the definition of $B_{\mathcal{BB}}$: it is the maximal value of $B$ over \emph{all non-signalling correlations} on two copies of $\mathcal{B}$.
More specifically, the quantity $B_{\mathcal{BB}}$ is defined as
\begin{align}
   B_{\mathcal{BB}}=\max_{\omega_{BB},e_{a|x},f_{b|y}}\sum_{a,b,x,y} b_{a,b,x,y}\,\omega_{BB}(e_{a|x},f_{b|y}),
\end{align}
where the maximization is over all bipartite states $\omega_{BB}$ in the sense of \cref{def:BipartiteStates}. 
This set of bipartite states is \emph{by definition} the largest possible set of normalized bilinear functionals that are nonnegative on all pairs of effects --- a set of states corresponding to what is known as the \emph{maximal tensor product} of $\mathcal{B}$ with itself.

In some cases, however, one is interested in studying embeddings into \emph{process theories}~\cite{CoeckeK17} that consist of more than a single stand-alone GPT $\mathcal{B}$ (or pair $\mathcal{BB}$). 
Process theories typically admit a variety of systems that can be composed in various ways. 
The process theory's composition rule might then give us a natural set of states on $\mathcal{BB}$ that is smaller than that of the maximal tensor product.
Thus, the process theory's actual maximal Bell violation on system pairs $\mathcal{BB}$ (say, $B^{\rm actual}_{\mathcal{BB}}$) -- the value that we would then most naturally study -- might be strictly smaller than the value that appears in \cref{eq:TheoremGeneralized}, i.e.\ $B^{\rm actual}_{\mathcal{BB}}<B_{\mathcal{BB}}$. 
But if we replace $B_{\mathcal{BB}}$ by $B^{\rm actual}_{\mathcal{BB}}$ in \cref{eq:TheoremGeneralized}, the corresponding modified version of Theorem~\ref{thm:QuantumFiniteEps} will be wrong in general.

Indeed, quantum theory is a process theory where systems combine with a tensor product that contains strictly less states than the maximal tensor product. 
It just so happens that in quantum theory, every correlation arising from a hypothetical ``pseudo quantum state'' in the maximal tensor product of two quantum systems can be exactly reproduced with an actual quantum state in the usual tensor product, a result that has been shown in Ref.~\cite{BarnumBBEW10}. 
This is why, for $\mathcal{B}$ a quantum system, the quantity $B_{\mathcal{BB}}$ is identical to the maximal Bell functional value achieved by valid physical bipartite quantum states, which is the reason why we can apply existing results on Bell inequalities to obtain the bound on embeddability in Theorem~\ref{thm:QuantumFiniteEps}.

Meanwhile, suppose $\mathcal{B}$ is an ``$m$ measurement $k$ outcome'' gbit.
The bipartite state space $\mathcal{BB}$ (according to our \cref{def:BipartiteStates}) indeed aligns with the process theory of ``boxworld'' (e.g.\ including PR-boxes for the $(m\!=\!2,k\!=\!2)$-gbit) that one is typically interested in~\cite{Barrett07,GrossMCD10}.
However, for any Bell functional involving at most $m$ measurements with at most $k$ outcomes each, for any $\mathcal{A}$, $\mathcal{AA}$ cannot have stronger correlations than $\mathcal{BB}$ (i.e.\ in this sense, boxworld is maximally nonlocal), and the numerator of \cref{eq:TheoremGeneralized} will hence never be positive.
In such a case, the generalized theorem implies nothing interesting.
It remains an open question whether, for a broader family of Bell functionals (with more than $m$ measurements, or more than $k$ outcomes),  some test of gbit--embeddability can be derived.

The technique of Theorem~\ref{thm:QuantumFiniteEps} also cannot be applied, for example, to \emph{triples} of quantum systems.
It has been shown that the correlations obtained from the maximal tensor product are in this case strictly more general than those allowed by the standard quantum tensor product~\cite{Acin10}.
Thus, there is no immediate generalization of Theorem~\ref{thm:QuantumFiniteEps} to three or more parties.

On the other hand, the result on embeddability into classical GPT $\mathcal{C}_n$ (Theorem~\ref{thm:ClassicalFiniteEps}) \emph{can} be extended to more than two parties: it only relies on the fact that the maximal tensor product of classical systems is a classical multipartite system.

Extending an argument by Barrett, \citet{Pusey18} has shown that in the simple scenario of four preparations and two tomographically complete binary measurements, the existence of a (classical) noncontextual ontological model is \emph{``equivalent to the existence of a Bell local model in the scenario considered by CHSH''}. 
Thus, Pusey's result also allows the certification of contextuality via nonlocality, but with a different construction where the local measurements and preparations are distributed among two parties.
 This is different from our approach, where the local system is mathematically duplicated.
 Moreover, our result applies to general GPTs and embeddings into classical \emph{and quantum} theory, and it is not clear whether Barrett's and Pusey's results can be similarly generalized (though it would be very interesting to explore this possibility).

\end{document}